\documentclass[accepted=2017-06-08,a4paper,onecolumn]{quantumarticle}

\usepackage[numbers,sort&compress]{natbib}

\usepackage[utf8]{inputenc}
\usepackage[english]{babel}
\usepackage[T1]{fontenc}

\usepackage[colorlinks = true, linkcolor = blue, urlcolor  = blue, citecolor = red]{hyperref} 

\usepackage{amsmath}
\usepackage{amssymb}
\usepackage{amsthm}
\usepackage{amsfonts}
\usepackage{longtable}

\usepackage{nomencl}

\DeclareMathOperator{\tr}{tr}

\newcommand{\proj}[1]{\left|#1\middle\rangle\!\middle\langle #1\right|}
\newcommand{\bra}[1]{\left\langle #1\right|}
\newcommand{\ket}[1]{\left|#1\right\rangle}

\newcommand{\eps}{\varepsilon}
\newcommand{\cX}{\mathcal{X}}
\newcommand{\cY}{\mathcal{Y}}

\newcommand{\cM}{\mathcal{M}}
\newcommand{\cS}{\mathcal{S}}
\newcommand{\cE}{\mathcal{E}}
\newcommand{\cSsub}{\mathcal{S}_{\bullet}} 
\newcommand{\cH}{\mathcal{H}}
\newcommand{\cN}{\mathcal{N}}

\newcommand{\rhot}{\tilde{\rho}}
\newcommand{\id}{\mathrm{id}}

\newcommand{\fail}{\varnothing}
\renewcommand{\succ}{\checkmark}
\newcommand{\qkdeb}{\mbox{\textnormal{\texttt{qkd\_eb$_{m,\mathrm{pe},\mathrm{ec},\mathrm{pa}}$}}}}
\newcommand{\qkdideal}{\mbox{\textnormal{\texttt{qkd\_ideal$_{m,\mathrm{pe},\mathrm{ec},\mathrm{pa}}$}}}}
\newcommand{\qkdpm}{\mbox{\textnormal{\texttt{qkd\_pm$_{M,m ,\mathrm{pe},\mathrm{ec},\mathrm{pa}}$}}}}
\newcommand{\qkdmod}{\mbox{\textnormal{\texttt{qkd\_modified$_{m,\mathrm{pe},\mathrm{ec},\mathrm{pa}}$}}}}

\newtheorem{lemma}{Lemma}

\newtheorem{theorem}[lemma]{Theorem}
\newtheorem{proposition}[lemma]{Proposition}
\newtheorem{corollary}[lemma]{Corollary}
\theoremstyle{definition}

\newtheorem{definition}{Definition}

\newcommand{\nc}[2]{#1 & #2\\}

\begin{document}

\title{A largely self-contained and complete security proof for quantum key distribution}

\author{Marco Tomamichel}
\affiliation{Centre for Quantum Software and Information, University of Technology Sydney, Australia}
\orcid{0000-0001-5410-3329}
\email{marco.tomamichel@uts.edu.au}
\homepage{http://www.marcotom.info}

\author{Anthony Leverrier}
\affiliation{Inria Paris, France}
\orcid{0000-0002-6707-1458}
\email{anthony.leverrier@inria.fr}
\homepage{https://who.paris.inria.fr/Anthony.Leverrier}
\maketitle

\begin{abstract}
In this work we present a security analysis for quantum key distribution, establishing a rigorous tradeoff between various protocol and security parameters for a class of entanglement-based and prepare-and-measure protocols. The goal of this paper is twofold: 1) to review and clarify the state-of-the-art security analysis based on entropic uncertainty relations, and 2) to provide an accessible resource for researchers interested in a security analysis of quantum cryptographic protocols that takes into account finite resource effects. For this purpose we collect and clarify several arguments spread in the literature on the subject with the goal of making this treatment largely self-contained.

More precisely, we focus on a class of prepare-and-measure protocols based on the Bennett-Brassard (BB84) protocol as well as a class of entanglement-based protocols similar to the Bennett-Brassard-Mermin (BBM92) protocol. We carefully formalize the different steps in these protocols, including randomization, measurement, parameter estimation, error correction and privacy amplification, allowing us to be mathematically precise throughout the security analysis. We start from an operational definition of what it means for a quantum key distribution protocol to be secure and derive simple conditions that serve as sufficient condition for secrecy and correctness. We then derive and eventually discuss tradeoff relations between the block length of the classical computation, the noise tolerance, the secret key length and the security parameters for our protocols. Our results significantly improve upon previously reported tradeoffs.
\end{abstract}

\section{Introduction}

Quantum key distribution (QKD) is a cryptographic task that allows two distant parties, Alice and Bob, to exchange secret keys and communicate securely over an untrusted quantum channel, provided they have access to an authenticated classical channel. 
The first QKD protocol, BB84, was proposed by \citet{bb84} more than three decades ago and the last 30 years have witnessed staggering experimental advances, making QKD the first quantum information technology. With the advent of quantum information theory, \citet{ekert91} offered a fruitful new perspective on quantum key distribution by casting it in terms of quantum entanglement and Bell nonlocality and it was quickly noted that the original BB84 protocol can be seen in this light as well~\cite{bennett92}. This new perspective was particularly useful for the development of formal security proofs of QKD.

Formalizing the intuitive security arguments accompanying the first protocols has proven to be challenging. Early proofs by \citet{Lo99}, \citet{SP00}, and \citet{mayers01} successfully attacked the problem in the asymptotic limit of infinitely many exchanged quantum signals (and unbounded classical computing power). A later work by \citet{koashi06} first brought to light that security can be certified using an entropic form~\cite{maassen88} of Heisenberg's uncertainty principle~\cite{heisenberg27}. However, these works all lacked a convincing treatment of the security tradeoff in a more realistic regime where the number of exchanged signals and the classical computing power (i.e., the length of the bit strings computations are performed on) are necessarily finite, and the final secret key is thus also of finite length. As absolute security is no longer feasible, the first and most crucial question arising in this finite regime is how to properly define approximate security of a cryptographic protocol. Following developments in classical cryptography, \citet{renner05} extended the concept of composable security to quantum key distribution and established a first security proof for finite key lengths. This security analysis essentially established a tradeoff between different parameters of a quantum key distribution protocol:

\begin{description}
  \item[Block length:] During the run of the protocol the two honest parties, typically called Alice and Bob, prepare, send and measure quantum signals and store the results of these measurements in bit strings that they store on their respective (classical) computers. These bit strings will eventually be used to check for the presence of an interfering eavesdropper and to compute a secret key. As these strings get longer it gets easier to guard against eavesdroppers and secret keys can be extracted more efficiently. On the other hand there are limitations on the length because we would like to start producing a secret key as early as possible\footnote{This is most relevant for implementations that suffer from a low rate of measurement events, e.g.\ entanglement-based implementations.} and because computations on longer strings get more and more difficult. The length of the bit strings used in the protocol is called the \emph{block length}.
  Current experimental and commercial implementations of quantum key distribution typically work with block lengths of the order $10^5$--$10^6$ whereas block lengths of the order $10^7$--$10^8$ can be achieved, but require an extreme stability of the system during the several hours required to collect the data \cite{comandar16}, \cite{jouguet13}.
  
  \item[Key length:] The \emph{secret key length} is the length of the secret key (in bits) that is extracted from a single block of measurement data. An ideal secret key is a uniformly random bit string perfectly correlated between Alice and Bob and independent of any information the eavesdropper might have collected after the run of the protocol. The ratio of the key length to the block length is a key performance indicator of quantum key distribution systems.

  \item[Security parameters:] Modern security definitions for quantum key distribution rely on approximate indistinguishability from an ideal protocol, which ensures that the resulting key can be safely used in any other (secure) application. In this case, the ideal protocol either has the two parties produce an ideal secret key or an abort flag that indicates that no secret key can be extracted, either because an eavesdropper is present, or more mundanely because the quantum channel is too noisy. The \emph{security parameter} is the distinguishing advantage between the real and the ideal protocols, given by the diamond norm distance between the two protocols. Evidently we would like this parameter to be very small and while there is no consensus on what value it should have we will take it to be $10^{-10}$ for our numerical examples.

  \item[Robustness:] According to the notion of security discussed above, a quantum key distribution protocol can be perfectly secure and completely useless because it always aborts. As an additional requirement we thus impose that the protocol succeeds with high probability when the quantum channel is subject to noise below a specific (and realistic) threshold. This describes the \emph{robustness} of the protocol against noise, that is, the probability that the protocol returns a nontrivial key for a given noise level. The noise model used should capture the dynamics of the quantum channel in the expected field operation; however, the exact specification of the noise model\,---\,and whether the noise is caused by an eavesdropper or just by the undisturbed operation of the channel\,---\,is independent of any security considerations and can thus be treated independently. The robustness, and more specifically the values of the channel parameters for which the robustness goes to zero, is an important figure of merit to compare the expected performance of various protocols. 
\end{description}

The tradeoffs between these parameters have been significantly improved since Renner's proof~\cite{renner05}, in particular by \citet{tomamichellim11} and \citet{hayashi11}, so that the proofs are now sufficiently tight to provide security for realistic implementations of quantum key distribution. 
The present analysis will mostly follow the approach in the former paper~\cite{tomamichellim11}.

So what justifies us revisiting this problem here?
Firstly, we believe that presenting a complete and rigorous security proof in a single article will make the topic of finite size security more accessible to researchers in quantum cryptography. Secondly, thanks to some improvements in the technical derivation and a steamlining of the analysis, our proof yields significantly stronger tradeoff relations between security and performance parameters. It is worth noting here that strengthening theoretical tradeoff relations of a QKD protocol has rather direct implications for practical implementations as it allows for the generation of more secure key at the same noise level without any changes to the hardware.
Thirdly, although all the necessary technical ingredients and conceptual insights are present in the literature, we were not able to find a concise security proof for any QKD protocol that satisfies the following two stringent criteria:
\begin{enumerate}
  \item The protocol is able to extract a composably secure key for reasonable parameters (i.e.\ realistic noise levels, security parameters and block lengths that can be handled with state-of-the-art computer hardware).
  \item The protocol and all the assumptions on the physical devices used in the protocol are completely specified and all aspects of the protocol are formalized, including the randomness that is required and all the communication transcripts that are produced. 
\end{enumerate}
The second point may appear trivial\,---\,but we believe the absence of a complete formalization of all aspects of a protocol found in many research papers presents a major obstacle in verifying the proofs and learning about the security of quantum key distribution and quantum cryptography in general. 
It is in fact common in much of the present literature to fully formalize some aspects of a security proof while keeping other aspects vague and informal\,---\,and this has lead to various misconceptions.

Let us illuminate this issue with an example. It is often argued (e.g.\ in~\cite{scarani08}) that collective attacks (where the eavesdropper attacks every quantum signal exchanged between Alice and Bob in the same way) are optimal for the eavesdropper using symmetry and de-Finetti arguments~\cite{renner07, christandl09}. To get such symmetry it is at some point or another used that measurements are performed in a random basis or that a random subset of raw key bits are used for parameter estimation. However, complete security proofs also must allow for the protocol to abort in case certain thresholds are not met, and one is then left to analyze the state of the system conditioned on the fact that the protocol does not abort.
This conditioning will in general introduce correlations between the state held by Alice and Bob and the seeds used to choose the measurement bases and parameter estimation subset, violating the strict symmetry assumptions. Even if these correlations are weak in typical cases, they prove problematic since they allow the eavesdropper to slightly influence Alice and Bob's measurement devices\footnote{An example is the following attack in the prepare-and-measure version of BB84: for each state sent by Alice, Eve randomly chooses to either measure it in the computational basis, or do nothing. Such an attack would only create errors when Bob measures in the Hadamard basis, and the conditioning on not aborting would therefore slightly favor measurement choices for Bob biased toward the computational basis.}, something which is usually explicitly forbidden in security proofs for quantum key distribution with trusted devices. Hence, many simple arguments based on symmetry or independent randomness simply do not go through without modification when the security proof is put under a microscope.

As mentioned above, the early asymptotic proofs fail with Point~1. Moreover, while Renner's analysis~\cite{renner05} gives bounds for finite keys, these are not sufficient to pass Point~1  since the bounds are not strong enough for realistic key lengths.\footnote{Unfortunately, de Finetti reductions often do not provide good bounds in practice and at least $10^5$ or $10^6$ uses of the quantum channel are typically required for the key rate to effectively become nonzero~\cite{scarani08,christandl09,sheridan10}.} 
More recent security proofs by \citet{tomamichellim11} and \citet{hayashi11} satisfy Point~1, but they are not fully formalized and thus do not satisfy Point~2.\footnote{For example, a small flaw in the formalization of the protocol in~\cite{tomamichellim11} has recently been pointed out by Pfister \emph{et al.}~\cite{pfister15}. While this does not suggest that the security guarantees in~\cite{tomamichellim11} are invalid, it is a stark reminder that Point~2 is often not taken seriously enough.} In fact, our requirements in Point~2 are very stringent and we are not aware of any security proof that has met this level of rigor, except arguably Renner's thesis~\cite{renner05}.  A recent security proof for the one-sided device-independent setting~\cite{tomamichelfehr12} satisfies Point 2 but provides a key rate that is not optimal asymptotically.

\paragraph{Results and Outline.}

In the present paper, we give a rigorous and largely self-contained security proof for QKD that satisfies the two conditions above. The proof is based on the security analysis in~\cite{tomamichellim11} and uses an entropic uncertainty relation~\cite{tomamichel11} as its main ingredient. A few additional technical results and modifications of previous results are needed. We hope that our proof is accessible to all researchers interested in the security of quantum cryptography. As such, our treatment does not require the reader to have prior knowledge of various tricks and security reductions in quantum cryptography, but presumes a solid understanding of the mathematical foundations of quantum information theory. 

The remainder of this manuscript is structured as follows.
First we introduce necessary notation and concepts in Section~\ref{sc:notation}.
In Part~\ref{part:eb} we analyze a class of simple entanglement-based protocols reminiscent of the BBM92 protocol~\cite{bennett92}. Section~\ref{sc:eb1} formally describes the class of protocols we are using (see also Table~\ref{tb:simple}). Section~\ref{sc:eb2} formally introduces our security definitions and claims. We discuss our results in Section \ref{sc:eb2.5} and provide a detailed security proof in Section~\ref{sc:eb3}.
In Part~\ref{part:pm}  we move on to a prepare-and-measure protocol that is essentially equivalent to BB84.
Section~\ref{sc:pm1} formally introduces this class of protocols (see also Table~\ref{tb:realistic}). We discuss our results in Section~\ref{sec:results-pm} and provide in Section~\ref{sec:proofpm} the security reduction from the prepare-and-measure protocol to the entanglement-based protocol.


\section{Formalism and notation}
\label{sc:notation}

We will summarize some concepts necessary for our formal security proofs here, assuming that the reader is familiar with the mathematical foundations of quantum information theory.
We refer to~\cite{mybook} for a comprehensive introduction into this mathematical toolkit. Sections~\ref{sec:not1}--\ref{sec:not4} are necessary for understanding our main exposition whereas the concepts introduced in Sections~\ref{sec:not5} will be employed only in the security proof.

\newpage
\subsection{Quantum systems, states and metrics}
\label{sec:not1}

Individual \emph{quantum systems} and the corresponding finite-dimensional Hilbert spaces are denoted by capital letters. The dimension of the system $A$ is denoted by $|A|$. A \emph{joint quantum system} $AB$ is defined via the tensor product of the corresponding Hilbert spaces of $A$ and $B$.  We use $[m]$ to denote the set $\{1, 2, \ldots, m\}$ and use $A_{[m]}$ to denote a joint quantum system comprised of quantum systems $A_1 A_2 \ldots A_m$. Similarly, if the subscript is a subset of $[m]$, we just refer to the subsystems in the subset. Let us also introduce the notation $\Pi_{m,k} := \{ \pi \subset [m] : |\pi| = k \}$, the set of subsets of $[m]$ of size $k$.

We write $\cS(A)$ to denote normalized states on $A$, i.e., positive semi-definite operators acting on the Hilbert space $A$ with unit trace. A state is called pure if it the corresponding operator has rank $1$. We will employ the trace distance between states, which is defined as
\begin{align}
\| \rho - \sigma \|_{\tr} :=  \sup_{P} \tr \big\{ P (\rho - \sigma) \big\} \, ,
\end{align}
where $P$ ranges over projectors, i.e.\ positive semi-definite operators with eigenvalues in $\{0, 1\}$.
In particular, we have $\| \rho - \sigma \|_{\tr} = \frac{1}{2}\| \rho - \sigma \|_{1}$, where $\|\cdot\|_1$ denotes the Schatten 1-norm. The trace distance has an immediate physical interpretation~\cite{helstrom76}: for two states with trace distance $\eps$, the maximum probability of distinguishing them with a single measurement equals $\frac12 (1 + \eps)$.

We also collect positive semi-definite operators with trace norm not exceeding 1 on $A$ in the set $\cSsub(A)$, and call them sub-normalized states.\footnote{The disk in the subscript of $\cSsub$ symbolizes the unit disk in trace norm.} Sub-normalized states will be very convenient for technical reasons as they allow us to represent the state of quantum systems and classical events simultaneously.
The following metric is very useful when dealing with sub-normalized states:
\begin{definition}[Purified distance] \label{df:pd} For $\rho_A, \sigma_A \in \cSsub(A)$, we define the \emph{generalized fidelity},
\begin{align}
  F(\rho_A, \sigma_A) &:= \bigg( \tr \Big\{ \sqrt{ \sqrt{\rho_A} \sigma_A \sqrt{\rho_A} } \Big\} + \sqrt{ 1-\tr\{\rho_A\}} \sqrt{1-\tr\{\sigma_A\}}  \bigg)^2 \,,
\end{align}
and the \emph{purified distance}, $P(\rho_A, \sigma_A) := \sqrt{1-F(\rho_A,\sigma_A)}$.
\end{definition}
The purified distance is a metric on sub-normalized states and satisfies~\cite[Lemma~2]{tomamichel09}
\begin{align}
  P(\rho_A, \sigma_A) \geq P\big(\mathcal{F}(\rho_A), \mathcal{F}(\sigma_A)\big) \label{eq:pd-dpi}
\end{align}
for every completely positive (CP) trace non-increasing map $\mathcal{F}$. 
This means in particular that the distance contracts when we apply a quantum channel to both states. 
An important property of the purified distance~\cite[Corollary~3.1]{mybook} is that for any two states $\rho_A, \sigma_A$ and any extension $\rho_{AB}$ of $\rho_A$, there exists an extension $\sigma_{AB}$ with $P(\rho_{AB}, \sigma_{AB}) = P(\rho_A, \sigma_B)$. (This property is not true in general for the trace distance.)
Moreover, it is related to the trace distance as follows~\cite[Lemma~6]{tomamichel09}:
\begin{align}
  \big\| \rho_A - \sigma_A \big\|_{\tr} + \frac12 \big| \tr\{\rho_A\} - \tr\{\sigma_A\} \big| \leq P(\rho_A, \sigma_A) \,.
\end{align}

\subsection{Classical registers and events}
\label{subsection:2.2}

We model discrete random variables by finite-dimensional quantum systems, called \emph{registers},  with a fixed orthonormal basis. For example, let $X \in \cX$ be a random variable with probability law $x \mapsto P_X(x)$. Then we write the corresponding quantum state as
\begin{align}
  \rho_X = \sum_{x \in \cX} P_X(x) \proj{x}_X , 
\end{align}
where $\{ \ket{x} \}_{x \in \cX}$ is an orthonormal basis of the space $X$. Conversely, we write
$\Pr[X = x]_{\rho} = P_X(x)$. 

More generally, the classical register might be correlated with a quantum system $A$, and this is modeled using \emph{classical-quantum (cq)} states:
\begin{align}
  \rho_{XA} = \sum_{x \in \cX} P_X(x) \proj{x}_X \otimes \rho_{A|X=x} \,,
\end{align}
where we use $\rho_{A|X=x}$ to denote the quantum state on $A$ conditioned on the register $X$ taking the value $x$.
We also write $\Pr[X=x]_{\rho} = \tr\{ \proj{x}_X \rho_{XA} \} = P_X(x)$. This convention is extended to arbitrary events defined on a classical register $X$, i.e.\ if $\Omega: \cX \to \{0, 1\}$ is an \emph{event}, we write
\begin{align}
  \Pr[\Omega]_{\rho} = \sum_{x \in \cX} P_X(x) \Omega(x) \quad \textrm{and} \quad 
  \rho_{XA \land \Omega} = \sum_{x \in \cX} P_X(x) \Omega(x) \proj{x}_X \otimes \rho_{A|X=x} \,,
\end{align}
a state that is generally sub-normalized. We will also write ${\rho}_{XA | \Omega} = \Pr[\Omega]_{\rho}^{-1} \rho_{XA \land \Omega}$ for the conditional state. For any event $\Omega: \mathcal{X} \to \{ 0, 1\}$ we denote its complement on $\cX$ by $\lnot \Omega$.

\subsection{Quantum channels and measurements}

A \emph{quantum channel} $\mathcal{E}: A \to B$ is a completely positive trace-preserving (CPTP) map that maps operators on $A$ to operators on $B$. Prime examples of quantum channels are the trace, denoted $\tr$, and the partial trace over system $A$, denoted $\tr_A$. We will encounter the diamond distance between CPTP maps, which we here define as
\begin{align}
  \| \mathcal{E} - \mathcal{F} \|_{\diamond} := \sup_{\rho_{AC} \in \cS(AC)} \| \mathcal{E}(\rho_{AC}) - \mathcal{F}(\rho_{AC}) \|_{\tr} \,,
\end{align}
where the optimization goes over joint states on $A$ and an auxiliary system $C$, and we can assume without loss of generality that $|C| \leq |A|$. The diamond distance also inherits the physical interpretation of the trace distance: for two quantum channels with diamond distance $\eps$, the maximum probability of distinguishing them by preparing a state on the input system and an ancilla system and then measuring the joint system after applying the channel equals $\frac12 (1 + \eps)$.

A \emph{generalized measurement} on $A$ is a set of linear operators $\{ M_A^x \}_{x \in \cX}$ such that 
\begin{align}
    \sum_{x \in \cX} (M_A^x)^{\dagger} (M_A^x) = 1_A \,,  \label{eq:measex}
\end{align}
where $1_A$ denotes the identity operator on $A$. A measurement on $A$ can be represented as a CPTP map $\mathcal{M}_{A \to X}$ that maps states on a quantum system $A$ to measurement outcomes stored in a register $X$. The measurement in~\eqref{eq:measex} applied to a bipartite state $\rho_{AB}$ yields
\begin{align}
  \mathcal{M}_{A \to X} : \rho_{AB} \mapsto \sigma_{XB} = \sum_{x \in \cX} \proj{x}_X \otimes  \tr_A\Big\{ M_A^x\, \rho_{AB} \big( M_A^{x} \big)^{\dag} \Big\} ,
\end{align}
where $\sigma_{XB}$ is now a (normalized) classical-quantum state.
Finally, let $f: \cX \to \cY$ be a function acting on two sets $\cX$ and $\cY$. We denote by $\mathcal{E}_f: X \to XY$ the corresponding CPTP map
\begin{align}
  \mathcal{E}_f [ \cdot ] = \sum_{x \in \cX} \ket{f(x)}_{Y} \proj{x}_X \cdot \proj{x}_X \bra{f(x)}_{Y} \,
\end{align}
that acts on general quantum states.
Note that we defined the map $\mathcal{E}_f$ such that the input register $X$ is kept intact and the operation is deterministic and invertible.

\subsection{Universal hash functions}
\label{sec:not4}
Universal hashing is used (at least\footnote{Universal hashing is also used to provide authentication of the classical channel, but we will not discuss this issue here.}) twice in the analysis of the quantum key distribution protocol: first in the error correction step to ensure the correctness of the protocol (Theorem~\ref{th:correct}), and then in the privacy amplification procedure to guarantee the secrecy of the final key. 
\begin{definition}[Universal$_2$ Hashing] 
Let $\mathcal{H} =\{h\}$ be a family of functions from $\mathcal{X}$ to $\mathcal{Z}$. The family $\mathcal{H}$ is said to be \emph{universal$_2$} if $\Pr\,[H(x) = H(x') ] = \frac{1}{|\mathcal{Z}|}$ for any pair of distinct elements $x, x' \in \mathcal{X}$, when $H$ is chosen uniformly at random in $\mathcal{H}$.
\end{definition}
In this work we do not need to specify any particular family of hash functions, and it suffices to note that such families of functions always exist if $|\mathcal{X}|$ and $|\mathcal{Z}|$ are powers of $2$. (See, e.g.,~\cite{carter79,wegman81}.)

\subsection{Conditional entropies}
\label{sec:not5}

Conditional entropies measure the amount of uncertainty present in a random variable from the perspective of an observer with access to correlated side information. Here we are particularly interested in observers that have access to a quantum system that serves as side information, for example the eavesdropper's memory after interfering with the quantum communication during the run of a quantum key distribution protocol. The most common measure of entropy is the Shannon or von Neumann entropy, defined as $H(X)_{\rho} := -\sum_{x \in \cX} P_X(x) \log P_X(x)$. However, while this entropy has various operational interpretations in the asymptotic limit of infinite repetitions of an information processing task, it is insufficient to describe finite size effects. On the other hand, smooth min- and max-entropy allow us to capture such finite size effects and share many properties with the von Neumann entropy. We will not need the full generality of the smooth entropy formalism here and instead refer to~\cite{mybook} for a comprehensive introduction. 

Min- and max-entropy are natural generalizations of conditional R\'enyi entropies~\cite{renyi61} to the quantum setting and were first proposed by~\citet{renner05} and~\citet{koenig08}. The conditional min-entropy captures how difficult it is for an observer with access to quantum side information to guess the content of a classical register.
   For a bipartite cq state $\rho_{XB} \in \cS(AB)$, we define
   \begin{align}
     p_{\rm guess}(X|B)_{\rho} = \sup_{ \{ E_B^x \}}\ \sum_{x \in \cX} \Pr[X = x]_{\rho} \tr \Big\{ E_B^x \, \rho_{B|X=x} \big(E_B^x\big)^\dag\Big\}, \label{eq:gp}
   \end{align}
   where the optimization goes over all generalized measurements on $B$.

The conditional min-entropy for a cq state is then defined as $H_{\min}(X|B)_{\rho} := - \log p_{\rm guess}(X|B)_{\rho}$. 
For later convenience we introduce the measure more generally for any bipartite, potentially sub-normalized, state:
\begin{definition}[Min-entropy] For any bipartite sub-normalized state $\rho_{AB} \in \cSsub(AB)$, we define
\begin{align}
  H_{\min}(A|B)_{\rho} &:=  \sup \big\{ \lambda \in \mathbb{R} : \exists\, \sigma_B \in \cS(B) \textrm{ such that } \rho_{AB} \leq 2^{-\lambda} \id_A \otimes \sigma_B ,  \big\} \, .
\end{align}
\end{definition}
Showing equivalence between this definition and the special case of cq states in~\eqref{eq:gp} involves semidefinite-programming duality~\cite{koenig08} and is outside the scope of this work. We will also encounter the max-entropy, which is a natural dual of the min-entropy in the following sense:
\begin{definition}[Max-entropy] For any bipartite sub-normalized state $\rho_{AB} \in \cSsub(AB)$, we define 
\begin{align}
H_{\max}(A|B)_{\rho} := -H_{\min}(A|C)_{\rho} \,,
\end{align}
where $\rho_{ABC}$ is any pure state with $\tr_C \{ \rho_{ABC} \} = \rho_{AB}$.
\end{definition}

The max-entropy is a measure of the size of the support of $X$. In particular, we have the following ordering of unconditional entropies:
\begin{align}
  H_{\min}(X)_{\rho} \leq H(X)_{\rho} \leq H_{\max}(X)_{\rho} \leq \log \big| \{ x \in X : \Pr[X = x] _\rho> 0 \} \big| ,
\end{align}
which is a consequence of the monotonicity of the R\'enyi entropies~\cite{renyi61} in the order parameter. Here and throughout this article $\log$ denotes the binary logarithm.

We will need a slight generalization of the concepts of conditional min- and max-entropy, which takes into account a ball of states close to $\rho_{AB}$ in terms of the purified distance introduced in the previous section.
\begin{definition}[Smooth Entropies]
  For $\rho_{AB} \in \cSsub(AB)$ and $\eps \in \Big[0, \sqrt{\tr(\rho_{AB})} \Big)$, we define
  \begin{align}
    H_{\min}^{\eps}(A|B)_{\rho} := \sup_{ \rhot_{AB} \in \cSsub(AB) , \atop  P(\rhot_{AB}, \rho_{AB}) \leq \eps } H_{\min}(A|B)_{\rhot}, \qquad
    H_{\max}^{\eps}(A|B)_{\rho} := \inf_{ \rhot_{AB} \in \cSsub(AB) , \atop  P(\rhot_{AB}, \rho_{AB}) \leq \eps } H_{\max}(A|B)_{\rhot} \,.
  \end{align}
\end{definition}
In the above definitions we can replace the supremum and infimum with a maximum and minimum, respectively.
Roughly speaking, the smooth conditional min-entropy of $X$ given $B$ approximates how much randomness that is uniform for an observer with access to $B$ can be extracted from $X$. (This will be made formal when discussing the Leftover Hashing Lemma in Section~\ref{sc:hash}.)
The smooth entropies inherit the duality relation~\cite{tomamichel09}. For any pure sub-normalized state $\rho_{ABC} \in \cSsub(ABC)$, we have
\begin{align}
  H_{\min}^{\eps}(A|B)_{\rho} = - H_{\max}^{\eps}(A|C)_{\rho} \,. \label{eq:dual}
\end{align}

The smooth entropies also satisfy a data-processing inequality (DPI)~\cite[Theorem~18]{tomamichel09}.
  For any cq state $\rho_{XB}$ and any completely positive trace-preserving map $\mathcal{E}_{B \to C}$, we have
  \begin{align}
     H_{\min}^{\eps}(X|B)_{\rho} \leq H_{\min}^{\eps}(X|C)_{\mathcal{E}(\rho)}, \quad \textrm{and} \quad
     H_{\max}^{\eps}(X|B)_{\rho} \leq H_{\max}^{\eps}(X|C)_{\mathcal{E}(\rho)} \,. \label{eq:dpi}
  \end{align}
This expresses our intuition that performing any processing of the side information can at most increase our uncertainty about $X$.
Moreover, we need a simple chain rule~\cite[Lemma~11]{winkler11}, which states that
\begin{align}
  H_{\min}^{\eps}(A|BX)_{\rho} \geq H_{\min}^{\eps}(A|B)_{\rho} - \log |X| \label{eq:min-chain-rule}
\end{align}
where $X$ is a (classical) register of dimension $|X|$. This corroborates our intuition that an additional bit of side information on $X$ cannot decrease our uncertainty about $X$ by more than one bit.

We have defined all these quantities for sub-normalized states so that we can easily treat restrictions to events. Let $\rho_{AXBY} \in \cS(ABXY)$ be classical on $X$ and $Y$ and let $\Omega: \cX \times \cY \to \{0, 1\}$ be an event. Then we denote by $H_{\min}(AX \land \Omega |BY)_{\rho}$ the conditional min-entropy evaluated for the state $\rho_{AXBY \land \Omega}$. Similarly, $H_{\min}(AX \land \Omega |B)_{\rho}$ denotes the conditional min-entropy evaluated for the marginal $\rho_{AXB \land \Omega} = \tr_Y \{ \rho_{AXBY \land \Omega} \}$. These states are in general sub-normalized. The same notational convention is used for (smooth) min- and max-entropy.

\part{Entanglement-based protocol}
\label{part:eb}

\begin{table}[h!]
{\footnotesize
\begin{tabular}{ll}
\nc{$\fail$, $\succ$}{Symbols for abort and passing, respectively}
\nc{$M_{A_i}^{\phi, x}$}{Measurement operator acting on register $A_i$ with setting $\phi$ and outcome $x$}
\nc{$c_i$}{Parameter quantifying the quality of the measurement on register $A_i$}
\nc{$\bar{c}$}{Parameter quantifying the overall (average) quality of measurements on $A$}
\hline
\nc{$m$}{Total number of quantum systems shared and measured by Alice and Bob}
\nc{pe}{Parameter estimation scheme: $\mathrm{pe} = \{ k, \delta \}$}
\nc{$k$}{Length (in bits) of the raw key used for parameter estimation}
\nc{$n = m - k$}{Length (in bits) of the raw key used for key distillation}
\nc{$\delta$}{Threshold for the parameter estimation test}
\nc{test}{Test function used in the parameter estimation step}
\nc{ec}{Error correcting scheme: $\mathrm{ec} = \{ t, r, \mathrm{synd}, \mathrm{corr},\cH_{\mathrm{ec}} \}$}
\nc{$r$}{Length (in bits) of the error correction syndrome}
\nc{$t$}{Length (in bits) of the hash used for verification in the error correcting scheme}
\nc{synd}{Function computing the error syndrome}
\nc{corr}{Function that calculate the corrected string}
\nc{$\mathcal{H}_{\mathrm{ec}}$}{Universal$_2$ family of hash functions used in the error correcting scheme}
\nc{pa}{Privacy amplification scheme: $\mathrm{pa} = \{ \ell, \cH_{\mathrm{pa}} \}$}
\nc{$\ell$}{Length (in bits) of the final key}
\nc{$\mathcal{H}_{\mathrm{pa}}$}{Universal$_2$ family of hash functions used in the privacy amplification scheme}

\hline
\nc{$A$, $B$}{Alice's and Bob's initial quantum system}
\nc{$E$}{Eve's quantum memory}
\nc{$\mathcal{M}_{A \to X |S}$}{Measurement map applied on register $A$ with setting $S$ and storing the result in register $X$}
\nc{$V, W$}{Register for Alice's and Bob's classical bits used for parameter estimation}
\nc{$X, Y$}{Register for Alice's and Bob's classical bits used for key distillation}
\nc{$Z$}{Register for Alice's syndrome during error correction}
\nc{$T$}{Register containing the hash of Alice's raw key during error correction}
\nc{$K_A$, $K_B$}{Register for Alice's and Bob's final keys}
\nc{$S^\Phi$}{Seed for the choice of the measurement bases in the idealized protocol}
\nc{$S^\Pi$}{Seed for the choice of the random subset $\pi \in \Pi_{m,k}$ used for parameter estimation}
\nc{$S^\Xi$}{Seed for the choice of the measurement bases for the subsystems used for parameter estimation}
\nc{$S^\Theta$}{Seed for the choice of the measurement bases for the subsystems used for key distillation}
\nc{$S^{H_{\mathrm{ec}}}$}{Seed for the choice of the hash function used in the error correction test}
\nc{$S^{H_{\mathrm{pa}}}$}{Seed for the choice of the hash function used in the privacy amplification step}
\nc{$S$}{Register corresponding to all the seeds, $S = (S^\Phi, S^\Pi, S^\Xi, S^\Theta, S^{H_{\mathrm{pe}}}, S^{H_{\mathrm{ec}}}, S^{H_{\mathrm{pa}}})$.}
\nc{$F^{\mathrm{pe}}$}{Flag for the parameter estimation test}
\nc{$F^{\mathrm{ec}}$}{Flag for the error correction test}
\nc{$F$}{Register corresponding to all the flags, $F = (F^{\mathrm{pe}},F^{\mathrm{ec}})$}
\nc{$C^V$}{Transcript of the register $V$ sent during parameter estimation}
\nc{$C^Z, C^T$}{Transcripts of the registers $Z$ and $T$ sent during error correction}
\nc{$C$}{Register containing all the communication transcripts, $C = (C^V, C^Z, C^T)$}

\hline
\nc{$\rho$}{Quantum state before any measurement took place}
\nc{$\tau$}{Quantum state after the registers used for parameter estimation have been measured}
\nc{$\sigma$}{Quantum state once Alice and Bob's quantum registers have been entirely measured}
\nc{$\omega$}{Quantum state describing the final output of the protocol}
\end{tabular}
}
\caption{Overview of the nomenclature and notation used in Part~\ref{part:eb}.}
\label{tb:not-eb}
\end{table}

\section{Formal description of the entanglement-based protocol}
\label{sc:protocol}
\label{sc:eb1}

We first focus on class of simple entanglement-based QKD protocols.
We give an overview of the protocols in Table~\ref{tb:simple}. Section~\ref{sec:assumpt} discusses the assumptions that go into our model, Section~\ref{sc:params} presents the protocol parameters, and the detailed mathematical description of the individual steps follows in Section~\ref{sec:model}. 

Let us emphasize that by \emph{simple} protocols, we mean that we restrict our attention to protocols where the \emph{sifting} procedure is essentially given for free, meaning that Alice and Bob are assumed to initially share a quantum state on which all the measurements are performed. Relatedly, we also do not allow for strategies where the measurement settings are biased towards a specific value.\footnote{Such strategies are advocated in the literature in order to increase the secret key rate by minimizing the cost of sifting~\cite{lo04}.} As discussed in Part \ref{part:pm}, the sifting procedure can be analyzed separately under certain assumptions.

\subsection{Assumptions of our model}
\label{sec:assumpt}
\label{sc:assumeb}

Every mathematical model of physical reality requires some assumptions, and in cryptography it is important to discuss these assumptions since if they are not met by an implementation then the security guarantees derived here are also not applicable to this implementation.

\paragraph*{Finite-dimensional quantum systems:}

We assume that Alice's and Bob's relevant quantum degrees of freedom can be effectively represented on a finite-dimensional Hilbert space. (This requirement is not strictly necessary to show security but allows us to circumvent some technical pitfalls.)

\paragraph*{Sealed laboratories:}

We assume that the laboratories of Alice and Bob are spatially separated. This allows us to model joint quantum systems $AB$ shared between Alice and Bob as tensor products of respective local Hilbert spaces $A$ and $B$. Moreover, an easily overlooked (an in practice hard to ensure) assumption we need is that we control exactly what information is released from Alice and Bob's laboratory.

\paragraph*{Random seeds:}

We assume that Alice has access to uniform randomness (uniformly random seeds). In practice, the seeds can be produced by a trusted quantum random number generator in Alice's lab.\footnote{See, for instance, ~\cite{frauchiger13} for an analysis of realistic quantum random number generators explaining how randomness originating from quantum processes can be turned into ideal seeds.}

\paragraph{Authenticated communication channel:}
We assume that Alice and Bob share an authenticated public (classical) communication channel.
Everything that is communicated over this channel will be in the public domain and is thus treated as an output of the protocol. The authentication of the classical channel can be obtained with information-theoretic security by tagging every classical message~\cite{wegman81}. A more detailed discussion of authentication for QKD is beyond the scope of the present work, and the interested reader is referred to~\citet[Appendix D]{portmann14}.

\paragraph*{Deterministic detection:}

We further assume that Alice and Bob's measurement devices always output a valid outcome, either $0$ or $1$. This is unrealistic in practice since it is often the case that detectors will not detect the quantum system (due to losses or imperfect detection efficiency). A simple fix is then to flip a coin and use the resulting bit as the measurement output. Unfortunately, this solution artificially decreases the robustness of the protocol beyond what is usually tolerable in a practical setting. Another much more practical solution consists in discarding these ``no detection'' events, but this should be done with care and requires extra-assumptions about the measurement devices to prevent various types of side-channel attacks such as that of~\citet{lydersen10}. We will discuss this solution in more detail in Part \ref{part:pm} of this work.

\paragraph*{Commuting measurements:}

The block length is given by a protocol parameter, $m$, which will be discussed in Section~\ref{sc:params}. For Alice and Bob to run a protocol with block length $m$, we need to assume that both can perform up to $m$ measurements (with either one of two possible settings) on their share of the quantum state in such a way that the order in which they do these measurements does not affect the resulting measurement outcome distribution. 
This is a standard assumption in the model of trusted measurement devices (by opposition to device-independent cryptography) and ensures that there are no memory effects in the measurement devices. 

More formally, we assume that Alice's and Bob's share of the quantum system can be decomposed into $m$ individual quantum systems, $A \equiv A_{[m]} = A_1 A_2 \ldots A_m$ and $B \equiv B_{[m]} = B_1 B_2 \ldots B_m$ and that the measurements can be represented as operators acting on the individual subsystems. We model Alice's $i$-th measurement with setting $\phi \in \{0,1\}$ by a binary generalized measurement $\{ M_{A_i}^{\phi,x} \}_{x \in \{0,1\}}$ acting on subsystem $A_i$. The index $x$ ranges over the two possible outcomes of Alice's measurement. 
Analogously, Bob's $i$-th measurement with setting $\phi \in \{0,1\}$  is a binary generalized measurement $\{ M_{B_i}^{\phi,y}  \}_{y \in \{0,1\}}$ acting on subsystem $B_i$. The index $y$ ranges over the two possible outcomes of Bob's measurement.

\paragraph*{Complementarity of Alice's measurements:}  
  
The exact description of the measurement devices will not be relevant for our derivations. However, we will need to assume that Alice's measurements are sufficiently complementary, a property that is encapsulated by the average overlap, $\bar{c}(m,n)$ that we introduce next.
Let $m$ be the block length and $n$ the number of bits used for key extraction (see Section~\ref{sc:params} for a discussion of the protocol parameters). Let us define
\begin{align}
  c \big( \{M^x\}_x, \{N^y\}_y \big) := \max_{x,y \in \{0,1\}} \left\| M^x \big(N^y\big)^{\dag}
    \right\|_{\infty}^2, \quad \textrm{and} \quad c_i := c\Big( \big\{ M_{A_i}^{0,x} \big\}_x, \big\{ M_{A_i}^{1,y} \big\}_y \Big) \,.
\label{eq:cdef}
 \end{align}
In an ideal physical implementation of the protocol with complementary measurements (for example in the computational and Hadamard basis), we would have $c_i = \frac12$ for all $i \in [m]$. In realistic implementations, its value will be larger. We assume that there exists a reliable upper bound on $c_i$. More precisely, we assume that
\begin{align}
  \bar{c}(m, n) := \max_{\pi \in \Pi_{m,n}} \left( \prod_{i \in \pi} c_i \right)^{\frac1{n}} \leq \bar{c} \label{eq:cbar} \,.
\end{align}

We always have $c_i \in [0,1]$ and the condition $\bar{c} <1$ is necessary to ensure secrecy.
As long as the commuting measurement assumption holds, the parameter $\bar{c}$ can in principle be measured directly in an experiment\,---\,even if the operators $\{ M_{A_i}^{\phi,x} \}_{\phi,x}$ are unknown.\footnote{See, e.g.,~\cite{haenggi11,lim12} for an estimation strategy based on Bell tests.}
  
For Bob we do not need to assume a bound on the complementarity parameter. (We only need to assume that the measurements commute, as described in the previous item.)

\subsection{Protocol parameters and overview}
\label{sc:params}

\begin{table}[t!]
\fbox{%
\begin{minipage}{\textwidth}
  
  { $(K_A, K_B, S, C, F) = $}
  {\large\qkdeb}
  { $\!\!\big( \rho_{AB} \big)$:}
       
  \begin{description}
  
    \item[Input:] Alice and Bob are given a state $\rho_{AB}$, where $A \equiv A_{[m]}$ and $B \equiv B_{[m]}$ are comprised of $m$ quantum systems each. 
    
    \item[Randomization:]
    They agree on a random string $\Phi \in \{0,1\}^m$, a random subset $\Pi \in \Pi_{m,k}$, and random hash functions $H_{\mathrm{ec}} \in \cH_{\textrm{ec}}$ as well as $H_{\mathrm{pa}} \in \cH_{\textrm{pa}}$. The corresponding uniformly random seeds are denoted $S = (S^{\Phi}, S^{\Pi}, S^{H_{\textrm{ec}}}, S^{H_{\textrm{pa}}})$.
        
    \item[Measurement:]
    Alice and Bob measure the $m$ quantum systems with the setting $\Phi$. They store the binary measurement outcomes in two strings, the \emph{raw keys}. These are denoted $(X, V)$ and $(Y, W)$ for Alice and Bob, respectively. Here $V, W$ are of length $k$ and correspond to the indices in $\Pi$, whereas $X, Y$ of length $n$ correspond to indices not in $\Pi$.

    \item[Parameter Estimation:]
    Alice sends $V$ to Bob, the transcript is denoted $C^V$. Bob compares $V$ and $W$. If the fraction of errors exceeds $\delta$, Bob sets the flag $F^{\mathrm{pe}} = \fail$ and they abort.  
    Otherwise he sets $F^{\mathrm{pe}} = \succ$ and they proceed.

   \item[Error Correction:]
  Alice sends the syndrome $Z = \mathrm{synd}(X)$ to Bob, with transcript $C^Z$. Bob computes $\hat{X} = \mathrm{corr}(Y, Z)$.
  
    To verify the success of the error correction step, Alice computes the hash $T = H_{\mathrm{ec}}(X)$ of length $t$ and sends it to Bob, with transcript $C^T$.
    Bob computes $H_{\mathrm{ec}}(\hat{X})$. If it differs from $T$, he sets the flag $F^{\mathrm{ec}} = \fail$ and they abort the protocol.  
    Otherwise he sets $F^{\mathrm{ec}} = \succ$ and they proceed.    
    
  \item[Privacy Amplification:] They compute keys $K_A = H_{\mathrm{pa}}(X)$ and $K_B = H_{\mathrm{pa}}(\hat{X})$ of length $\ell$.
  
  \item[Output:] The output of the protocol consists of the keys $K_A$ and $K_B$, the seeds $S = (S^{\Phi}, S^{\Pi}, S^{H_{\textrm{ec}}}, S^{H_{\textrm{pa}}})$, the transcript $C = (C^V, C^Z, C^T)$ and the flags $F = (F^{\mathrm{pe}}, F^{\mathrm{ec}})$.
  In case of abort, we assume that all registers are initialized to a predetermined value.
  
  \end{description}
\end{minipage}
}
  \caption{Simple QKD Protocol. The precise mathematical model is to be found in Section~\ref{sec:model}.}
  \label{tb:simple}
\end{table}

From an information theoretic and mathematical point of view, an entanglement-based QKD protocol is simply a completely positive trace-preserving (CPTP) map composed of local operations and classical communication (LOCC) that takes a bipartite state $\rho_{AB}$ as an input and either aborts or returns two classical binary strings, the keys, which should ideally be identical and independent of the knowledge of any third party having access to a purifying system of $\rho_{AB}$ and to the transcript of the communication performed by the protocol.

We consider protocols \qkdeb{} that are parametrized by the block length, $m$, and the sub-protocols for parameter estimation, error correction and privacy amplification, respectively denoted by $\mathrm{pe}$, $\mathrm{ec}$, and $\mathrm{pa}$.

\begin{itemize}
\item
The block length, $m \in \mathbb{N}$, determines the number of individual quantum systems that are shared between Alice and Bob, and thus available to them for parameter estimation and key extraction.

\item
Parameter estimation is characterized by a tuple $\mathrm{pe} = \{ k, \delta \}$, where $k \in \mathbb{N}, k \leq m$ determines the number of quantum systems used for parameter estimation and $\delta \in (0,\frac12)$ is the tolerated error rate.
Let us for later convenience also define $n := m - k$ to denote the number of quantum systems used for key generation. 
 
 \item
  The error correcting scheme is described by a quintuple $\mathrm{ec} = \{ t, r, \mathrm{synd}, \mathrm{corr},\cH_{\mathrm{ec}} \}$. Here, $r \in \mathbb{N}$ is the length (in bits) of the error correction syndrome.
  Moreover, $\mathrm{synd}$ and $\mathrm{corr}$ are functions of the form $\mathrm{synd}: \{0,1\}^n \to \{0,1\}^r$ and $\mathrm{corr}: \{0,1\}^n \times \{0,1\}^r \to \{0,1\}^n$ used to compute the error syndrome and calculate the corrected string, respectively. We do not need to assume anything about the structure of this code. To fix ideas, let us just note that there exist good error correction codes with $r \approx n h(\delta)$, where $h$ is the binary entropy function and $\delta$ is the number of errors to be corrected.\footnote{For example, $\mathrm{synd}$ could be a linear code described by an $r \times n$ parity check matrix $H$ such that $\mathrm{synd}(x) = H x$. Moreover, $\mathrm{corr}$ can be any decoder, for example the (optimal) maximum likelihood decoder, but also a more practical suboptimal iterative decoder.} 

  Finally, $t\in \mathbb{N}$ is the length (in bits) of the hash used for verification and $\cH_{\mathrm{ec}} := \big\{h_{\mathrm{ec}}: \{0,1\}^n \rightarrow \{0,1\}^{t} \big\}$ is a 
universal$_2$ family of hash functions. We will see in Theorem~\ref{th:correct} that the size $t$ only depends logarithmically on the targeted correctness parameter.
 
 \item
  Privacy amplification is characterized by a tuple $\mathrm{pa} = \{ \ell, \cH_{\mathrm{pa}} \}$, where $\ell \in \mathbb{N}$ with $\ell \leq n$ is the length (in bits) of the extracted key and
  $\cH_{\mathrm{pa}}  := \big\{h_{\mathrm{pa}}: \{0,1\}^n \rightarrow \{0,1\}^{\ell} \big\}$ is
  a universal$_2$ family of hash functions.

  Note that $\ell$, the length of the final key, is fixed. It is in principle possible to design adaptive protocols where the final key length is chosen after parameter estimation, but this is beyond the scope of this work. 

 \end{itemize}

  This allows us to define a family of protocols \qkdeb{} in Table~\ref{tb:simple}. Note that any such protocol is simply a completely positive trace-preserving map that maps bipartite quantum states shared between Alice and Bob onto probability distributions of the classical outputs, and we will define their exact operation in Section~\ref{sec:model}.

\subsection{Exact mathematical model of the protocol} 
\label{sec:model}

  Here we describe in detail the mathematical model underlying the protocol in Table~\ref{tb:simple}.
  It is worth emphasizing that the eavesdropper does not appear anywhere in this description, but will of course be required when assessing the security of the protocol as discussed in Section \ref{sc:eb2}.

   \paragraph*{Input:}
    Alice and Bob are given a state
     $\rho_{A B}$, where $A = A_{[m]} = A_1 A_2 \ldots A_m$ consists of $m$ quantum systems of arbitrary, finite dimension, $B = B_{[m]} = B_1 B_2 \ldots B_m$ consists of $m$ quantum systems of arbitrary, finite dimension. 
     Note that apart from the above structure, the state $\rho_{A B}$ is fully general. The situation is depicted in Figure~\ref{fig:state_input}.

   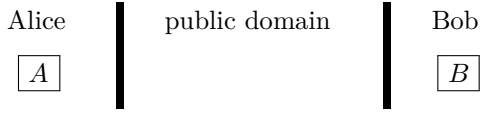
\begin{figure}[h!]
\centering
\begin{picture}(200, 40)
  \put(11, 30){Alice}
  \put(70, 30){public domain}
  \put(170, 30){Bob}
  \put(15, 10){\boxed{\phantom{X}}}
  \put(19, 10){$A$}
  \put(172, 10){\boxed{\phantom{X}}}
  \put(176, 10){$B$}
  \linethickness{1mm}
  \put(50, 0){ \line(0, 1){40} }
  \put(150, 0){ \line(0, 1){40} }
\end{picture}
\caption{State of the classical and quantum systems at the beginning of the protocol. The initial state is denoted $\rho_{AB}$.}
\label{fig:state_input}
\end{figure}

    \paragraph*{Randomization:}
    
    We model the randomization by random seeds (uniform random variables), shared between Alice and Bob over the public authenticated channel. These random seeds are represented by a quantum state $\rho_{S}$ which is assumed to be maximally mixed and independent of $\rho_{AB}$. The situation after randomization is depicted in Figure~\ref{fig:state_rand}. Let us now detail the content of the system $S$.

    The first random variable is a random basis choice for each quantum system. 
     This is modeled as a register $S^{\Phi}$ in the state
     \begin{align}
     \rho_{S^\Phi} = \sum_{ \phi \in \{0,1\}^m } \frac{1}{2^m}\, \proj{\phi}_{S^\Phi} ,
     \end{align}
     where $\{ |\phi\rangle \}_{\phi \in \{0,1\}^m}$ is an orthonormal basis of the space $S^{\Phi}$ and $\phi = \phi_{[m]} = (\phi_1, \phi_2, \ldots, \phi_m)$ with $\phi_i \in \{0,1\}$. The total state at the beginning of the protocol is thus of the form
     $\rho_{AB} \otimes \rho_{S^\Phi}$.
     
    The seed for the choice of the random subset is denoted $S^{\Pi}$ and is initially in the state
    \begin{align}
       \rho_{S^\Pi} = \sum_{ \pi \in  \Pi_{m,k} } \frac{1}{ {m \choose k} }\, \proj{\pi}_{S^\Pi} \,,
     \end{align}
     where $\{ |\pi\rangle \}_{\pi \in \Pi_{m,k}}$ is an orthonormal basis of the space $S^{\Pi}$. For any $\pi \in \Pi_{m,k}$, we denote its $k$ elements by $\pi_i$, for $i \in [k]$ and we denote by $\bar{\pi} \in [m]$ the complement of $\pi$. 
     
     At this point we reorder the measurement settings in $S^{\Phi}$ into two parts: the settings to be used for measuring quantum systems in $\pi$ will be stored in a register $S^{\Xi}$ and the settings to be used for measuring the remaining $n$ quantum systems in $\bar{\pi}$ will be stored in a register $S^{\Theta}$. Formally, we consider the function
     \begin{align}
       \textrm{ro}: \{0,1\}^m \times \Pi_{m,k} \to \{0,1\}^k \times \{0,1\}^n , \quad (\phi, \pi) \mapsto (\phi_{\pi}, \phi_{\bar{\pi}}) \,.
     \end{align}
     Since $S^{\Phi}$ is uniformly random, the resulting state after applying this function and discarding $S^{\Phi}$ is of the form 
     \begin{align}
       \rho_{S^\Pi S^\Xi S^\Theta} = \tr_{S^{\Phi}} \big\{ \mathcal{E}_\textrm{ro}(\rho_{S^{\Phi}} \otimes \rho_{S^{\Pi}} )\big\} = \rho_{S^\Pi} \otimes \rho_{S^\Xi} \otimes \rho_{S^{\Theta}}\, ,
     \end{align}
     where the registers containing $S^{\Xi}$ and $S^{\Theta}$ are again uniformly random:
     \begin{align} \label{eq:uniform-basis}
    \rho_{S^{\Xi}} = \sum_{ \xi \in \{0,1\}^k } \frac{1}{2^k}\, \proj{\xi}_{S^\Xi}
       \quad \textrm{and} \quad \rho_{S^\Theta} = \sum_{ \theta \in \{0,1\}^n } \frac{1}{2^n}\, \proj{\theta}_{S^\Theta}  
     \end{align}
     for $\xi = \xi_{[k]} = (\xi_1, \xi_2, \ldots, \xi_k)$ and $\theta = \theta_{[n]} = (\theta_1, \theta_2, \ldots, \theta_n)$ with $\xi_i, \theta_i \in \{0,1\}$.
     
The choice of the hash function in the family $\mathcal{H}_{\mathrm{ec}} = \{h_{\mathrm{ec}}: \{0,1\}^n \rightarrow \{0,1\}^{t} \}$ and the choice of hash function in the family $\mathcal{H}_{\mathrm{pa}} = \{h_{\mathrm{pa}}: \{0,1\}^n \rightarrow \{0,1\}^{t} \}$ are modeled via random seeds
\begin{align}
\rho_{S^{H_{\mathrm{ec}}}} = \sum_{h \in \mathcal{H}_{\mathrm{ec}}} \frac{1}{|\mathcal{H}_{\mathrm{ec}}|} |h\rangle \!\langle h|_{S^{H_{\mathrm{ec}}}} \qquad \textrm{and} \qquad \rho_{S^{H_{\mathrm{pa}}}} = \sum_{h \in \mathcal{H}_{\mathrm{pa}}} \frac{1}{|\mathcal{H}_{\mathrm{pa}}|} |h\rangle \!\langle h|_{S^{H_{\mathrm{pa}}}}.
\end{align}

\begin{figure}[t]
\centering
\begin{picture}(200, 60)
  \put(11, 50){Alice}
  \put(70, 50){public domain}
  \put(170, 50){Bob}
  \put(15, 30){\boxed{\phantom{X}}}
  \put(15, 5){\boxed{\phantom{X}}}
  \put(19, 30){$A$}
  \put(19, 5){$S$}
  \put(172, 30){\boxed{\phantom{X}}}
  \put(172, 5){\boxed{\phantom{X}}}
  \put(176, 30){$B$}
  \put(176, 5){$S$}
  \put(32, 9){\vector(1, 0){138}}
  \put(100,9){\vector(0, 1){15}}
  \put(92, 30){\boxed{\phantom{X}}}
  \put(96, 30){$S$}
  \linethickness{1mm}
  \put(50, 0){ \line(0, 1){60} }
  \put(150, 0){ \line(0, 1){60} }
\end{picture}
\caption{State of the classical and quantum systems after randomization. The state of the total system is $\rho_{AB} \otimes \rho_S$, where $\rho_S = \rho_{S^\Phi} \otimes \rho_{S^\Pi} \otimes \rho_{S^{H_{\mathrm{ec}}}} \otimes \rho_{S^{H_{\mathrm{pa}}}}$.}
\label{fig:state_rand}
\end{figure}
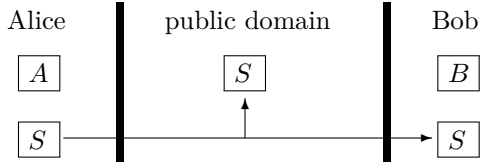

\paragraph*{Measurement:}
  
  We split the measurement process into two parts, measuring the systems in the set $\pi$ and $\bar{\pi}$ separately. While this distinction is not relevant for the practical implementation of the protocol, the notation introduced here will be important for the security analysis later.
  The first measurement concerns the registers in $\pi$, which are used for parameter estimation. For any subset $\pi \in \Pi_{m,k}$, we define a completely positive trace-preserving map $\cM_{A \to V|S^{\Xi}}^{\pi}: A_{\pi}S^\Xi \to V A_{\pi} S^{\Xi}$ where $V = V_{[k]} = V_1 \otimes V_2 \otimes \ldots V_k$ models $k$ binary classical registers storing the measurement outcomes. The map is given by
  \begin{align}
    \cM_{A \to V|S^{\Xi}}^{\pi}(\cdot) = \!\! \sum_{\xi \in \{0, 1\}^{k}} \sum_{ v \in \{0, 1\}^{k}} 
    \ket{v}_{V} 
      \Big( M_{A_{\pi}}^{\xi,v} \otimes  \proj{\xi}_{S^{\Xi}} \Big)\ \cdot\ \Big( M_{A_{\pi}}^{\xi,v} \otimes  \proj{\xi}_{S^{\Xi}} \Big)^{\dagger} \bra{v}_{V} ,
  \end{align}
    where
    $M_{A_{\pi}}^{\xi,v} := \bigotimes_{i \in [k]} M_{A_{\pi_i}}^{\xi_i,v_i}$. This map measures the $k$ subsystems determined by $\pi$ using the (random) measurement settings stored in the register $S^\Xi$. The results are stored in the classical register $V$, and the post-measurement state remains in the systems $A_{\pi}$.
  
  Similarly, we define $\cM_{B\to W|S^{\Xi}}^{\pi}: B_{\pi}S^\Xi \to W B_{\pi} S^{\Xi}$ as
  \begin{align}
    \cM_{B\to W|S^{\Xi}}^{\pi}(\cdot) = \!\! \sum_{\xi \in \{0, 1\}^{k}} \sum_{ w \in \{0, 1\}^{k}} 
    \ket{w}_{W} 
      \Big(M_{B_{\pi}}^{\xi,w} \otimes  \proj{\xi}_{S^{\Xi}}  \Big)\ \cdot\ \Big(M_{B_{\pi}}^{\xi,w} \otimes  \proj{\xi}_{S^{\Xi}}  \Big)^{\dagger} \bra{w}_{W}  ,
  \end{align}
    where
    $M_{B_{\pi}}^{\xi,w} := \bigotimes_{i \in [k]} M_{B_{\pi_i}}^{\xi_i,w_i}$.
  The two maps $\cM_{A\to V|S^{\Xi}}^{\pi}$ and $\cM_{B\to W|S^{\Xi}}^{\pi}$ commute since they act on different systems and we write their concatenation as $\cM_{A\to V|S^{\Xi}}^{\pi} \circ \cM_{B\to W|S^{\Xi}}^{\pi} = \cM_{B\to W|S^{\Xi}}^{\pi} \circ \cM_{A\to V|S^{\Xi}}^{\pi}$.
  
  So far we have considered $\pi$ to be fixed. The full measurement for parameter estimation instead consults the register $S^{\Pi}$ and is modeled as a map $\cM_{AB\to VW|S^{\Pi}S^{\Xi}}: A B S^{\Pi}S^{\Xi} \to A B V W S^{\Pi}S^{\Xi}$ given by
  \begin{align}
    \cM_{AB\to VW|S^{\Pi}S^{\Xi}}(\cdot) = \sum_{\pi \in \Pi_{m,k}}  \,  \cM_{A\to V|S^{\Xi}}^{\pi} \circ \cM_{B\to W|S^{\Xi}}^{\pi} \big( \proj{\pi}_{S^\Pi} \ \cdot\ \proj{\pi}_{S^{\Pi}} \big).
  \end{align}
  The state of the total system after the measurement required for parameter estimation is thus given by
  \begin{align}
  \tau_{A B V W S^{\Pi}S^{\Xi} S^{\Theta}}
  &= \cM_{AB\to VW|S^{\Pi}S^{\Xi}} ( \rho_{A B S^\Pi  S^\Xi S^\Theta} ) \\
  &= \sum_{\pi \in \Pi_{m,k}} \sum_{\xi \in \{0,1\}^k} \sum_{v,w \in \{0,1\}^k} \frac{1}{ 2^k {m \choose k} } \proj{\pi,\xi}_{S^{\Pi}S^{\Xi}} \otimes \rho_{S^{\Theta}} \otimes~  \nonumber\\
  &\qquad\qquad\qquad  \dots\  \proj{v, w}_{VW} \otimes \big( M_{A_{\pi}}^{\xi,v} \otimes M_{B_{\pi}}^{\xi,w} \big) \rho_{AB} \big( M_{A_{\pi}}^{\xi,v} \otimes M_{B_{\pi}}^{\xi,w} \big)^{\dagger} \,.
  \end{align}
  
  The second measurement concerns the quantum systems used for extracting the secret key. The corresponding measurement maps are defined analogously to the measurements maps above, but now act on the systems determined by $\bar{\pi}$, the complement of $\pi$ in $[m]$.
We define
    \begin{align}
    \cM_{A \to X|S^{\Pi} S^{\Theta}}(\cdot) &=  \sum_{\pi \in \Pi_{m,k}}\sum_{\theta,x \in \{0, 1\}^{n}} 
    \ket{x}_{X} 
      \big( M_{A_{\bar{\pi}}}^{\theta,x}\otimes  \proj{\pi,\theta}_{S^{\Pi} S^{\Theta}}  \big)\ \cdot\ \big(M_{A_{\bar{\pi}}}^{\theta,x}\otimes  \proj{\pi,\theta}_{S^{\Pi} S^{\Theta}} \big)^{\dagger} \bra{x}_{X},\label{eq:xymeasure}\\
     \cM_{B\to Y|S^{\Pi} S^{\Theta}}(\cdot) &=  \sum_{\pi \in \Pi_{m,k}}\sum_{\theta,y \in \{0, 1\}^{n}}  
    \ket{y}_{Y} 
      \big(M_{B_{\bar{\pi}}}^{\theta,y} \otimes \proj{\pi,\theta}_{S^{\Pi} S^{\Theta}}  \big)\ \cdot\ \big(M_{B_{\bar{\pi}}}^{\theta,y} \otimes \proj{\pi,\theta}_{S^{\Pi}S^{\Theta}}\big)^{\dagger} \bra{y}_{Y} 
  \end{align}
  as well as $\cM_{AB\to XY|S^{\Pi}S^{\Theta}} = \cM_{A \to X|S^{\Pi}S^{\Theta}} \circ \cM_{B \to Y|S^{\Pi}S^{\Theta}}$.  
  It is evident that all measurements $\cM$ defined so far mutually commute because they act either on classical registers or on distinct quantum registers. Finally, we define the total measurement map as $\cM_{AB \to VWXY|S^{\Pi}S^{\Xi}S^{\Theta}} := \cM_{AB \to VW|S^{\Pi} S^{\Xi}} \circ \cM_{AB\to XY|S^{\Pi}S^{\Theta}}$.

Of particular interest is the state of the system after measurement and after we discard the quantum systems. This is given by a classical state $\sigma_{VWXY S^{\Pi} S^{\Xi} S^{\Theta}}$. This state is of the form
\begin{align}
  &\sigma_{V W X Y S^{\Pi} S^{\Xi} S^{\Theta}} \\
  &\quad=\tr_{AB} \big( \cM_{AB \to VWXY|S^{\Pi}S^{\Xi}S^{\Theta}}( \rho_{A B S^{\Pi}S^{\Xi}S^{\Theta} }) \big) \\
  &\quad=
  \tr_{AB} \big( \cM_{AB\to XY|S^{\Pi}S^{\Theta}} ( \tau_{A B V W S^{\Pi}S^{\Xi}S^{\Theta} } ) \big) \\
  &\quad = \sum_{\pi \in \Pi_{m,k}} \sum_{\xi \in \{0,1\}^k \atop \theta \in \{0,1\}^n } \sum_{v,w \in \{0,1\}^k \atop
  x,y \in \{0,1\}^n} \frac{1}{ 2^m {m \choose k} } \proj{\pi,\xi,\theta}_{S^{\Pi}S^{\Xi}S^{\Theta}} \otimes \proj{v, w, x, y}_{VWXY}\otimes~  \nonumber\\
  &\qquad\qquad  \dots\  \tr_{AB} \Big\{ \big( \widetilde{M}_{A_{\pi}}^{\xi,v} \otimes \widetilde{M}_{A_{\bar{\pi}}}^{\theta,x} \otimes \widetilde{M}_{B_{\pi}}^{\xi,w} \otimes \widetilde{M}_{B_{\bar{\pi}}}^{\theta,y} \big) \rho_{AB} \Big\} \, ,
\end{align}
where we write $\widetilde{M}_{A_{\pi}}^{\xi,v} = \big(M_{A_{\pi}}^{\xi,v}\big)^{\dagger} M_{A_{\pi}}^{\xi,v}$ and analogously introduce 
$\widetilde{M}_{A_{\bar{\pi}}}^{\theta,x}$, $\widetilde{M}_{B_{\pi}}^{\xi,w}$ and $\widetilde{M}_{B_{\bar{\pi}}}^{\theta,y}$.

The situation after the complete measurement is depicted in Figure~\ref{fig:state_meas}.

\begin{figure}[t]
\centering
\begin{picture}(300, 60)
  \put(31, 50){Alice}
  \put(120, 50){public domain}
  \put(240, 50){Bob}
  
  \put(35, 30){\boxed{\phantom{X}}}
  \put(39, 30){$S$}
  \put(5, 5){\boxed{\phantom{X}}}
  \put(9, 5){$A$}
  \put(43,25.5){\line(0, -1){16}}
  \put(4, 0){\line(1,1){17}}
  \put(23, 9){\vector(1, 0){30}}
  \put(55, 5){\boxed{\phantom{X}}}
  \put(59, 5){$V$}
  \put(75, 5){\boxed{\phantom{X}}}
  \put(79, 5){$X$}

  \put(245, 30){\boxed{\phantom{X}}}
  \put(249, 30){$S$}
  \put(215, 5){\boxed{\phantom{X}}}
  \put(219, 5){$B$}
  \put(253,25.5){\line(0, -1){16}}
  \put(214, 0){\line(1,1){17}}
  \put(233, 9){\vector(1, 0){30}}
  \put(265, 5){\boxed{\phantom{X}}}
  \put(268, 5){$W$}
  \put(285, 5){\boxed{\phantom{X}}}
  \put(289, 5){$Y$}

  \put(245,-1){$\mathcal{M}$}
  \put(35,-1){$\mathcal{M}$}
 
 \put(142, 30){\boxed{\phantom{X}}}
  \put(146, 30){$S$}
  \linethickness{1mm}
  \put(100, 0){ \line(0, 1){60} }
  \put(200, 0){ \line(0, 1){60} }
\end{picture}
\caption{State of the classical and quantum systems during and after measurement. The measurement can be summarized as a CPTP map $\rho_{AB} \otimes \rho_{S^\Phi} \otimes \rho_{S^\Pi} \mapsto \sigma_{VWXYAB S^\Phi S^\Pi}$.}
\label{fig:state_meas}
\end{figure}
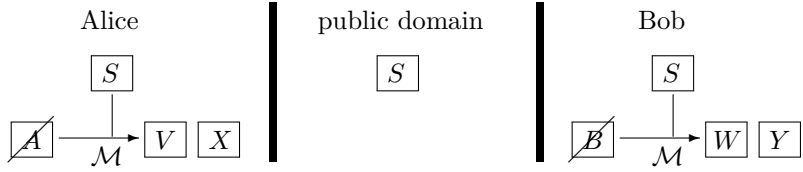

\paragraph*{Parameter estimation:}

We model parameter estimation by a test function acting on the registers $V$ and $W$ and creating a binary flag $F^{\textrm{pe}}$ as follows:
\begin{align}
  \mathrm{test}: \{ 0, 1 \}^k \times \{ 0, 1\}^k \to \{ \fail, \succ \}, \quad \mathrm{pe}(v, w) = \begin{cases} \fail &
    \textrm{if}\ \sum_{i \in [k]} 1\{ v_i \neq w_i \} \geq k \delta,
     \\
  \succ & \textrm{otherwise}. \end{cases} 
\end{align}

This test can be applied to the states $\tau_{A B V W S^{\Pi} S^{\Xi}S^{\Theta}}$ or $\sigma_{V W X Y S^{\Pi}S^{\Xi} S^{\Theta} }$ defined previously. 
This requires Alice to communicate $V$ to Bob on the authenticated classical channel in order to evaluate the value of $\mathrm{pe}(v, w)$ and the transcript of this communication is stored in the variable $C^V = V$.

We are specifically interested in the state
  $\tau_{A B V W S^{\Pi}  S^{\Xi}S^{\Theta} F^{\textrm{pe}}} = \mathcal{E}_{\textrm{pe}} \big( \tau_{A B V W S^{\Pi} S^{\Xi}S^{\Theta} } \big) $
and the corresponding state conditioned on the outcome $F^{\textrm{pe}}=\succ$, given by
\begin{align}
  &\tau_{A B V W S^{\Pi} S^{\Xi} S^{\Theta}|F^{\textrm{pe}}=\succ} = \frac{1}{\Pr [F^{\textrm{pe}} = \succ]_{\tau}} \sum_{\pi \in \Pi_{m,k}} \sum_{\xi \in \{0,1\}^k} \sum_{v,w \in \{0,1\}^k \atop \sum_{i =1}^k 1\{ v_i \neq w_i \} < k \delta} \frac{1}{ 2^k {m \choose k} } \proj{\pi,\xi}_{S^{\Pi}S^{\Xi}} \otimes~  \nonumber\\
  &\qquad\qquad\qquad  \dots\   \rho_{S^{\Theta}} \otimes \proj{v,w}_{VW} \otimes \big( M_{A_{\pi}}^{\xi,v} \otimes M_{B_{\pi}}^{\xi,w} \big) \rho_{AB} \big( M_{A_{\pi}}^{\xi,v} \otimes M_{B_{\pi}}^{\xi,w} \big)^{\dagger} \,. \label{eq:tautau}
\end{align}
We will see that this state is crucial for the security analysis in the next section. Finally, we note that $\cM_{AB\to XY|S^{\Pi} S^{\Theta}}$ and $\mathcal{E}_{\textrm{pe}}$ commute, and thus in particular we find that
\begin{align}
  \tr_{AB} \big\{ \cM_{AB\to XY|S^{\Pi} S^{\Theta} } ( \tau_{A B V W S^{\Pi} S^{\Xi} S^{\Theta}|F^{\textrm{pe}}=\succ} ) \big\} &= \sigma_{V W  X Y S^{\Pi} S^{\Xi} S^{\Theta}|F^{\textrm{pe}}=\succ } ,  \qquad \textrm{where}\\
  \sigma_{V W X Y S^{\Pi}  S^{\Xi} S^{\Theta }F^{\textrm{pe}}} &= \mathcal{E}_{\textrm{pe}} ( \sigma_{V W X Y S^{\Pi} S^{\Xi} S^{\Theta}} ) \,.
  \label{eq:sigmastate}
\end{align}

We then relabel $V$ to $C^V$ and keep it around as part of the transcript, while we discard $W$ after performing parameter estimation. The situation after parameter estimation is depicted in Figure~\ref{fig:state_parest}.

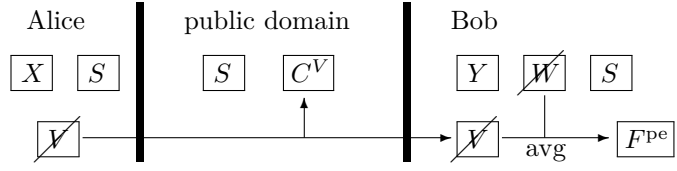
\begin{figure}[t]
\centering
\begin{picture}(250, 60)
  \put(11, 50){Alice}
  \put(70, 50){public domain}
  \put(170, 50){Bob}

  \put(5, 30){\boxed{\phantom{X}}}
  \put(9, 30){$X$}
  \put(30, 30){\boxed{\phantom{X}}}
  \put(34, 30){$S$}
  \put(15, 5){\boxed{\phantom{X}}}
  \put(19, 5){$V$}
  \put(14, 0){\line(1,1){17}}

  \put(172, 30){\boxed{\phantom{X}}}
  \put(176, 30){$Y$}
  \put(197, 30){\boxed{\phantom{X}}}
  \put(200, 30){$W$}
  \put(195, 25){\line(1,1){17}}
  \put(222, 30){\boxed{\phantom{X}}}
  \put(226, 30){$S$}
  \put(172, 5){\boxed{\phantom{X}}}
  \put(176, 5){$V$}
  \put(170, 0){\line(1,1){17}}

  \put(232, 5){\boxed{\phantom{Xx}}}
  \put(235, 5){$F^{\mathrm{pe}}$}
  \put(189, 9){\vector(1,0){40}}
  \put(205, 9){\line(0,1){16}}

  \put(32, 9){\vector(1, 0){138}}
  \put(115,9){\vector(0, 1){15}}
  \put(77, 30){\boxed{\phantom{X}}}
  \put(81, 30){$S$}
  \put(107, 30){\boxed{\phantom{X~}}}
  \put(110, 30){$C^V$}
  \put(198, 2){avg}

 \linethickness{1mm}
  \put(50, 0){ \line(0, 1){60} }
  \put(150, 0){ \line(0, 1){60} }
\end{picture}
\caption{State of the classical and quantum systems during and after parameter estimation. Parameter estimation is summarized as a CPTP map $\sigma_{VW} \mapsto \sigma_{C^V F^{\mathrm{pe}}}$.}
\label{fig:state_parest}
\end{figure}

\paragraph*{Error correction:}

The error correction part of the protocol is split into two parts. The first part consists of the actual error correction procedure, determined by two functions $\mathrm{synd}$ and $\mathrm{corr}$ that are executed by Alice and Bob, respectively. We do not assume anything about these functions, but rather check their success in the second part by evaluating hash functions.

First Alice computes a syndrome $Z = \mathrm{synd}(X)$ and sends it to Bob over the public channel. Bob then computes an estimate $\hat{X} = \mathrm{corr}(Y, Z)$, discarding $Y$ in the process. 

Alice and Bob then need to check that the decoding procedure succeeded with high probability by comparing hashes of their respective strings $X$ and $\hat{X}$ and abort the protocol if they differ. 
Alice computes a hash of size $t$ (in bits) of $X$ and sends it to Bob, who computes the corresponding hash for $\hat{X}$. 
This test is summarized as a classical map $\mathrm{ec}$ acting on registers $X$, $\hat{X}$ and $S^{H_{\mathrm{ec}}}$ creating a transcript of the hash value $C^T$ and a binary flag $F^{\mathrm{ec}}$ as follows:
\begin{align}
  \mathrm{ec}: \{ 0, 1 \}^n \times \{ 0, 1\}^n \times \mathcal{H}_{\mathrm{ec}}  \to \{ 0,1\}^t \times \{ \fail, \succ \}, \quad (x, \hat{x}) \mapsto \begin{cases} (h_{\mathrm{ec}}(x), \fail) &
    \textrm{if}\ h_{\mathrm{ec}}(x) \ne h_{\mathrm{ec}}(\hat{x}),
     \\
  (h_{\mathrm{ec}}(x), \succ) & \textrm{otherwise}. \end{cases} 
\end{align}

These classical functions are modeled using CPTP maps $\mathcal{E}_{\mathrm{synd}}$, $\mathcal{E}_{\mathrm{corr}}$ and $\mathcal{E}_{\mathrm{ec}}$, respectively.
Applying them to the state $\sigma_{X Y C^V S^{\Pi} S^{\Xi} S^{\Theta} F^{\textrm{pe}}}$ yields
\begin{align}
\sigma_{X \hat{X} C^V C^Z C^T S^{\Pi} S^{\Xi} S^{\Theta} S^{H_{\mathrm{ec}}} F^{\textrm{pe}} F^{\textrm{ec}} } = \tr_Y \Big\{ \mathcal{E}_{\mathrm{ec}} \circ \mathcal{E}_{\mathrm{corr}} \circ \mathcal{E}_{\mathrm{synd}} \left(\sigma_{X Y C^V S^{\Pi} S^{\Xi} S^{\Theta} F^{\textrm{pe}}} \otimes \rho_{S^{H_{\mathrm{ec}}}} \right) \Big\}, \label{eq:statebeforepa}
 \end{align}
 where the transcript register $C^Z$ contains the value of the syndrome and $C^T$ the output of Alice's hash. This process is depicted in Figure~\ref{fig:state_ec}.

\begin{figure}[h]
\centering
\begin{picture}(320, 120)
  \put(31, 110){Alice}
  \put(120, 110){public domain}
  \put(250, 110){Bob}

  \put(30, 65){\boxed{\phantom{X}}}
  \put(34, 65){$X$}
  \put(75, 65){\boxed{\phantom{X}}}
  \put(79, 65){$Z$}
  \put(74, 60){\line(1,1){17}}
  \put(47, 69){\vector(1,0){27}}
  \put(50, 73){synd}

  \put(0, 33){\boxed{\phantom{X}}}
  \put(4, 33){$S$}
  \put(75, 33){\boxed{\phantom{X}}}
  \put(79, 33){$T$}
  \put(74, 28){\line(1,1){17}}
  \put(17, 37){\vector(1,0){57}}
  \put(40, 37){\line(0,1){23.5}}
  \put(35, 28){hash}

  \put(247, 90){\boxed{\phantom{X}}}
  \put(251, 90){$Y$}
  \put(246, 85){\line(1,1){17}}
  \put(222, 65){\boxed{\phantom{X}}}
  \put(226, 65){$Z$}
  \put(220, 60){\line(1,1){17}}
  \put(239, 69){\vector(1,0){35}}
  \put(255, 69){\line(0,1){16}}
  \put(276, 65){\boxed{\phantom{X}}}
  \put(279, 65){$\hat{X}$}
  \put(244, 60){corr}

  \put(222, 33){\boxed{\phantom{X}}}
  \put(226, 33){$T$}
  \put(220, 28){\line(1,1){17}}
  \put(302, 33){\boxed{\phantom{X}}}
  \put(306, 33){$S$}
  \put(301, 37){\vector(-1,0){35}}
  \put(285, 37){\line(0,1){23.5}}
  \put(249, 33){\boxed{\phantom{X}}}
  \put(253, 33){$\hat{T}$}
  \put(246, 28){\line(1,1){17}}
  \put(232, 5){\boxed{\phantom{Xx}}}
  \put(235, 5){$F^{\mathrm{ec}}$}
  \put(237,37){\line(1,0){12}}
 \put(262, 5){\boxed{\phantom{Xx}}}
  \put(265, 5){$F^{\mathrm{pe}}$}
   \put(243,37){\vector(0,-1){20}}
  \put(274, 28){hash}

  \put(92, 69){\vector(1, 0){128}}
  \put(182,69){\vector(0, 1){15}}
  \put(117, 90){\boxed{\phantom{X}}}
  \put(121, 90){$S$}
  \put(142, 90){\boxed{\phantom{X~}}}
  \put(145, 90){$C^V$}
  \put(172, 90){\boxed{\phantom{X~}}}
  \put(175, 90){$C^Z$}
  \put(92, 37){\vector(1, 0){128}}
  \put(182,37){\vector(0, -1){20}}
  \put(172, 5){\boxed{\phantom{X~}}}
  \put(175, 5){$C^T$}

 \linethickness{1mm}
  \put(100, 0){ \line(0, 1){120} }
  \put(200, 0){ \line(0, 1){120} }
\end{picture}
\caption{State of the classical and quantum systems during and after error correction. Error correction is summarized as a CPTP map $\sigma_{XY} \otimes \rho_{S^{H_\mathrm{ec}}} \mapsto \sigma_{X \hat{X} C^Z F^{\mathrm{ec}}}$.}
\label{fig:state_ec}
\end{figure}
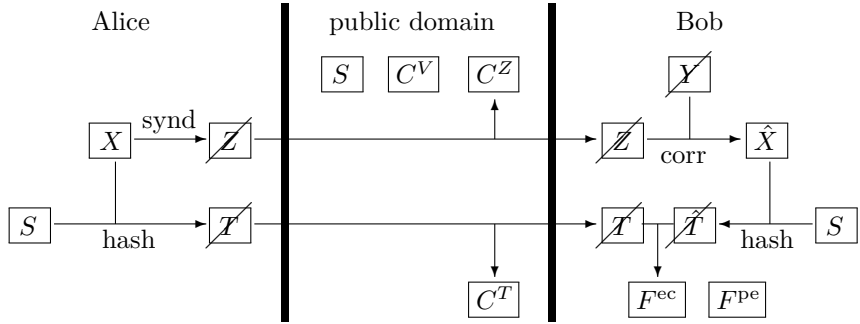

\paragraph*{Privacy amplification:}

Alice and Bob use the seed $H_{\mathrm{pa}}$ to choose a hash function, which they then both apply on their raw key to compute $K_A = H_{\mathrm{pa}}(X)$ and $K_B = H_{\mathrm{pa}}(\hat{X})$, their respective keys. 
Formally, the privacy amplification map is defined as:
\begin{equation}
  \mathrm{pa}:
  \{ 0, 1 \}^n \times \{ 0, 1\}^n \times \mathcal{H}_{\mathrm{pa}}  \to \{ 0, 1 \}^{\ell} \times \{ 0, 1\}^\ell, \quad
 ( x, \hat{x}, h_{\mathrm{pa}} ) \mapsto  (h_{\mathrm{pa}}(x), h_{\mathrm{pa}}(\hat{x})).
 \end{equation}
Denoting by $K_A$ and $K_B$ the respective key spaces of Alice and Bob, the final quantum state is
\begin{align}
\omega_{K_A K_B C { S} F} = \tr_{X\hat{X}} \big\{ \mathcal{E}_{\mathrm{pa}}(\sigma_{X \hat{X}  C { S} F} \otimes \rho_{S^{H_{\mathrm{pa}}}})\big\}.
\end{align}
Finally, Bob reveals the status of his flag registers. This final step is depicted in Figure~\ref{fig:state_privacy}.

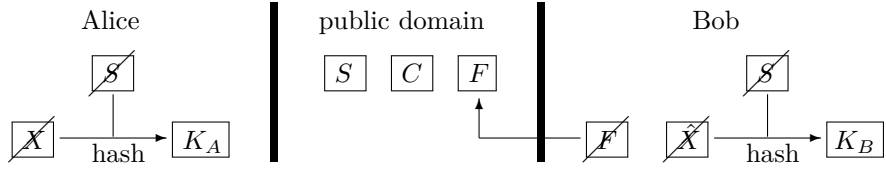
\begin{figure}[t]
\centering
\begin{picture}(330, 60)
  \put(31, 50){Alice}
  \put(120, 50){public domain}
  \put(260, 50){Bob}
  
  \put(35, 30){\boxed{\phantom{X}}}
  \put(39, 30){$S$}
   \put(34, 25){\line(1,1){17}}
 \put(5, 5){\boxed{\phantom{X}}}
  \put(9, 5){$X$}
  \put(43,25.5){\line(0, -1){16}}
  \put(4, 0){\line(1,1){17}}
  \put(23, 9){\vector(1, 0){40}}
  \put(65, 5){\boxed{\phantom{Xx}}}
  \put(69, 5){$K_A$}

  \put(280, 30){\boxed{\phantom{X}}}
  \put(284, 30){$S$}
  \put(279, 25){\line(1,1){17}}
  \put(250, 5){\boxed{\phantom{X}}}
  \put(254, 5){$\hat{X}$}
  \put(288,25.5){\line(0, -1){16}}
  \put(249, 0){\line(1,1){17}}
  \put(268, 9){\vector(1, 0){40}}
  \put(310, 5){\boxed{\phantom{Xx}}}
  \put(313, 5){$K_B$}
  \put(220, 5){\boxed{\phantom{X}}}
  \put(224, 5){$F$}
  \put(219, 0){\line(1,1){17}}
  \put(218, 9){\line(-1,0){38}}
  \put(180, 9){\vector(0,1){15}}

  \put(280,0){hash}
  \put(35,0){hash}
 
 \put(122, 30){\boxed{\phantom{X}}}
  \put(126, 30){$S$}
 \put(147, 30){\boxed{\phantom{X}}}
 \put(151, 30){$C$}
 \put(172, 30){\boxed{\phantom{F}}}
 \put(176, 30){$F$}
 
  \linethickness{1mm}
  \put(100, 0){ \line(0, 1){60} }
  \put(200, 0){ \line(0, 1){60} }
\end{picture}
\caption{State of the classical and quantum systems during and after privacy amplification. Privacy amplification is a CPTP map $\sigma_{X\hat{X}} \otimes \rho_{S^{H_\mathrm{pa}}} \mapsto \omega_{K_A K_B}$. The complete final state is denoted by $\omega_{K_A K_B S C F}$.}
\label{fig:state_privacy}
\end{figure}


\section{Security of the generated key}
\label{sc:eb2}

For a detailed discussion of the security of quantum key distribution, we refer the reader to~\citet{portmann14}. For our purposes here (consistent with~\cite{renner05} and~\cite{portmann14}), we say that our protocol is $\Delta$-secure if it is $\Delta$-close to an ideal protocol in terms of the diamond distance. Note in particular that an ideal protocol is allowed to abort, but it will always output a uniformly random shared key in case it does not. Table~\ref{tb:ideal} gives such an ideal protocol, denoted \qkdideal{} designed in such a way that it is close to the original \qkdeb{} protocol.

\begin{table}[h!]
\fbox{%
\begin{minipage}{\textwidth}
  
  { $(K_A, K_B, S, C, F) = $}
  {\large \qkdideal}
  { $\!\!\big( \rho_{AB} \big)$:}

\begin{description}

  \item[Run protocol:]
  Set $(K_A, K_B, S, C, F) = \qkdeb(\rho_{AB})$. 
  
  \item[Output:]
  If $F^{\mathrm{pe}} = F^{\mathrm{ec}} = \succ$, then replace $K_A$ and $K_B$ by an independent and uniformly distributed random string $K$ of length $\ell$, i.e.\ set $K_A = K_B = K$. 
  \end{description}

\end{minipage}
}
  \caption{An ideal QKD protocol that is close to \qkdeb.}
  \label{tb:ideal}
\end{table}

In order to show that the protocol is secure, it thus suffices to show that
\begin{align}
  \Delta_{m,\mathrm{pe},\mathrm{ec},\mathrm{pa}} 
  &:=  \big\| \qkdeb{} - \qkdideal{} \big\|_{\diamond} \\
  &\;= \sup_{\rho_{ABE} \in \cS(ABE)}  \big\| \qkdeb(\rho_{ABE}) - \qkdideal(\rho_{ABE}) \big\|_{\mathrm{tr}} \label{eq:distance}
\end{align}
is very small for certain choices of parameters $k,n,\delta,\mathrm{ec}$ and $\mathrm{pa}$. In the latter expression $\rho_{ABE}$ is an arbitrary extension of $\rho_{AB}$ to an auxiliary system $E$. Without loss of generality we may take $|E| = |A| |B|$, which is sufficient to achieve the supremum.
Physically the system $E$ is held by a potential adversary, the eavesdropper. In particular, this assures that $E$ can be assumed finite-dimensional.
Hence, we need to show that the trace distance between the protocols' outputs is small for all possible input states $\rho_{ABE}$.

Let us now fix $\rho_{ABE}$ for the moment. The trace distance in~\eqref{eq:distance} can be simplified by noting that the output of \texttt{qkd\_ideal} equals the output of \texttt{qkd\_eb} if the protocol aborts. We find
\begin{align}
   &\big\| \qkdeb(\rho_{ABE}) -\qkdideal(\rho_{ABE}) \big\|_{\tr} \nonumber\\
   &\qquad = \Pr\big[ F =\,(\succ, \succ) \big]_{\omega} \cdot
   \big\| \omega_{K_A K_B S C E | F =\,(\succ, \succ)} - \chi_{K_A K_B} \otimes \omega_{S C E | F =\,(\succ, \succ)} \big\|_{\tr} \label{eq:boundthat}\\
   &\qquad =
   \big\| {\omega}_{K_A K_B S C F E \land F =\,(\succ, \succ)} - \chi_{K_A K_B}\otimes \omega_{S C F E \land F =\,(\succ, \succ)} \big\|_{\tr} \label{eq:boundthis}
\end{align}
where we use $\omega_{K_A K_B S C F E} = \texttt{qkd\_eb}_{k,n,\delta,\mathrm{ec},\mathrm{pa}}(\rho_{ABE})$ and define
a perfect key $\chi_{K_A K_B}$ as follows:
\begin{align}
  \chi_{K_A K_B} := \frac{1}{2^{\ell}} \sum_{k \in \{0,1\}^{\ell}} \proj{k}_{K_A} \otimes \proj{k}_{K_B} \,.
\end{align}
Recall that ${\omega}$ in~\eqref{eq:boundthis} corresponds to a subnormalized state with trace equal to $\Pr[ F =\,(\succ, \succ) ]_{\omega}$.

Our goal in the following is to bound~\eqref{eq:boundthis} or~\eqref{eq:boundthat} uniformly in $\rho_{ABE}$, which implies an upper bound on~\eqref{eq:distance} as well. In order to do this we will employ the following lemma which allows us to split the norm into two terms corresponding to correctness and secrecy. This has been shown, e.g., in \cite[Theorem 4.1]{portmann14}, but we provide a proof here for completeness.

\begin{lemma}
\label{lm:security}
  Let $\eps_{\textrm{ec}}, \eps_{\textrm{pa}} \in [0,1)$ be two constants. If, for every state $\rho_{ABE} \in \cS(ABE)$ and $\omega_{K_A K_B S C F E} = \qkdeb(\rho_{ABE})$, we have
  \begin{align}
    \Pr[ K_A \neq K_B \land F^{\textrm{pe}} = F^{\textrm{ec}} = \succ ]_{\omega} &\leq \eps_{\textrm{ec}}
    \qquad \textrm{and}\\
     \big\| \omega_{K_A S C F E \land F =\,(\succ, \succ)} - \chi_{K_A}\otimes \omega_{S C F E \land F =\,(\succ, \succ)} \big\|_{\tr} 
    &\leq \eps_{\textrm{pa}} \,.
  \end{align}
  Then, $\Delta_{m,\mathrm{pe},\mathrm{ec},\mathrm{pa}}  \leq \eps_{\mathrm{ec}} + \eps_{\textrm{pa}}$.
\end{lemma}

\begin{proof}
  Let us introduce an auxiliary state $\eta_{K_A K_B S C F E}$ that is equal to $\omega_{K_A K_B S C F E}$ except that we set $K_B = K_A$. Then, applying the triangle inequality to the trace distance in~\eqref{eq:boundthat} and simplifying the resulting terms, we find
  \begin{align}
    &\big\| \omega_{K_A K_B S C E | F =\,(\succ, \succ)} - \chi_{K_A K_B} \otimes \omega_{S C E | F =\,(\succ, \succ)} \big\|_{\tr} \notag\\
    &\qquad \leq \big\| \omega_{K_A K_B S C E | F =\,(\succ, \succ)} - \eta_{K_A K_B S C E | F =\,(\succ, \succ)}  \big\|_{\tr} \notag\\
    &\qquad \qquad \qquad + \big\| \eta_{K_A K_B S C E | F =\,(\succ, \succ)} - \chi_{K_A K_B} \otimes \omega_{S C E | F =\,(\succ, \succ)} \big\|_{\tr} \\
    &\qquad = \Pr \big[ K_A \neq K_B | F =\,(\succ, \succ) \big]_{\omega} + \big\| \eta_{K_A S C E | F =\,(\succ, \succ)} - \chi_{K_A} \otimes \omega_{S C E | F =\,(\succ, \succ)} \big\|_{\tr} \,.
  \end{align}
  Multiplying this with $\Pr[ F =\,(\succ, \succ)]_{\omega}$ as in~\eqref{eq:boundthat} yields the desired implication.
\end{proof}

 The first condition of the above Lemma~\ref{lm:security} ensures that the protocol is \emph{$\eps_{\textrm{ec}}$-correct}, and the second condition ensures that the protocol \emph{$\eps_{\textrm{pa}}$-secret}. If both are satisfied, we say that the protocol is \emph{$(\eps_{\mathrm{ec}} + \eps_{\textrm{pa}})$-secure}. In the security proof we can thus verify the two conditions separately.

\section{Results and discussion}
\label{sc:eb2.5}

 We will show the following theorems, which essentially give bounds on the security parameters in terms of the protocol parameters. The first theorem establishes correctness of the protocol. Correctness of the protocol is ensured in the error correction step using hash functions, and consequently correctness can be bounded in term of the length $t$ of the hash that is used. The proof is given in Section~\ref{sc:sec-corr}.

\begin{theorem}
  \label{th:correct}
  Consider the protocol \qkdeb{} in Section~\ref{sc:protocol} with $\mathrm{ec} = \{ t, \ldots \}$. Then for every state $\rho_{AB} \in \cS(AB)$ and $\omega_{K_A K_B S C F} = \qkdeb(\rho_{AB})$ we have
  \begin{align}
     \Pr[ K_A \neq K_B \land F =\,(\succ, \succ) ]_{\omega} &\leq \eps_{\rm ec} := 2^{-t} \,.
  \end{align}
\end{theorem}

  The second theorem asserts secrecy. Secrecy is ensured by a combination of the parameter estimation and privacy amplification steps of the protocol, which both introduce an error. There is a tradeoff between these two errors, parametrized by a scalar $\nu$, which ought to be optimized numerically. The proof is given in Sections~\ref{sc:sec-first}--\ref{sc:sec-last}.
  \begin{theorem}
    \label{th:secure}
    Consider the protocol \qkdeb{} in Section~\ref{sc:protocol} with $\mathrm{pe} = \{k, \delta\}$, $\mathrm{ec} = \{ t, r, \ldots \}$ and $\mathrm{pa} = \{ \ell, \ldots \}$. Then, for every state $\rho_{ABE}$ and $\omega_{K_A K_B S C F E} = \qkdeb(\rho_{ABE})$, we have
     \begin{align}
 \big\| \omega_{K_A S C F E \land F =\,(\succ,\succ)} - \chi_{K_A}\otimes \omega_{S C F E \land F =\,(\succ,\succ)} \big\|_{\tr} \leq  \inf_{\nu \in (0, \frac12 - \delta) } \eps_{\textrm{pe}}(\nu) + \eps_{\textrm{pa}}(\nu) \,,
\end{align}
  where the error functions are given as
  \begin{align}
   \eps_{\textrm{pa}}(\nu) := \frac12 \sqrt{2^{- (m-k) \big( \log \frac{1}{\bar{c}} - h(\delta + \nu) \big) + r + t + \ell}} 
   \quad \textrm{and} \quad
   \eps_{\textrm{pe}}(\nu) := 2\, e^{ - \frac{(m-k) k^2 \nu^2}{m (k+1)}}
  \end{align}
  and $h(x) := - x \log x - (1-x)\log(1-x)$ denotes the binary entropy.
  \end{theorem}

Combining Theorems~\ref{th:correct} and~\ref{th:secure} we see that total error is thus composed of three components, $\eps_{\textrm{pe}}$, $\eps_{\rm ec}$, and $\eps_{\textrm{pa}}$. Let us take a close look at these errors for the case of large $m$. First, we note that $\eps_{\rm ec}$ vanishes asymptotically when we choose $t = \log(m)$, or any other slowly growing function of $m$. To make sure that $\eps_{\textrm{pe}}$ vanishes we choose $k = \sqrt{m}$ and $\nu = \log(m)^{-1}$, for example. For a robust operation at noise level $\delta$ it is necessary (and in theory sufficient) that the error correction leakage satisfies $r \approx (m - k) h(\delta)$. Since $h(\delta + \nu) \approx h(\delta)$ by continuity, we find that
$\eps_{\textrm{pa}}$ vanishes as long as
\begin{align}
     (m-k) \left( \log \frac{1}{\bar{c}} - 2 h(\delta) \right) - \ell - \log(m)
\end{align}
is positive and grows in $m$. Since $k$ and $\log m$ become negligible compared to $m$ as $m$ gets large, 
our protocol thus achieves the asymptotically optimal rate by~\citet{devetak05}, with $\ell/m = \log \frac{1}{\bar{c}} - 2 h(\delta)$.

\begin{figure}[t]
  \begin{center}
  \includegraphics[width=0.80\textwidth]{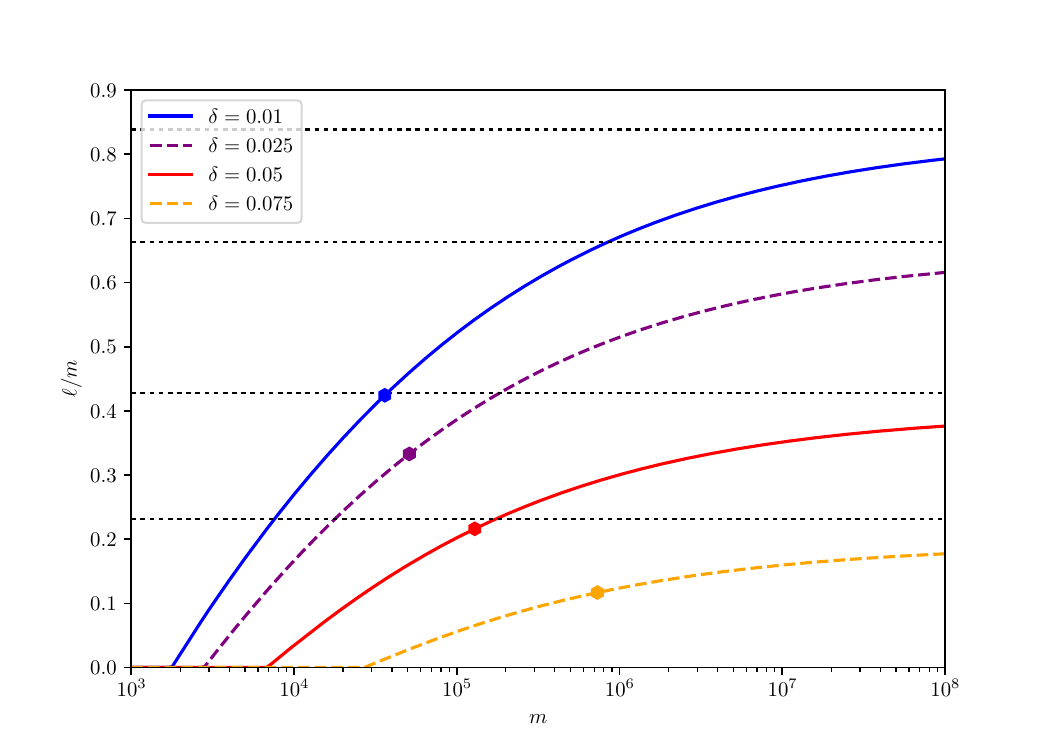}
  \end{center}
  \caption{This plot shows the maximal secret key rate $\ell/m$ as a function of $m$ for different error thresholds $\delta \in \{0.01, 0.025, 0.05, 0.075\}$, optimized over all protocols. The protocols are required to be $\eps$-secure with $\eps = 10^{-10}$ and the device parameter is assumed to be $\bar{c} = 0.5$. The error correction leakage is approximated to be $r = 1.1 (m - k) h(\delta)$, see for instance \cite{elkouss09}. (A more detailed approximation that includes finite-size effects was recently given in~\cite{tomamichel14-2}.) All remaining parameters, i.e.\ $\nu, k$ and $t$, are optimized numerically to maximize $\ell$ (code available online). The dotted horizontal lines show the corresponding asymptotic limit of the key rate for each value of $\delta$, given as $1 - 2 h(\delta)$. The markers indicate the points at which the key rate matches $50\%$ of the asymptotic limit.}
    \label{fig:keylength}
\end{figure}


\section{Security proof}
\label{sc:eb3}

The purpose of this section is to prove Theorems~\ref{th:correct} and~\ref{th:secure}.

\subsection{Error correction: Proof of Theorem~\ref{th:correct}}
\label{sc:sec-corr}

We wish to upper bound the probability of the protocol not aborting and outputting distinct final keys for Alice and Bob.
\begin{proof}[Proof of Theorem~\ref{th:correct}]
We consider the following chain of inequalities:
\begin{align}
\Pr[ K_A \neq K_B \land F^{\textrm{pe}} = F^{\textrm{ec}} =\, \succ ]_{\omega}
&\leq \Pr[ K_A \neq K_B \land  F^{\textrm{ec}} =\, \succ ]_{\omega} \\
&=  \Pr[ H_{\mathrm{pa}}(X) \neq H_{\mathrm{pa}}(X') \land H_{\mathrm{ec}}(X) = H_{\mathrm{ec}}(X') ]_{\omega}\\
&\leq  \Pr[ X \neq X' \land H_{\mathrm{ec}}(X) = H_{\mathrm{ec}}(X') ]_{\sigma}\\
& =\Pr[X \neq X']_{\sigma} \Pr[H_{\mathrm{ec}}(X) = H_{\mathrm{ec}}(X') \: | \: X \neq X']_{\sigma}\\
& \leq \Pr[H_{\mathrm{ec}}(X) = H_{\mathrm{ec}}(X') \: | \: X \neq X']_{\sigma} \\
& \leq |\mathcal{H}_{\mathrm{ec}}|^{-1}  = 2^{-t}.
\end{align}
The first inequality follows since we ignore the status of the flag $F^{\mathrm{pe}}$. The second inequality is a consequence of the fact $X = X'$ implies $H_{\mathrm{pa}}(X) = H_{\mathrm{pa}}(X')$. The third inequality follows since $\Pr[X \neq X']_{\sigma} \leq 1$ and the last one by definition of universal$_2$ hashing.
\end{proof}

\subsection{Measurements: Uncertainty tradeoff between smooth min- and max-entropy}
\label{sc:sec-first}

The crucial bound on the smooth entropy of Alice's measurement outcomes follows by the entropic uncertainty relation, suitably applied. We state the uncertainty relation in a natural form~\cite[Corollary~7.4]{mythesis}. 

\begin{proposition}
\label{pr:ur}
  Let $\tau_{APRS} \in \cSsub(APRS)$ be an arbitrary sub-normalized state with $P$ a classical register, and set $t := \tr\{\tau_{APRS}\}$. Furthermore, let $\eps \in [0, \sqrt{t})$ and let $q$ be a bijective function on $P$ that is a symmetry of $\tau_{ABCP}$ in the sense that $\tau_{ARS, P=p} = \tau_{ARS, P=q(p)}$ for all $p \in P$. Then, we have
  \begin{align}
    H_{\min}^{\eps}(X| P R)_{\sigma} + H_{\max}^{\eps}(X| P S)_{\sigma} \geq \log \frac{1}{c_q}
    , \qquad \textrm{where}  \label{eq:ucr-orig}
  \end{align}
  where $c_q = \max_{p \in P} \max_{x,z \in X} \big\| {F_{A}^{q(p),x}} \big( {F_{A}^{p,z}} \big)^{\dag} \big\|_{\infty}^2$.
  Here, $\sigma_{XPRS} = \cM_{A\to X|P}(\tau_{A P RS})$ for the map 
  \begin{align}
     \cM_{A\to X|P}\big[\cdot\big] &=  \tr_A \Bigg( \sum_{p \in P} \sum_{x \in X} 
    \ket{x}_{X} 
      \Big( \proj{p}_{P} \otimes F_{A}^{p,x} \Big)\ \cdot\ \Big( \proj{p}_{P} \otimes F_{A}^{p,x} \Big)^{\dagger} \bra{x}_{X} \Bigg) \,. \label{eq:urth2-orig}
  \end{align}
  and any set (indexed by $p \in P$) of generalized measurements $\{ F_A^{p,x} \}_{x \in X}$.
\end{proposition}

A variant of this uncertainty relation was first shown in~\cite{mythesis}, based on the techniques introduced in~\cite{tomamichel11}.
We provide a full proof of the uncertainty relation in Appendix~\ref{app:ur} for completeness. In the following corollary we apply it to the situation at hand during our protocol.

\begin{corollary}
\label{cor:smooth-min-bound}
Consider the protocol \qkdeb{} in Section~\ref{sc:protocol} with $\mathrm{pe} = \{ k, \ldots \}$ applied to a state $\rho_{ABE} \in \cS(ABE)$ and the
state $\sigma_{X Y V W S^{\Pi} S^{\Xi} S^{\Theta} F^{\mathrm{pe}} E}$ as in~\eqref{eq:sigmastate} that results after measurement and parameter estimation. 
Define $\bar{c}$ as in~\eqref{eq:cbar}.
Then, for $\eps \in \big[0,\sqrt{\Pr[ F^{\mathrm{pe}} = \succ ]_{\sigma}}\big)$, we have
\begin{align}
  H_{\min}^{\eps}(X \land F^{\mathrm{pe}} =\succ|V W S^{\Pi} S^{\Xi} S^{\Theta} E)_{\sigma} + H_{\max}^{\eps}(X \land F^{\mathrm{pe}} =\succ|Y)_{\sigma}\geq (m-k) \log \frac{1}{\bar{c}} \,.
\end{align}
\end{corollary}

\begin{proof}
Consider the state $\tau_{A B V W S^{\Pi} S^{\Xi} S^{\Theta} F^{\mathrm{pe}} E \land F^{\textrm{pe}}=\succ}$ defined in~\eqref{eq:tautau} and note that it is of the form
\begin{align}
  \tau_{A B V W S^{\Pi} S^{\Xi} S^{\Theta} F^{\mathrm{pe}} E \land F^{\textrm{pe}}=\succ} = \tau_{A B V W S^{\Pi} S^{\Xi} F^{\mathrm{pe}} E \land F^{\textrm{pe}}=\succ} \otimes \rho_{S^\Theta} .
\end{align}
This is the state of the system after parameter estimation and after measuring $V$ and $W$, but with the measurement of $X$ and $Y$ (in the basis determined by $S^{\Theta}$) delayed. In particular we have used the fact that the register $S^{\Theta}$ has not yet been touched in the protocol, and is thus independent and uniform and independent even after we consider the event $F^{\textrm{pe}} = \succ$.

Let us now apply Proposition~\ref{pr:ur} to this state. For this purpose we equate $P = S^{\Pi} S^{\Xi}S^{\Theta}$, $R = VW E$, and $S = B$. 
The symmetry is determined by the map $q: \theta \mapsto \bar{\theta}$ with $\bar{\theta}_i = 1 - \theta_i$, which only acts on $S^{\Theta}$ and since this system is uniform and in product with the rest of the state trivially satisfies the symmetry condition of the theorem. The measurement map is then simply $\cM_{A \to X|S^{\Pi} S^{\Theta}}$ and we can calculate
\begin{align}
  c_q = \max_{\pi \in \Pi_{m,k}} \max_{\theta,x,z \in \{ 0,1\}^n} \Big\| M_{A_{\bar{\pi}}}^{\bar{\theta},x} \big(M_{A_{\bar{\pi}}}^{{\theta},z}\big)^{\dag} \Big\|_{\infty}^2 
  = \max_{\pi \in \Pi_{m,k}} \left( \prod_{i \in \bar{\pi}} c_i \right)
  = \bar{c}^n \,.
\end{align}
Proposition~\ref{pr:ur} applied to our setup thus yields
\begin{align}
  H_{\min}^{\eps}(X \land F^{\mathrm{pe}} = \succ|VW E S^{\Pi} S^{\Xi} S^{\Theta})_{\sigma} + H_{\max}^{\eps}(X \land F^{\mathrm{pe}} = \succ|B S^{\Pi} S^{\Xi} S^{\Theta})_{\tau} \geq n \log \frac{1}{\bar{c}}. \label{eq:ucr2}
\end{align}
Finally, the statement of the Proposition follows by applying the measurement map $\cM_{B \to Y|S^{\Pi}S^{\Theta}}$ (and discarding the seed registers) and noting that $H_{\max}^{\eps}(X \land F^{\mathrm{pe}} = \succ|B S^{\Pi} S^{\Xi} S^{\Theta})_{\tau} \leq H_{\max}^{\eps}(X \land F^{\mathrm{pe}} = \succ|Y)_{\sigma}$ by the data-processing inequality.
\end{proof}


\subsection{Parameter estimation: Statistical bounds on smooth max-entropy}

This section covers the necessary statistical analysis. This is essentially a variation of the analysis in~\cite{tomamichellim11}, but requires a new tool, Lemma~\ref{lm:likely}, presented in Section \ref{subsection:2.2}, as we are finding it clearer to do the analysis with sub-normalized states here. 
We use the following standard tail bound. 
\begin{lemma}
  \label{lm:stat}
  Consider a set of binary random variables $Z = (Z_1, Z_2, \ldots, Z_m)$ with $Z_i$ taking values in $\{0, 1\}$ and $m = n+k$. Let $\Pi \in \Pi_{m,k}$ be an independent, uniformly distributed random variable. Then,
  \begin{align}
    \Pr \Bigg[ \sum_{i \in \Pi} Z_i \leq k\delta \ \land\ \sum_{i \in \bar{\Pi}} Z_i \geq n (\delta+\nu) \Bigg] \leq e^{ - 2 \nu^2  \frac{n k^2}{(n+k) (k+1)} }.
  \end{align}
\end{lemma}
Remarkably this bound is valid without any assumption on the distribution of $Z$.

\begin{proof}
Let $\mu(z) = \frac{1}{m} \sum_{i \in [m]} z_i$. Consider the following sequence of inequalities:
  \begin{align}
    \Pr \Bigg[ \frac{1}{k} \sum_{i \in \Pi} Z_i \leq \delta\ \land\ \frac{1}{n} \sum_{i \notin \Pi} Z_i \geq \delta+\nu \Bigg] &\leq \Pr \Bigg[ \frac{1}{n} \sum_{i \in \bar{\Pi}} Z_i \geq \frac{1}{k} \sum_{i \in \Pi} Z_i +\nu \Bigg] \\
    &= \sum_{z \in \{0,1\}^{m} }\!\! \Pr[ Z = z ] 
      \Pr \Bigg[  \frac{1}{n} \sum_{i \in \bar{\Pi}} z_i \geq \frac{1}{k} \sum_{i \in \Pi} z_i+ \nu \Bigg] \\
        &= \sum_{z \in \{0,1\}^{m} }\!\! \Pr[ Z = z ] 
      \Pr \Bigg[ \frac{1}{n} \sum_{i \in \bar{\Pi}} z_i \geq \mu(z) + \frac{k \nu}{m} \Bigg] . \label{eq:backintohere}
  \end{align}
  Here, the first inequality holds since $A \implies B$ implies $\Pr[A] \leq \Pr[B]$ for any events $A$ and $B$. The first equality follows from the fact that $\Pi$ is independent of $Z$. The last equality follows by substituting $\sum_{i \in \Pi} z_i = m \mu(z) - \sum_{i \in \bar{\Pi}} z_i$ and rearranging the terms appropriately.
  
  Now note that the random sums $S_n := \sum_{i \in \bar{\Pi}} z_i$ can be seen as emanating from randomly sampling without replacement $n$ balls labelled by $z_i \in \{0, 1\}$ from a population $z$ with mean $\mu(z)$. Serfling's bound~\cite[Corollary~1.1]{serfling74} then tells us that
  \begin{align}
    \Pr \bigg[ \frac{1}{n} S_n \geq \mu(z) + \frac{k \nu}{m} \bigg] 
    &\leq e^{ - 2 n \bigg( \frac{k \nu}{m} \bigg)^2 \frac{1}{1 - f_n^*} } = e^{ - 2 \nu^2  \frac{n k^2}{(n+k) (k+1)} } \,.
  \end{align}
where we substituted $f_n^* = \frac{n-1}{m}$. It is important to note that this bound is independent of $\mu(z)$. Thus, substituting this back into~\eqref{eq:backintohere}, we conclude the proof.
\end{proof}

The following lemma ensures that disregarding an unlikely event will not disturb the state too much in terms of the purified distance.
\begin{lemma}
\label{lm:likely}
Let $\rho_{AX} \in \cSsub(AX)$ be classical on $X$ and $\Omega: \mathcal{X} \to \{0, 1\}$ an event with $\Pr[\Omega]_{\rho} = \eps < \tr\{\rho_{AX}\}$. Then there exists a sub-normalized state $\tilde{\rho}_{AX} \in \cSsub(AX)$ with $\Pr[\Omega]_{\tilde{\rho}} = 0$ and $P(\rho_{AX}, \tilde{\rho}_{AX}) \leq \sqrt{\eps}$.
\end{lemma}

\begin{proof}
  Set $\xi = \tr\{\rho_{AX}\}$. Let $\tilde{\rho}_{AX} = \frac{\sin^2(\phi)}{\xi - \eps} \rho_{AX \land \lnot \Omega}$ for some normalization $\phi \in [0, \frac{\pi}{2}]$ to be determined. Then the generalized fidelity evaluates to
  \begin{align}
    \sqrt{F\big(\rho_{AX}, \tilde{\rho}_{AX}\big)} &= \frac{\sin(\phi)}{\sqrt{\xi-\eps}} \tr \left\{ \sqrt{\sqrt{\rho_{AX}} \rho_{AX\land \lnot \Omega} \sqrt{\rho_{AX}}} \right\} + \cos(\phi) \sqrt{1-\xi} \\
    &= \sin(\phi) \sqrt{\xi - \eps} + \cos(\phi) \sqrt{1 - \xi} \,. \label{eq:substhere}
  \end{align}
  This expression is maximized for $\tan(\phi) = \sqrt{\frac{\xi - \eps}{1-\xi}}$ and, thus, $\sin(\phi) = \sqrt{\frac{\xi - \eps}{1 - \eps}}$ and $\cos(\phi) = \sqrt{\frac{1 - \xi}{1 - \eps}}$. Substituting this into~\eqref{eq:substhere} yields $F(\rho_{AX}, \tilde{\rho}_{AX}) = 1 - \eps$, concluding the proof.
\end{proof}

With this in hand, we wish to bound the smooth max-entropy of the state when passing the parameter estimation test.
\begin{proposition}
  \label{pr:smooth-max-bound}
  Consider the protocol \qkdeb{} in Section~\ref{sc:protocol} with $\mathrm{pe} = \{k, \delta\}$ applied to a state $\rho_{AB} \in \cS(AB)$ and the state $\sigma_{X Y F^{\textrm{pe}}}$ in~\eqref{eq:sigmastate} that results after measurement and parameter estimation. For any $\nu \in (0, 1)$, we first define
  \begin{align}
    \eps(\nu) := e^{-   \frac{(m-k) k^2 \nu^2}{m (k+1)} } \,. \label{eq:epsdef}
  \end{align}
  Then, for any $\nu \in (0, \frac12-\delta]$ such that $\eps(\nu)^2 < \Pr[ F^{\mathrm{pe}} =\, \succ]_{\sigma}$, the following holds:
\begin{align}
  \label{eq:thisstatement}
  H_{\max}^{\eps(\nu)}(X \land F^{\mathrm{pe}} = \succ|Y)_{\sigma} \leq (m - k)\, h(\delta + \nu) \qquad \textrm{where} \qquad h(x) := - x \log x - (1-x)\log(1-x) \,.
\end{align}
\end{proposition}

Intuitively, this result is a consequence of the fact that when we pass the parameter estimation test, conditioned on any particular value of $Y$, the support of $X$ is small as the number of errors (positions where $x_i \neq y_i$) is bounded (with high probability).

\begin{proof}
  We use the shorthand $p = \Pr[ F^{\mathrm{pe}} =\, \succ]_{\sigma}$ and $n := m - k$. Define the event event $\Omega := 1 \big\{\sum_{i \in [n]} 1\{ X_i \neq Y_i \} \geq n (\delta + \nu)\big\}$.
  We show that the statement in~\eqref{eq:thisstatement} holds when $p > \eps^{2}$. Using Lemma~\ref{lm:stat}, we find
\begin{align}
  &\Pr\bigg[ F^{\mathrm{pe}} = \succ  \land \Omega \bigg]_{\sigma} = \Pr\bigg[ \sum_{i \in [k]} 1\{ V_i \neq W_i \} \leq k \delta \ \land\ \sum_{i \in [n]} 1\{ X_i \neq Y_i \} \geq n (\delta + \nu) \bigg]_{\sigma} \leq \eps(\nu)^2 \label{eq:omegabound}.
\end{align}
This gives an upper bound on the probability of the unlikely coincidence where the parameter estimation test passes with threshold $\delta$ but the fraction of errors between $X$ and $Y$ exceeds the threshold $\delta$ by a constant amount.
We now want to remove the above unlikely events from our state $\sigma_{XYF^{\textrm{pe}} \land F^{\textrm{pe}} = \,\succ}$ by means of smoothing. Lemma~\ref{lm:likely} allows to do just that, and we introduce the state $\tilde{\sigma}_{XYF^{\textrm{pe}}}$ that is $\eps(\nu)$-close to $\sigma_{XYF^{\textrm{pe}} \land F^{\textrm{pe}} = \,\succ}$ in purified distance and satisfies $\Pr[\Omega]_{\tilde{\sigma}} = 0$.
From this we conclude that 
\begin{align}
  H_{\max}^{\eps(\nu)}(X \land F^{\mathrm{pe}} = \succ |Y)_{\sigma} \leq H_{\max}(X \land F^{\mathrm{pe}} = \succ |Y)_{\tilde{\sigma}} = H_{\max}(X|Y)_{\tilde{\sigma}} \,, \label{eq:bound-smooth-max-ent}
\end{align} 
where the last equality is a consequence of the fact that $\tilde{\sigma}$ is only supported on $F^{\mathrm{pe}} = \succ$.

It remains to show that $H_{\max}(X|Y)_{\tilde{\sigma}} \leq n h(\delta+\nu)$. Using the expansion of the conditional max-entropy in~\cite[Sec.~4.3.2]{mythesis}, we find
\begin{align}
  H_{\max}(X|Y)_{\tilde{\sigma}} &= \log \Bigg( \sum_{y \in \{0,1\}^n} \Pr[Y=y]_{\tilde{\sigma}} \ 2^{ H_{\max}(X|Y=y)_{\tilde{\sigma}} } \Bigg) \\
  &\leq \max_{y \in \{0,1\}^n \atop \Pr[Y = y]_{\tilde{\sigma}} > 0} H_{\max}(X|Y=y)_{\tilde{\sigma}} \\
  &\leq \max_{y \in \{0,1\}^n \atop \Pr[Y = y]_{\tilde{\sigma}} > 0} \log \Big| \Big\{ x \in \{0,1\}^n : \Pr[X=x|Y=y]_{\tilde{\sigma}} > 0 \Big\} \Big| \\
  &= \max_{y \in \{0,1\}^n} \log \Big| \Big\{ x \in \{0,1\}^n : \Pr[X=x \land Y=y]_{\tilde{\sigma}} > 0 \Big\} \,.
\end{align}
In the ultimate inequality we used that the (unconditional) R\'enyi entropy is upper bounded by the logarithm of the distribution's support~\cite{renyi61}. Furthermore, since $\Pr[\Omega]_{\tilde{\sigma}} = 0$, we have
  \begin{align}
    \Big| \Big\{ x \in \{0,1\}^n : \Pr[X=x \land Y=y]_{\tilde{\sigma}} > 0 \Big\} \Big| &\leq 
    \sum_{x \in \{0,1\}^n} 1 \Bigg\{ \sum_{i \in [n]} 1\{ x_i \neq y_i \} < n (\delta + \nu) \Bigg\} \\
    &= \sum_{e \in \{0,1\}^n} 1 \bigg\{ \sum_{i=1}^n e_i < n (\delta+\nu) \bigg\} \label{eq:comb1} \\
    &= \sum_{\lambda = 0}^n {n \choose \lambda} 1 \big\{ \lambda < n (\delta+\nu) \big\} = \sum_{\lambda = 0}^{\lfloor n (\delta + \nu) \rfloor} {n \choose \lambda}. \label{eq:comb2}
  \end{align}
  Here, in order to derive~\eqref{eq:comb1} we reparameterize $e_i = x_i \textrm{ xor } y_i$, indicating if there is an error at the $i$-th position. Finally, in~\eqref{eq:comb2} we substitute $\lambda = \sum_{i=1}^n e_i$, the total number of errors.
  The inequality $\sum_{\lambda = 0}^{\lfloor n (\delta + \nu) \rfloor} {n \choose \lambda} \leq 2^{n h(\delta+\nu)}$ for $\delta+\nu \leq 1/2$ (see, e.g.,~\cite[Sec.~1.4]{vanlint99}) then concludes the proof.
\end{proof}

\subsection{Privacy amplification: Proof of Theorem~\ref{th:secure}}
\label{sc:hash}
\label{sc:sec-last}

The last main ingredient of our proof is a so-called Leftover Hashing Lemma. It ensures that if the smooth min-entropy of $X$ given some side information $B$ is large, then we can extract randomness from $X$ that is independent of $B$.
The Leftover Hashing Lemma is, up to a slight change of the definition of the smooth min-entropy, due to Renner~\cite[Corollary~5.6.1]{renner05}. The proof of this exact statement is provided in Appendix~\ref{app:leftover} for the convenience of the reader.
\begin{proposition}
\label{pr:leftover}
Let $\sigma_{XD} \in \cSsub(XD)$ be classical on $X$ and $\eps \in \big[0, \sqrt{\tr(\sigma_{XD})}\big)$. Let $\cH$ be a universal$_2$ family of hash functions from $\mathcal{X} = \{0,1\}^n$ to $\mathcal{K} = \{0,1\}^\ell$. Moreover, let $\rho_{S^{H}} = \frac{1}{|\cH|} \sum_{h \in \cH} \proj{h}_{S^H}$. Then,
\begin{align}
\|\omega_{KS^HD} - \chi_K \otimes \omega_{S^HD}\|_{\tr} \leq \frac12 2^{-\frac{1}{2}\left( H_{\min}^{\eps}(X|D)_{\sigma} - \ell \right)} + 2 \eps
\end{align}
where $\chi_K = \frac{1}{2^\ell} \id_{K}$ is the fully mixed state and $\omega_{KS^HD} = \tr_{X}\big( \mathcal{E}_{f}(\sigma_{XD} \otimes \rho_{S^{H}}) \big)$ for the function $f: (x, h) \mapsto h(x)$ that acts on the registers $X$ and $S^{H}$.
\end{proposition}%

The following technical lemma allows us to bound the smooth conditional min-entropy restricted to events in terms of the unrestricted entropy.
\begin{lemma}
  \label{lm:conditioning}
  Let $\rho_{ABXY} \in \cSsub(ABXY)$ be classical on $X$ and $Y$ and let $\Omega: \mathcal{X} \times \mathcal{Y} \to \{0, 1\}$ be an event with $\Pr[\Omega]_\rho > 0$. Then, for $\eps \in \big[ 0, \sqrt{\Pr[\Omega]_{\rho}}\big)$, we have
  \begin{align}
    H_{\min}^{\eps}(AX \land \Omega |BY)_{\rho} \geq H_{\min}^{\eps}(AX|BY)_{\rho} 
    \quad \textrm{and} \quad
        H_{\min}^{\eps}(AX \land \Omega |B)_{\rho} \geq H_{\min}^{\eps}(AX|B)_{\rho} 
  \end{align}
\end{lemma}

\begin{proof}
  Let us start with the first inequality.
  By definition of the smooth conditional min-entropy there exists a sub-normalized state $\tilde{\rho}_{ABXY} \in \cSsub(ABXY)$ and a state $\sigma_{BY} \in \cS(BY)$ such that
  \begin{align}
    \tilde{\rho}_{ABXY} \leq 2^{-H_{\min}^{\eps}(AX|BY)_{\rho}} \id_{AX} \otimes \sigma_{BY} 
    \quad \textrm{and} \quad
    P(\tilde{\rho}_{ABXY}, \rho_{ABXY}) \leq \eps \,.
  \end{align}
  Without loss of generality~\cite[Lemma~6.6]{mybook} we can assume that $\tilde{\rho}_{ABXY}$ is classical on $X$ and $Y$. As such, we have
  $\tilde{\rho}_{ABXY \land \Omega} \leq \tilde{\rho}_{ABXY}$ and $P(\tilde{\rho}_{ABXY \land \Omega}, \rho_{ABXY \land \Omega}) \leq P(\tilde{\rho}_{ABXY}, \rho_{ABXY}) \leq \eps$ by the monotonicity of the purified distance under trace non-increasing maps.
  The desired inequality then follows by definition of the smooth min-entropy evaluated for the state with the event $\Omega$.

  The second inequality follows similarly. By definition of the smooth conditional min-entropy there exists a sub-normalized state $\tilde{\rho}_{ABX} \in \cSsub(ABX)$ and a state $\sigma_{B} \in \cS(B)$ such that
  \begin{align}
    \tilde{\rho}_{ABX} \leq 2^{-H_{\min}^{\eps}(AX|B)_{\rho}} \id_{AX} \otimes \sigma_{B} 
    \quad \textrm{and} \quad
    P(\tilde{\rho}_{ABX}, \rho_{ABX}) \leq \eps \,.
  \end{align}
  As discussed previously, this implies in particular the existence of an extension $\tilde{\rho}_{ABXY} \in \cSsub(ABXY)$ that satisfies $P(\tilde{\rho}_{ABXY}, \rho_{ABXY}) \leq \eps$. Without loss of generality we can assume that $\tilde{\rho}_{ABXY}$ is classical on $X$ and $Y$. (To see this, note that pinching in the computational basis on $Y$ would indeed only decrease the distance between the $\tilde{\rho}_{ABXY}$ and $\rho_{ABXY}$, leaving the latter state invariant.) On the state $\tilde{\rho}_{ABXY}$ we can now define the restriction on the event $\Omega$ and find
  \begin{align}
    \tilde{\rho}_{ABXY \land \Omega} \leq \tilde{\rho}_{ABXY} \implies \tilde{\rho}_{ABX \land \Omega} = \tr_Y\{ \tilde{\rho}_{ABXY \land \Omega}\} \leq \tilde{\rho}_{ABX} \,.
  \end{align}
  Finally, we proceed in the same fashion as for the first inequality to show that $P(\tilde{\rho}_{ABX \land \Omega}, \rho_{ABX \land \Omega}) \leq \eps$ and conclude the proof.
\end{proof}

The next proposition builds on Corollary~\ref{cor:smooth-min-bound} and Proposition~\ref{pr:smooth-max-bound} and the above Leftover Hashing Lemma to establish the secrecy of the key. 

\begin{proposition}
  \label{pr:secure}
  Let $\rho_{ABE} \in \cS(ABE)$. Consider the protocol \qkdeb{} in Section~\ref{sc:protocol} with $\mathrm{pe} = \{k, \delta\}$, $\mathrm{ec} = \{ t, r, \ldots \}$ and $\mathrm{pa} = \{ \ell, \ldots \}$ and the state $\omega_{K_A K_B S C F E} = \qkdeb(\rho_{ABE})$. Define $\eps(\nu)$ as in \eqref{eq:epsdef}.
  Then, for any $\nu \in (0, \frac12-\delta]$ such that $\eps(\nu)^2 < \Pr[F = (\succ, \succ)]_{\sigma}$, the following holds:
  \begin{align}
     \Big\| \omega_{K_A S C F E \land F =\,(\succ,\succ)} - \chi_{K_A} \otimes \omega_{SCFE \land F =\,(\succ,\succ)} \Big\|_{\mathrm{tr}}
     \leq  \frac12 2^{-\frac12 g(\nu)} + 2 \eps(\nu) \,.
  \end{align}
  where $g(\nu) := (m - k) \big( \log \frac{1}{\bar{c}} - h(\delta + \nu) \big) - r - t - \ell$.
\end{proposition}

Note that in the above proposition, in order to ensure that the smooth min-entropy is always well-defined, we restricted ourselves to the case where the success probabillity exceeds the squared smoothing parameter, i.e.\ we required that $\eps(\nu)^2 < \Pr[F =\,(\succ,\succ)]$. The case where the success probability is small will be handled separately in Corollary~\ref{cor:smallsuccess} below.

\begin{proof}
   We use  $n = m - k$. Since $\Pr[F = (\succ \succ)]_{\sigma} \leq \Pr[F^{\mathrm{pe}} = \succ]_{\sigma}$ the condition of Proposition~\ref{pr:smooth-max-bound} is satisfied and 
   we find that $H_{\max}^{\eps(\nu)}(X \land F^{\mathrm{pe}} = \succ|Y)_{\sigma} \leq n h(\delta + \nu)$ for the state $\sigma_{X Y V W S^{\Pi} S^{\Xi} S^{\Theta} F^{\mathrm{pe}} E}$ as in~\eqref{eq:sigmastate} that results after measurement and parameter estimation.
   Combining this with Corollary~\ref{cor:smooth-min-bound} yields
    \begin{align}
  H_{\min}^{\eps(\nu)}(X \land F^{\mathrm{pe}} = \succ |V W S^{\Pi} S^{\Xi} S^{\Theta} E)_{\sigma} \geq n q \,,
   \end{align}
  where we introduced the shorthand $q = \log \frac{1}{\bar{c}} - h(\delta + \nu)$.

  Our goal is to translate this in a condition on the state $\sigma_{X \hat{X} C^V C^Z C^T S^{\Pi} S^{\Xi} S^{\Theta} S^{H_{\mathrm{ec}}} F^{\textrm{pe}} F^{\textrm{ec}} E}$ as in~\eqref{eq:statebeforepa} that results after error correction. The following chain of inequalities holds:
   \begin{align}
   n q &\leq H_{\min}^{\eps(\nu)}(X \land F^{\mathrm{pe}} = \succ|S^{\Pi} S^{\Xi} S^{\Theta} C^V E)_{{\sigma}} \label{eq:min-ineq0} \\
      &\leq H_{\min}^{\eps(\nu)}(X \land F^{\mathrm{pe}} = \succ|S^{\Pi} S^{\Xi} S^{\Theta}  C^V C^Z E)_{{\sigma}} + r \label{eq:min-ineq1} \\
      &= H_{\min}^{\eps(\nu)}(X \land F^{\mathrm{pe}} = \succ|S^{\Pi} S^{\Xi} S^{\Theta} S^{H_{\mathrm{ec}}}  C^V C^Z E)_{{\sigma} \otimes \rho} + r\label{eq:min-ineq2} \\
       &\leq H_{\min}^{\eps(\nu)}(X \land F^{\rm pe} = \succ |S^{\Pi} S^{\Xi} S^{\Theta} S^{H_{\mathrm{ec}}} C^V C^Z C^T E)_{{\sigma}} + r + t \\
      &\leq H_{\min}^{\eps(\nu)}\big(X \land F^{\rm pe} = \succ \land F^{\rm ec} = \succ \big| S^{\Pi} S^{\Xi} S^{\Theta} S^{H_{\mathrm{ec}}} C^V C^Z C^T E \big)_{{\sigma}} + r + t . \label{eq:min-ineq3}
   \end{align}
   The first inequality follows by relabeling $V$ to $C^V$ and discarding $W$, an instance of the data-processing inequality.
 The transcript register $C^Z$ contains the syndrome sent from Alice to Bob and the inequality~\eqref{eq:min-ineq1} follows by the chain rule in~\eqref{eq:min-chain-rule}, and the fact that $\log |C^Z| = r$. The register $S^{H_{\mathrm{ec}}}$ in the state $\rho_{S^{H_{\mathrm{ec}}}}$ is independent of the other registers. The register $C^T$ contains the hash of the raw key $X$ of size $\log |C^T| = t$ leading to the penultimate inequality.
   In the last step we used Lemma~\ref{lm:conditioning}.

   Summarizing $S' = (S^{\Pi}, S^{\Xi}, S^{\Theta}, S^{H_{\mathrm{ec}}})$ as well as $C = (C^V, C^Z, C^T)$, and $F = (F^{\mathrm{pe}}, F^{\mathrm{ec}})$ as usual, we can thus more compactly write this as 
 \begin{align}
  H_{\min}^{\eps(\nu)}(X \land F = (\succ,\succ)\, | S' C E)_{{\sigma}} \geq nq - r - t \,.
 \end{align} 
 Proposition~\ref{pr:leftover} applied with this bound then immediately yields the desired inequality.
\end{proof}

The above proposition suffers from the assumption $\eps(\nu)^2 < \Pr[F = (\succ,\succ)]_{\omega}$. However, the other case can easily be dealt with.
Indeed, the trace distance $\big\| \omega_{K_A S C F E \land F =\,(\succ,\succ)} - \chi_{K_A} \otimes \omega_{SCFE \land F =\,(\succ,\succ)} \big\|_{\mathrm{tr}}$ is upper bounded by the the trace of both states, i.e.\ the probability $\Pr[F = (\succ,\succ)]_{\omega}$. The inequalities $\Pr[F = (\succ,\succ)]_{\sigma} \leq \eps(\nu)^2 \leq \eps(\nu)$ then yield the following corollary, which is equivalent to Theorem~\ref{th:secure}.

 \begin{corollary}
  \label{cor:smallsuccess}
  Consider the setup of Proposition~\ref{pr:secure}. For any $\nu \in (0, \frac12-\delta]$ the following holds:
  \begin{align}
     \Big\| \omega_{K_A S C F E \land F =\,(\succ,\succ)} - \chi_{K_A} \otimes \omega_{SCFE \land F =\,(\succ,\succ)} \Big\|_{\mathrm{tr}}
     \leq  \frac12 \sqrt{2^{- (m - k) \big( \log \frac{1}{\bar{c}} - h(\delta + \nu) \big) + r + t + \ell}} + 2 \eps(\nu) \,.
  \end{align}
\end{corollary}


\part{Prepare-and-measure protocol}
\label{part:pm}

\begin{table}[h!]
{\small
\begin{tabular}{ll}
\nc{$\mathcal{N}_{A\to B}$}{Quantum channel between Alice and Bob}
\nc{$\mathcal{P}_{\fail \to A |RS^{\Phi_{\! A}}}$}{Preparation map that returns a state in register $A$ depending on the settings $R, S^{\Phi_{\! A}}$}

\hline
\nc{$M$}{Number of states sent by Alice in the prepare-and-measure version}
\nc{$\Omega$}{Subset of $[M]$ for which Bob obtains a conclusive measurement result}  
\nc{$\Sigma$}{Subset of $m$ indices where Alice and Bob's settings agree and Bob obtained a conclusive outcome}  
\nc{ro}{Reordering map used in the sifting step.}
\hline

\nc{$R$}{Register for Alice's raw key in the prepare-and-measure protocol}
\nc{$U$}{Register for Bob's measurement results in the prepare-and-measure protocol}

\nc{$S^{\Phi_{\! A}}$}{Seed for the choice of Alice's measurement bases in the prepare-and-measure protocol}
\nc{$S^{\Phi_{\! B}}$}{Seed for the choice of Bob's measurement bases in the prepare-and-measure protocol}

\nc{$S$}{Register corresponding to all the seeds that Alice communicates to Bob after state distribution}
\nc{}{$S =(S^\Pi, S^\Xi, S^\Theta, S^{H_{\mathrm{pe}}},S^{H_{\mathrm{ec}}})$}

\nc{$F^{\mathrm{si}}$}{Flag for the sifting procedure in the prepare-and-measure protocol}
\nc{$F$}{Register corresponding to all the flags, $F = (F^{\mathrm{si}},F^{\mathrm{pe}},F^{\mathrm{ec}})$}
\nc{$C^\Omega, C^\Sigma$}{Transcripts of the registers $\Omega$ and $\Sigma$ sent during sifting}
\nc{$C$}{Register containing all the communication transcripts, $C = (C^\Omega, C^\Sigma,C^V, C^Z, C^T)$}
\end{tabular}
}
\caption{Additional nomenclature and notation used in Part~\ref{part:pm}. See also Table~\ref{tb:not-eb}}
\end{table}

\section{Formal description of the prepare-and-measure protocol}
\label{sec:model_real}
\label{sc:pm1}

Here we discuss a prepare-and-measure (PM) protocol for QKD, denoted \qkdpm{}, which is essentially equivalent to BB84~\cite{bb84}, and prove that its security follows from that of the entanglement-based protocol considered in Part \ref{part:eb}, provided that some \textit{additional assumptions} are made. 

Section~\ref{sc:modelpm} provides the details of the protocol described in Table \ref{tb:realistic}, for the steps where it differs from the entanglement-based protocol. We describe the additional assumptions on the preparation and measurement devices in \ref{sub:dev} and present a mathematical model of the protocol in \ref{sub:mod}.

\subsection{Additional assumptions on preparation and measurement devices}
\label{sub:dev} \label{sc:assumpm}

The physical equipment of Alice an Bob is modified compared to that of the entanglement-based protocol considered in Part \ref{part:eb}. Indeed, the main point of implementing a prepare-and-measure protocol is that it is no longer required for Alice and Bob to share an entangled state, a task that remains very challenging if the two parties are a few tens of kilometers apart, which is typical in realistic scenarios where one wants to distribute secret keys at the scale of a metropolitan area. In the prepare-and-measure setup Alice and Bob do not start with an entangled state but instead have access to a quantum channel from Alice to Bob.

The assumption on finite-dimensional quantum systems, sealed laboratories, random seeds and authenticated communication channel discussed in Section~\ref{sc:assumeb} still apply. The assumption of sealed laboratories in particular implies that the quantum channel between Alice and Bob models all quantum communication leaving Alice's lab. 
However, we will replace the assumption of commuting measurements, deterministic detection and the complementarity of Alice's measurements.

\paragraph{Alice's preparation:}

In every round, indexed by $i \in [M]$, Alice's preparation device takes two bits as input: $\phi \in \{0,1\}$ describing a basis choice and $x  \in \{0,1\}$ describing the bit value within each basis. It produces a quantum state $\rho_{A_i}^{\phi,x}$, ideally corresponding to one of the four BB84 states. The commuting measurement assumption for Alice is replaced with the requirement that these states do not depend on the preparations in previous or later rounds (which is already ensured by our notation).
Our next assumption is that the state $\rho_{A_i}^{\phi,x}$ does not leak any information about the basis choice $\phi$, i.e., we require that
\begin{align}
 \sum_{x \in \{0,1\}} \rho_{A_i}^{0,x} = \sum_{x \in \{0,1\}} \rho_{A_i}^{1,x} \,. \label{eq:pm-assum1} 
\end{align}
Moreover, instead of the complementarity assumption on the measurements, we require that the prepared states are sufficiently complementary. More precisely, let us define
\begin{align}
  c'\big(\{ \rho^{x} \}_x, \{ \sigma^y \}_y \big) := \max_{x,y} \left\| \sqrt{\rho^{x}} X^{-1} \sqrt{\sigma^{y}} \right\|_{\infty}^2 , \quad \textrm{for states satisfying} \quad \sum_x \rho^x = \sum_y \sigma^y := X \,.  \label{eq:cprimedef}
\end{align}
In case $X$ has not full support we take the generalized inverse (on its support) in the above definition. Our second assumption on Alice's preparation is that
\begin{align}
  c_i' := c'\Big( \big\{ \rho_{A_i}^{0,x} \big\}_x, \big\{ \rho_{A_i}^{1,y} \big\}_y \Big) \leq \bar{c}' \label{eq:pm-assum2}
\end{align}
for all $i \in [M]$ and some constant $\bar{c}' < 1$. The constant $\bar{c}'$ is closely related to the constant $\bar{c}$ that described the complementarity of Alice's measurement in the entanglement-based protocol, as we will see in Corollary~\ref{cor:virtual}.

In an ideal implementation of the BB84 protocol, the states $\rho^{\phi,r}$ would be single-qubit states given by 
\begin{align}
\rho^{0,0} = \proj{0}, \quad \rho^{0,1} = \proj{1}, \quad  \rho^{1,0} = \proj{+}, \quad \rho^{1,1} = \proj{-}. \label{eq:qubits}
\end{align}
These states obviously satisfy the assumption in~\eqref{eq:pm-assum1} and it is easy to verify that $\bar{c}' = \frac12$ is a valid bound.

We should note that our first assumption is rather strong and for instance does not allow us to assess the security of popular implementations of the BB84 protocol relying on a weak-coherent-state encoding. Since single-photon sources remain expensive and imperfect today, it is indeed tempting to encode each qubit with two polarization modes and replace single-photons by phase-randomized weak-coherent states\footnote{If the single-photon polarization qubit states are given by~\eqref{eq:qubits}, then the four encoded BB84 states are 
$\tau^{i,j} = e^{-\alpha^2} \bigoplus_{k=0}^\infty \frac{\alpha^{2k}}{k!} (\rho^{i,j})^{\otimes k}$ for $i,j \in \{0, 1\}$,
where $\alpha>0$ is the amplitude of the coherent states.}. 
For such an implementation, the four BB84 states become linearly independent, and Eq.~\eqref{eq:pm-assum1} cannot hold. It is well-known that such implementations are sensitive to ``photon-number-splitting'' attacks but that solutions exist to restore their security, for instance with the help of decoy states \cite{lo05}. While we believe that our framework could accommodate such modifications (see for instance~\cite{hasegawa07,lim14}), we do not address this issue here.

\paragraph{Bob's measurement:}
As for the simple protocol of Part \ref{part:eb}, we require that Bob's quantum system can be decomposed as $B \equiv B_{[M]} = B_1 B_2 \ldots B_M$. Bob's measurement device is similar to that of the simple protocol of Part \ref{part:eb}, but the measurement operators now need to be specified for indices in $[M]$ and allow for an additional outcome, `$\fail$', corresponding to an inconclusive result. Such an inconclusive result can for instance occur when no detector clicked (photon loss) or when more than 1 detector clicked (dark counts)\footnote{Indeed, in most experiments, the measurement device is usually implemented with the help of two single-photon detectors, and a conclusive measurement outcome, 0 or 1, will correspond to which detector clicked while inconclusive outcomes occur if none or both of the detectors clicked.}. 

For any $i \in [M]$, we model Bob's measurement on subsystem $B_i$ with setting $\phi \in \{0,1\}$ by a ternary generalized measurement $\{ M_{B_i}^{\phi,z}  \}_{z \in \{0,1,\fail\}}$ acting on $B_i$. The index $z$ ranges over the two conclusive outcomes, 0 and 1, of Bob's measurement as well as the inconclusive outcomes $\fail$.
We require that 
\begin{align}
\label{eq:pm-assum3}
(M_{B_i}^{0,\fail})^\dagger (M_{B_i}^{0,\fail}) =(M_{B_i}^{1,\fail})^\dagger (M_{B_i}^{1,\fail}), 
\end{align}
meaning that the element corresponding to an inconclusive result coincides for both measurements. 
As will be formalized in Lemma \ref{lm:sifting2}, this implies that Bob's measurement map can be interpreted as a two-step process first deciding whether the result is conclusive or not, and then, in the former case, proceeding with the ideal measurement considered in Part \ref{part:eb}.
While this assumption seems quite reasonable for photon detector working in the few photons regime, it usually fails to apply when avalanche photo diodes are accessed in the linear mode, and this was precisely the origin of the ``blinding attack'' of~\cite{lydersen10}.

In any realistic implementation, Alice and Bob would need to be synchronized so that Bob can keep track of which system he is currently measuring. This is especially relevant in high-loss regimes where Bob's detectors would not click most of the time. Such a synchronization procedure can be realized classically, provided that both players have access to an authenticated channel. For simplicity, we ignore this synchronization issue in our model.

\subsection{Protocol parameters and overview}
\label{sc:modelpm}

The protocol $\qkdpm$ is parametrized as $\qkdeb$, but with one extra parameter: 
\begin{itemize}
\item The number of individual states prepared and sent through the quantum channel by Alice, $M \in \mathbb{N}$. We require that $M \geq m$ and for optical implementations a typical value for $M$ is $\frac{2m}{\eta}$ where $\eta$ is the overall transmittance of the optical channel between Alice and Bob.
\end{itemize}

\begin{table}[t!]
\fbox{%
\begin{minipage}{\textwidth}
  
  { $(K_A, K_B, S, C, F) = \qkdpm \big( \cN_{A \to B} \big)$:}
       
  \begin{description}
  
    \item[Input:] Alice and Bob have access to a quantum channel $\cN_{A\to B}: A \to B$ where $A \equiv A_{[M]}$ and $B \equiv B_{[M]}$ are comprised of $M$ quantum systems.
    
       \item[Randomization:]
   Alice and Bob respectively choose two random strings $\Phi_{\! A}, \Phi_{\! B} \in \{0,1\}^M$. Alice also chooses a random string $R \in \{0,1\}^M$. These private seeds are denoted $S^{\Phi_{\! A}}, S^{\Phi_{\! B}}$ and $R$.
   Finally, similarly as in the simple protocol, Alice chooses a random subset $\Pi \in \Pi_{m,k}$, and random hash functions $H_{\mathrm{ec}} \in \cH_{\textrm{ec}}$ as well as $H_{\mathrm{pa}} \in \cH_{\textrm{pa}}$. These uniformly random seeds are denoted $S = (S^{\Pi}, S^{H_{\textrm{ec}}}, S^{H_{\textrm{pa}}})$.
   
    \item[State Preparation:] Alice prepares a quantum state $\rho_{A}^{R, \Phi_A}$, encoding the string $R$ in the measurement basis corresponding to $\phi_A$.
    
    \item[State Distribution:] Alice sends the state $\rho_{A}^{R, \Phi_A}$ through the quantum channel $\cN$ and Bob receives the output state $\rho_B^{R,\Phi_A} = \cN(\rho_{A}^{R, \Phi_A})$. 
    (In practice, Alice would send the systems one by one, and use the quantum channel $M$ times.)
        
    \item[Measurement:]
Bob measures the $M$ quantum systems with the setting $\Phi_{\! B}$, and stores his ternary measurement outcomes in a string $U \in \{0,1, \fail\}^M$, where $\fail$ denotes an inconclusive measurement result.  He also computes the set $\Omega \subseteq 2^{[M]}$ of indices corresponding to conclusive measurements
(In practice, Bob would start measuring the systems as soon as he receives them.)

    \item[Randomness distribution:]
Bob publicly announces both the value of $\Phi_{\! B}$ and $\Omega$. The corresponding transcripts are denoted by $C^{\Phi_{\! B}}$ and $C^{\Omega}$, respectively.
Alice sends the value of the seeds $S=(S^\Pi,  S^{H_{\mathrm{ec}}},S^{H_{\mathrm{pa}}})$ to Bob through the authenticated public channel.

    \item[Sifting:]
    If it exists, Alice publicly announces a set $\Sigma \subseteq \Omega$, with transcript $C^\Sigma$ of cardinality $m$, such that $\Phi_{\! A}$ and $\Phi_{\! B}$ coincide on $\Sigma$, and sets the flag $F^{\mathrm{si}} = \succ$. Otherwise, she sets $F^{\mathrm{si}} = \fail$ and they abort.
    The respective binary substrings $R'$ and $U'$ of $R$ and $U$ restricted to $\Sigma$ become the raw keys. 
     As in the idealized protocol, they are then reordered and denoted $(X, V)$ and $(Y, W)$ for Alice and Bob, respectively. Here $V, W$ are of length $k$ and correspond to the indices in $\Pi$, whereas $X, Y$ of length $n$ correspond to indices not in $\Pi$.

    \item[Parameter Estimation:]
    Alice sends $V$ to Bob, the transcript is denoted $C^V$. Bob compares $V$ and $W$. If the fraction of errors exceeds $\delta$, Bob sets the flag $F^{\mathrm{pe}} = $ `$\fail$' and they abort.  
    Otherwise he sets $F^{\mathrm{pe}} = $ `$\succ$' and they proceed.
 
   \item[Error Correction:]
  Alice sends the syndrome $Z = \mathrm{synd}(X)$ to Bob, with transcript $C^Z$. Bob computes $\hat{X} = \mathrm{corr}(Y, Z)$.
  
    Alice computes the hash $T = H_{\mathrm{ec}}(X)$ of length $t$ and sends it to Bob, with transcript $C^T$.
    Bob computes $H_{\mathrm{ec}}(\hat{X})$. If it differs from $T$, he sets the flag $F^{\mathrm{ec}} = $ `$\fail$' and they abort the protocol.  
    Otherwise he sets $F^{\mathrm{ec}} = $ `$\succ$' and they proceed.    
    
  \item[Privacy Amplification:] They compute keys $K_A = H_{\mathrm{pa}}(X)$ and $K_B = H_{\mathrm{pa}}(\hat{X})$ of length $\ell$.
  
  \item[Output:] The output of the protocol consists of the keys $K_A$ and $K_B$, the seeds $ (S^{\Phi_{\! B}}, S^{\Pi}, S^{H_{\textrm{ec}}}, S^{H_{\textrm{pa}}})$, the transcript $C = (C^\Omega, C^\Sigma, C^V, C^Z, C^H)$ and the flags $F = (F^{\mathrm{si}},F^{\mathrm{pe}}, F^{\mathrm{ec}})$.
  In case of abort, we assume that all registers are initialized to a predetermined value.
  
  \end{description}
\end{minipage}
}
  \caption{Realistic Prepare-and-Measure QKD Protocol $\qkdpm$. The precise mathematical model is described in Section~\ref{sec:model_real}. This protocol differs from the entanglement-based protocol in several points: in particular, the input now corresponds to the quantum channel $\cN$ between Alice and Bob.
  }
  \label{tb:realistic}
\end{table}

\subsection{Exact mathematical model of the protocol}
\label{sub:mod}

Here we describe in detail the mathematical model corresponding to the protocol in Table~\ref{tb:realistic}, for the steps where it differs from the simple protocol of Part \ref{part:eb}.

\paragraph*{Input:}

The realistic protocol $\qkdpm$ we consider is a prepare-and-measure protocol, and the role of the input is now played by an (arbitrary) quantum channel $\cN_{A\to B}$ between Alice and Bob. Here $A = A_{[M]}$ and $B = B_{[M]}$. This situation is depicted in Figure \ref{fig:PM-input}.

As before, we make the assumption that the input and output of the quantum channel are finite-dimensional. This arguably appears quite restrictive since any complete description of the optical channel would involve infinite-dimensional Fock spaces, with the idea that each of the $M$ systems prepared by Alice corresponds to two polarization modes for instance. However, we point out that we do not require any explicit upper bound on the dimension of the Hilbert spaces occurring in the protocol, and that any \emph{physical} state necessarily has a bounded energy, which means that it can be arbitrarily well approximated by a quantum state in a finite-dimensional Hilbert space of sufficiently large dimension.

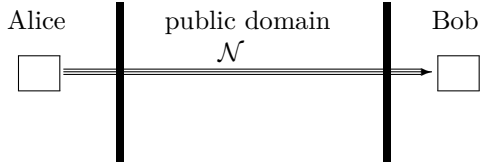
\begin{figure}[ht]
\centering
\begin{picture}(200, 60)
  \put(11, 50){Alice}
  \put(70, 50){public domain}
  \put(170, 50){Bob}
  
  \put(32, 34){\vector(1, 0){138}}
  \put(32, 33){\line(1, 0){134}}
  \put(32, 35){\line(1, 0){134}}

  \put(15, 30){\boxed{\phantom{X}}}
    \put(90, 38){$\mathcal{N}$}
  \put(172, 30){\boxed{\phantom{X}}}

  \linethickness{1mm}
  \put(50, 0){ \line(0, 1){60} }
  \put(150, 0){ \line(0, 1){60} }
\end{picture}
\caption{(Input.) Alice and Bob have access to a quantum channel: $\cN: A \to B$.}
\label{fig:PM-input}
\end{figure}

\paragraph*{Randomization:}

The random seeds are modeled similarly as for the idealized version of the protocol. Here, the seed $S^\Phi$ corresponding to identical measurement settings is not provided directly. Instead, Alice and Bob initially choose independently two strings $\Phi_{\! A}, \Phi_{\! B} \in \{0,1\}^M$, and it will later be the role of the sifting procedure to produce a set of identical measurement settings $\Phi$.
The random choice of the strings $\Phi_{\! A}, \Phi_{\! B}$ is modeled by two registers $S^{\Phi_{\! A}}, S^{\Phi_{\! B}}$ in the state
\begin{align}
\rho_{S^{\Phi_{\! A}}} \otimes \rho_{S^{\Phi_{\! B}}}  = \frac{1}{4^M} \sum_{\phi_A, \phi_B \in \{0,1\}^M} \proj{\phi_A}_{S^{\Phi_{\! A}}} \otimes \proj{\phi_B}_{S^{\Phi_{\! B}}},
\end{align}
where $\{|\phi_A\rangle \}, \{|\phi_B\rangle \}$ are orthonormal bases of $S^{\Phi_{\! A}}$ and $S^{\Phi_{\! B}}$, respectively.

Another difference with the simple protocol of Part \ref{part:eb} is that Alice also has access to a register $R$ that she will use to choose which state to prepare. This register is modeled similarly as the other seeds as a maximally mixed state:
\begin{align}
\rho_R = \frac{1}{2^M} \sum_{r \in \{0,1\}^M} \proj{r}_R,
\end{align}
where $\{|r\rangle\}$ is an orthonormal basis of $R$. 
The other random seeds $\rho_{S^\Pi}, \rho_{S^{H_{\mathrm{ec}}}}, \rho_{S^{H_{\mathrm{pa}}}}$ are identical to the idealized version. This situation is depicted in Figure \ref{fig:PM-rand}.

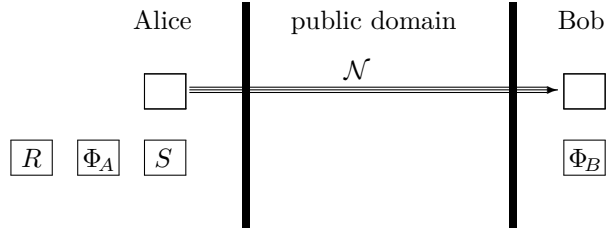
\begin{figure}[ht]
\centering
\begin{picture}(200, 90)
  \put(11, 75){Alice}
  \put(70, 75){public domain}
  \put(170, 75){Bob}
  
    \put(32, 52){\vector(1, 0){138}}
  \put(32, 51){\line(1, 0){134}}
  \put(32, 53){\line(1, 0){134}}

  \put(15, 48){\boxed{\phantom{X}}}
    \put(90, 56){$\mathcal{N}$}
  \put(172, 48){\boxed{\phantom{X}}}

  \put(-10, 23){\boxed{\phantom{X}}}
    \put(-8, 23){$\Phi_{\! A}$}

  \put(-35, 23){\boxed{\phantom{X}}}
    \put(-31, 23){$R$}

  \put(15, 48){\boxed{\phantom{X}}}
  \put(15, 23){\boxed{\phantom{X}}}
  \put(19, 23){$S$}
  \put(172, 48){\boxed{\phantom{X}}}
  \put(172, 23){\boxed{\phantom{X}}}
  \put(174, 23){$\Phi_{\! B}$}

  \linethickness{1mm}
  \put(50, 0){ \line(0, 1){85} }
  \put(150, 0){ \line(0, 1){85} }
\end{picture}
\caption{(Randomization.) Alice prepares a public seed $S$ that will later be communicated to Bob through an authenticated classical channel. Alice and Bob also prepare private seeds: $R, \Phi_{\! A}$ for Alice, $\Phi_{\! B}$ for Bob.}
\label{fig:PM-rand}
\end{figure}

Note that while in the simple protocol all the seeds are communicated publicly during a step of the protocol, it is crucially not the case here for the seed $R$, from which the final key could be immediately inferred. We will see that it is not necessary to communicate the seed $S^{\Phi_{\! A}}$ to Bob since the sifting procedure is performed by Alice. 

In a practical implementation, the various random seeds, except for $S^{\Phi_{\! B}}$ would be initially prepared by Alice, and only communicated to Bob when needed (except for $R$ and $S^{\Phi_{\! A}}$). In particular, one should wait until the state distribution is over before communicating the value of the chosen subset for parameter estimation or of the various hash functions. 

\paragraph*{State preparation:}

Alice prepares a quantum state on $M$ systems $A \equiv A_{[M]}$ using the map
\begin{align}
  \mathcal{P}_{\fail\to A|R S^{\Phi_{\! A}}}(\cdot) = 
      \sum_{r, \phi \in \{0,1\}^M}  \Big( \proj{r}_R \otimes \proj{\phi}_{S^{\Phi_{\! A}}} \Big) \cdot \Big( \proj{r}_R \otimes \proj{\phi}_{S^{\Phi_{\! A}}} \Big) \otimes \rho_{A}^{r,\phi}
\end{align}
where $\rho_{A}^{r, \phi} = \bigotimes_{i=1}^M \rho_{A_i}^{r_i, \phi_i}$. 
Applying this map to the seeds in registers $R$ and $S^{\Phi_{\! A}}$ results in the state
\begin{align}
\rho_{R S^{\Phi_{\! A}} A} = \frac{1}{4^M} \sum_{r, \phi \in \{0,1\}^M} \proj{r}_R \otimes \proj{\phi}_{S^{\Phi_{\! A}}} \otimes \rho_{A}^{r, \phi}.
\end{align}
This situation is depicted in Figure \ref{fig:PM-prep}.

\begin{figure}[ht]
\centering
\begin{picture}(200, 90)
  \put(11, 75){Alice}
  \put(70, 75){public domain}
  \put(170, 75){Bob}

    \put(32, 52){\vector(1, 0){138}}
  \put(32, 51){\line(1, 0){134}}
  \put(32, 53){\line(1, 0){134}}
  \put(15, 48){\boxed{\phantom{X}}}
    \put(90, 56){$\mathcal{N}$}
  \put(172, 48){\boxed{\phantom{X}}}
  \put(-10, 23){\boxed{\phantom{X}}}
    \put(-8, 23){$\Phi_{\! A}$}
  \put(-35, 48){\boxed{\phantom{X}}}
    \put(-31, 48){$R$}
  \put(-18, 52){\vector(1, 0){30}}
  \put(-8, 34){ \line(0, 1){18} }
  \put(-10, 56){$\mathcal{P}$}
  \put(15, 48){\boxed{\phantom{X}}}
  \put(19, 48){$A$}
  \put(15, 23){\boxed{\phantom{X}}}
  \put(19, 23){$S$}
  \put(172, 48){\boxed{\phantom{X}}}
  \put(172, 23){\boxed{\phantom{X}}}
  \put(174, 23){$\Phi_{\! B}$}
  \linethickness{1mm}
  \put(50, 0){ \line(0, 1){85} }
  \put(150, 0){ \line(0, 1){85} }
\end{picture}
\caption{(State Preparation.) Using her private seeds $R$ and $\Phi_{\! A}$, Alice prepares system $A$.}
\label{fig:PM-prep}
\end{figure}
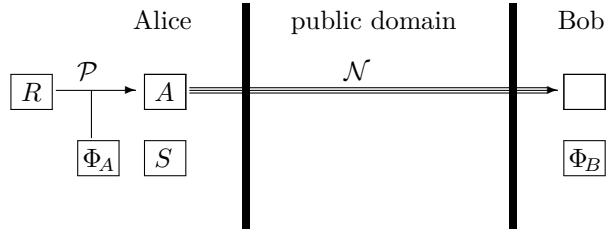

\paragraph*{State distribution:}

Alice sends her register $A$ to Bob through the quantum channel $\cN: A \to B$. 
The state shared by Alice and Bob is given by:
\begin{align}
\rho_{R S^{\Phi_{\! A}} B} &=  \cN_{A \to B}(\rho_{R S^{\Phi_{\! A}} A})\\
&=  \frac{1}{4^M} \sum_{r, \phi \in \{0,1\}^M} \proj{r}_R \otimes \proj{\phi}_{S^{\Phi_{\! A}}} \otimes \rho_{B}^{r, \phi},
\end{align}
where we defined $\rho_B^{r, \phi} = \cN\left(\rho_{A}^{r, \phi}\right)$.
This situation is depicted in Figure \ref{fig:PM-dist}.

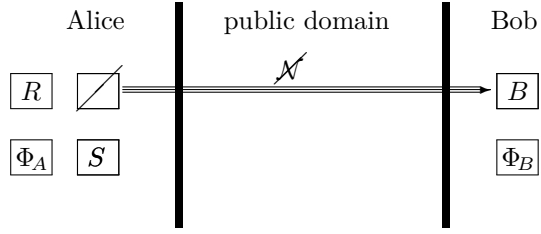
\begin{figure}[ht]
\centering
\begin{picture}(200, 90)
  \put(11, 75){Alice}
  \put(70, 75){public domain}
  \put(170, 75){Bob}
  
    \put(32, 52){\vector(1, 0){138}}
  \put(32, 51){\line(1, 0){134}}
  \put(32, 53){\line(1, 0){134}}
  \put(15, 48){\boxed{\phantom{X}}}
   \put(15,44){\line(1,1){17}}
    \put(90, 56){$\mathcal{N}$}
    \put(89,54){\line(1,1){12}}
  \put(172, 48){\boxed{\phantom{X}}}
  \put(-10, 23){\boxed{\phantom{X}}}
    \put(-8, 23){$\Phi_{\! A}$}
  \put(-10, 48){\boxed{\phantom{X}}}
    \put(-6, 48){$R$}
  \put(15, 48){\boxed{\phantom{X}}}
  \put(15, 23){\boxed{\phantom{X}}}
  \put(19, 23){$S$}

  \put(172, 48){\boxed{\phantom{X}}}
  \put(176, 48){$B$}
 \put(172, 23){\boxed{\phantom{X}}}
  \put(174, 23){$\Phi_{\! B}$}
    \put(15, 23){\boxed{\phantom{X}}}
  \put(19, 23){$S$}
  \linethickness{1mm}
  \put(50, 0){ \line(0, 1){85} }
  \put(150, 0){ \line(0, 1){85} }
\end{picture}
\caption{(State Distribution.) Alice sends the system $A$ through the quantum channel $\cN$ and Bob receives system $B$.}
\label{fig:PM-dist}
\end{figure}

\paragraph*{Measurement:}

\begin{figure}[ht]
\centering
\begin{picture}(200, 90)
  \put(11, 75){Alice}
  \put(70, 75){public domain}
  \put(170, 75){Bob}
  \put(15, 48){\boxed{\phantom{X}}}
   \put(15,44){\line(1,1){17}}
  \put(-10, 23){\boxed{\phantom{X}}}
    \put(-8, 23){$\Phi_{\! A}$}
  \put(-10, 48){\boxed{\phantom{X}}}
    \put(-6, 48){$R$}
  \put(15, 48){\boxed{\phantom{X}}}
  \put(15, 23){\boxed{\phantom{X}}}
  \put(19, 23){$S$}
   \put(15, 23){\boxed{\phantom{X}}}
  \put(19, 23){$S$}
  \put(172, 37){\boxed{\phantom{X}}}
  \put(176, 37){$\Omega$}
  \put(172, 13){\boxed{\phantom{X}}}
  \put(176, 13){$U$}
  \put(236, 27){\line(-1,0){35}}
  \put(201, 15){\line(0,1){24}}
 \put(201, 15){\vector(-1,0){11}}
 \put(201, 39){\vector(-1,0){11}}
  \put(236, 23){\boxed{\phantom{X}}}
  \put(239, 23){$B$}
  \put(234, 18){\line(1,1){17}}
  \put(207, 19){meas}
  \put(209, 48){\boxed{\phantom{X}}}
  \put(211, 48){$\Phi_{\! B}$}
 \put(217,45){\line(0, -1){18}}
  \linethickness{1mm}
  \put(50, 0){ \line(0, 1){85} }
  \put(150, 0){ \line(0, 1){85} }
\end{picture}
\caption{(Bob's Measurement.) Bob measures the quantum system $B$ using the measurement settings given by his private randomness $\Phi_{\! B}$ and obtains two classical strings stored in registers $\Omega$ and $U$.}
\label{fig:PM-meas}
\end{figure}
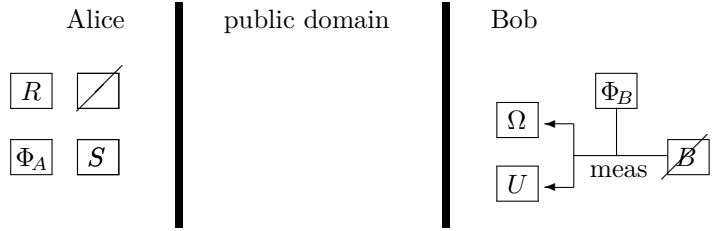

Bob measures each of his $M$ quantum systems in the basis corresponding to $\Phi_{\! B}$ and stores his measurement outcomes, either 0, 1, or $\fail$ in the case of inconclusive outcomes, in a classical register $U$ taking values in $\{0, 1, \fail\}^M$. The measurement map $\cM_{B\to U\Omega|S^{\Phi_{\! B}}}$ is defined as:
  \begin{align}
    \cM_{B\to U\Omega|S^{\Phi_{\! B}}}(\cdot) = \!\! \sum_{\phi \in \{0, 1\}^{M}} \sum_{ u \in \{0, 1, \fail\}^{M}} 
    \ket{u, \omega}_{TC^\Omega} 
      \Big( M_{B}^{\phi,u} \otimes  \proj{\phi}_{S^{\Phi_{\! B}}} \Big)\ \cdot\ \Big( M_{B}^{\phi,u} \otimes  \proj{\phi}_{S^{\Phi_{\! B}}} \Big)^{\dagger} \bra{u, \omega}_{UC^\Omega},
  \end{align}
  where $\omega=\omega(u)$ is the subset of $[M]$ where $u$ takes values in $\{0,1\}$, namely
  \begin{align}
  \omega(u) := \{ i \in [M] \: : \: u_i \neq \fail \}.
  \end{align}
The state of the total system after Bob's measurement is given by
\begin{align}
\sigma_{RUC^\Omega B S^{\Phi_{\! A}} S^{\Phi_{\! B}}} = \frac{1}{8^M}  \sum_{r, \phi_A, \phi_B \in \{0,1\}^M} \sum_{ u \in \{0, 1, \fail\}^{M}}  \proj{r,u,\omega, \phi_A, \phi_B}_{RUC^\Omega S^{\Phi_{\! A}}S^{\Phi_{\! B}}}\otimes M_{B}^{\phi_B,u} \rho_{B}^{r, \phi_A} \left(M_{B}^{\phi_B,u}\right)^\dagger,
\end{align}
and the situation is depicted in Figure \ref{fig:PM-meas}.

\paragraph*{Randomness distribution:}

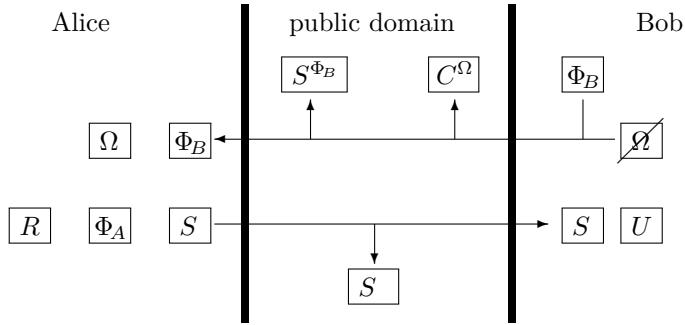
\begin{figure}[ht]
\centering
\begin{picture}(320, 120)
  \put(31, 110){Alice}
  \put(120, 110){public domain}
  \put(250, 110){Bob}
  \put(45, 65){\boxed{\phantom{X}}}
  \put(49, 65){$\Omega$}
  \put(75, 65){\boxed{\phantom{X}}}
  \put(77, 65){$\Phi_{\! B}$}
  \put(15, 33){\boxed{\phantom{X}}}
  \put(19, 33){$R$}
  \put(45, 33){\boxed{\phantom{X}}}
  \put(47, 33){$\Phi_{\! A}$}
  \put(75, 33){\boxed{\phantom{X}}}
  \put(79, 33){$S$}
  \put(222, 90){\boxed{\phantom{X}}}
  \put(224, 90){$\Phi_{\! B}$}
  \put(244, 65){\boxed{\phantom{X}}}
  \put(248, 65){$\Omega$}
      \put(243, 60){\line(1,1){17}}
  \put(222, 33){\boxed{\phantom{X}}}
  \put(226, 33){$S$}
  \put(244, 33){\boxed{\phantom{X}}}
  \put(248, 33){$U$}
  \put(152,37){\vector(0, -1){15}}
  \put(92, 37){\vector(1, 0){125}}
  \put(128,69){\vector(0, 1){15}}
  \put(182,69){\vector(0, 1){15}}
    \put(230,69){\line(0, 1){15}}
  \put(117, 90){\boxed{\phantom{Xx~}}}
  \put(121, 90){$S^{\Phi_{\!B}}$}
  \put(172, 90){\boxed{\phantom{X~}}}
  \put(175, 90){$C^\Omega$}
  \put(242, 69){\vector(-1, 0){150}}
  \put(142, 10){\boxed{\phantom{X~}}}
  \put(146, 10){$S$}
 \linethickness{1mm}
  \put(100, 0){ \line(0, 1){120} }
  \put(200, 0){ \line(0, 1){120} }
\end{picture}
\caption{(Randomness distribution.) Bob communicates the both the value of $\Phi_{\! B}$ and $\Omega$ to Alice with the authenticated classical channel. Alice publicly announces the value of the various seeds, $S =(S^\Pi,  S^{H_{\mathrm{ec}}},S^{H_{\mathrm{pa}}})$.
}
\label{fig:PM-rand-dist}
\end{figure}

Bob publicly announces the content of the register $S^{\Phi_{\! B}}$ together with the description $C^\Omega$ of the set $\omega$ of indices corresponding to conclusive measurement results. This situation is depicted in Figure~\ref{fig:PM-rand-dist}.
Alice publicly announces the value of the various seeds, $S =(S^\Pi, S^{H_{\mathrm{ec}}},S^{H_{\mathrm{pa}}})$.

\paragraph*{Sifting:}

 Alice applies the \textit{sifting map}, a classical map `sift' defined as follows
 \begin{align}
       \textrm{sift}: \left\{
       \begin{array}{ccc}
       \{0,1 \}^M \times  \{0,1 \}^M \times 2^{[M]} &  \to & \Pi_{M,m}  \times \{ \fail, \succ\}  \\
       (\phi_{\! A}, \phi_{\! B}, \omega ) & \mapsto & (\Sigma, F^{\mathrm{si}})
       \end{array}
       \right.
     \end{align}
  where $\Sigma$ is either the first subset of $\Omega$ of cardinality $m$ in the lexicographic order where $\phi_{\! A}$ and $\phi_{\! B}$ coincide, if such a set exists, or else it is set to a dummy value, for instance $[m]$. In the first case, the flag $F^{\mathrm{si}}$ is set to $\succ$, otherwise it is set to $\fail$ and the protocol aborts. The output of the sifting map immediately allows Alice and Bob to compute the value of the seed $S^{\Phi}$ which is simply the restriction of either $S^{\Phi_{\! A}}$ or $S^{\Phi_{\! B}}$ to the indices in~$\Sigma$.
  This classical map is lifted to give a CPTP map $\cE_{\mathrm{si}}= S^{\Phi_{\! A}} S^{\Phi_{\! B}} C^\Omega \to C^\Sigma C^\Omega F^{\mathrm{si}} S^{\Phi_{\! A}} S^{\Phi_{\! B}}$. This situation is depicted in Figure \ref{fig:PM-sift1}.
  
 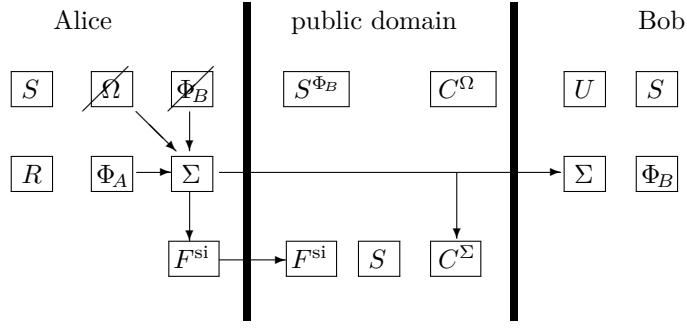
\begin{figure}[h]
\centering
\begin{picture}(320, 120)
  \put(31, 110){Alice}
  \put(120, 110){public domain}
  \put(250, 110){Bob}
  \put(15, 84){\boxed{\phantom{X}}}
  \put(19, 84){$S$}
  \put(45, 84){\boxed{\phantom{X}}}
  \put(49, 84){$\Omega$}
     \put(42, 79){\line(1,1){17}}
  \put(75, 84){\boxed{\phantom{X}}}
  \put(77, 84){$\Phi_{\! B}$}
 \put(74, 79){\line(1,1){17}}
  \put(15, 52){\boxed{\phantom{X}}}
  \put(19, 52){$R$}
  \put(45, 52){\boxed{\phantom{X}}}
  \put(47, 52){$\Phi_{\! A}$}
  \put(75, 52){\boxed{\phantom{X}}}
  \put(79, 52){$\Sigma$}
  \put(62, 56){\vector(1,0){13}}
  \put(62, 79){\vector(1,-1){15}}
  \put(82,79){\vector(0, -1){15}}
 \put(74, 20){\boxed{\phantom{X~}}}
  \put(76, 20){$F^{\mathrm{si}}$}
  
   \put(118, 20){\boxed{\phantom{X~}}}
  \put(120, 20){$F^{\mathrm{si}}$}
  \put(93,23){\vector(1,0){25}}
  
  \put(82,49){\vector(0,-1){19}}
  \put(222, 84){\boxed{\phantom{X}}}
  \put(226, 84){$U$}
  \put(222, 52){\boxed{\phantom{X}}}
  \put(226, 52){$\Sigma$}
  \put(249, 84){\boxed{\phantom{X}}}
  \put(253, 84){$S$}
  
    \put(145, 20){\boxed{\phantom{X}}}
  \put(149, 20){$S$}
  
  \put(249, 52){\boxed{\phantom{X}}}
  \put(251, 52){$\Phi_{\! B}$}

  \put(182,56){\vector(0, -1){25}}
  \put(93, 56){\vector(1, 0){128}}
  \put(172, 20){\boxed{\phantom{X~}}}
  \put(175, 20){$C^\Sigma$}
  \put(117, 84){\boxed{\phantom{Xx~}}}
  \put(121, 84){$S^{\Phi_{\!B}}$}
  \put(172, 84){\boxed{\phantom{Xx~}}}
  \put(175, 84){$C^{\Omega}$}
 \linethickness{1mm}
  \put(100, 0){ \line(0, 1){120} }
  \put(200, 0){ \line(0, 1){120} }
\end{picture}
\caption{(Sifting $1/2$.) Using her private randomness $\Phi_{\! A}$ together with the registers $\Omega$ and $\Phi_{\! B}$ sent by Bob, Alice computes the set $\Sigma$ when it exists and sets the value of the flag $F^{\mathrm{si}}$. Both values are then publicly announced.}
\label{fig:PM-sift1}
\end{figure}

We then define a CPTP map $\mathcal{E}_{\mathrm{di}}$ that discards the systems in $R$, $U$ and $\Phi_{\! A}$ which do not correspond to the subset $\Sigma$ and put the remaining systems in registers denoted $R'$, $U'$ and $\Phi$ of size $m$, respectively.
  \begin{align}
  \mathcal{E}_{\mathrm{di}} : C^\Sigma   R U  S^{\Phi_{\! A}} \to R' U' C^\Phi C^\Sigma .
  \end{align}    
 This situation is depicted in Figure \ref{fig:PM-sift2}.

\begin{figure}[h]
\centering
\begin{picture}(320, 120)
  \put(31, 110){Alice}
  \put(120, 110){public domain}
  \put(250, 110){Bob}

  \put(18, 43){\boxed{\phantom{X}}}
  \put(22, 43){$S$}

  \put(50, 43){\boxed{\phantom{X}}}
  \put(54, 43){$\Sigma$}
   \put(47, 38){\line(1,1){17}}
  \put(57,54){\line(0, 1){24}}

 \put(74, 43){\boxed{\phantom{X~}}}
  \put(76, 43){$F^{\mathrm{si}}$}
  \put(76, 38){\line(1,1){17}}

  \put(18, 20){\boxed{\phantom{X}}}
  \put(20, 20){$\Phi_{\! A}$}
   \put(15, 15){\line(1,1){17}}

    \put(35, 22){\vector(1,0){40}}
  \put(57,22){\line(0, 1){17}}

  \put(75, 20){\boxed{\phantom{X}}}
  \put(79, 20){$\Phi$}

  \put(222, 20){\boxed{\phantom{X}}}
  \put(226, 20){$\Phi$}

  \put(75, 75){\boxed{\phantom{X}}}
  \put(79, 75){$R$}
 \put(74, 70){\line(1,1){17}}

  \put(18, 75){\boxed{\phantom{X}}}
  \put(22, 75){$R'$}
  \put(75, 78){\vector(-1,0){40}}

  \put(222, 43){\boxed{\phantom{X}}}
  \put(226, 43){$S$}

  \put(222, 75){\boxed{\phantom{X}}}
  \put(226, 75){$U$}
  \put(221, 70){\line(1,1){17}}

  \put(254, 43){$\Sigma$}
  \put(250, 43){\boxed{\phantom{X}}}
   \put(247, 38){\line(1,1){17}}
  \put(282, 75){\boxed{\phantom{X}}}
  \put(286, 75){$U'$}

  \put(282, 20){\boxed{\phantom{X~}}}
  \put(286, 20){$\Phi_{\! B}$}
    \put(279, 15){\line(1,1){17}}

  \put(238, 78){\vector(1,0){43}}
  \put(256,54){\line(0, 1){24}}

  \put(256,22){\line(0, 1){17}}
  \put(281, 22){\vector(-1,0){43}}

 \put(118, 43){\boxed{\phantom{X~}}}
  \put(120, 43){$F^{\mathrm{si}}$}

  \put(172, 75){\boxed{\phantom{X~}}}
  \put(174, 75){$C^\Omega$}

  \put(146, 43){\boxed{\phantom{X}}}
  \put(150, 43){$S$}
  \put(146, 75){\boxed{\phantom{X~}}}
  \put(148, 75){$C^{\Sigma}$}
  \put(117, 75){\boxed{\phantom{Xx~}}}
  \put(121, 75){$S^{\Phi_{\! B}}$}
 \linethickness{1mm}
  \put(100, 0){ \line(0, 1){120} }
  \put(200, 0){ \line(0, 1){120} }
\end{picture}
\caption{(Sifting $2/2$.) Using $\Sigma$, Alice (resp.~Bob) discards the $M-m$ irrelevant systems of $R$ (resp.~$U$) and stores the $m$ remaining ones in $R'$ (resp.~$U'$).  Both Alice and Bob further compute $\Phi$ from $\Sigma$ and either $\Phi_{\!A}$ or $\Phi_{\! B}$.}
\label{fig:PM-sift2}
\end{figure}
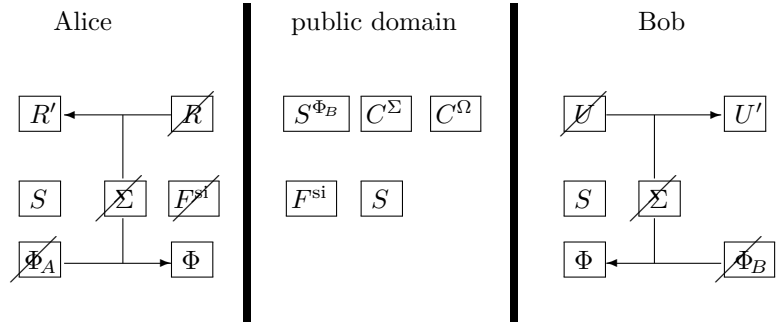

Finally, similarly as in the simple protocol of Part \ref{part:eb}, Alice and Bob use the content of register $\Pi$ to reorder their raw keys, which become $(V,X)$ for Alice and $(W,Y)$ for Bob. 
The situation here is similar to that obtained after measurement in the simple protocol (compare Figures \ref{fig:reorder} and \ref{fig:state_meas}), with the addition of the registers $C^\Sigma$, $C^\Omega$, $S^{\Phi_{\! B}}$ and $F^{\mathrm{si}}$ now available in the public domain.

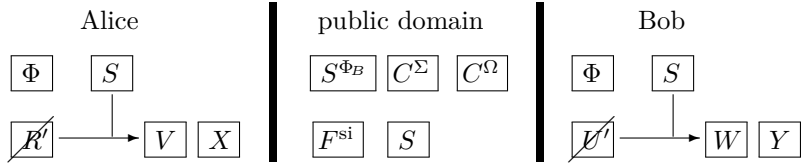
\begin{figure}[t]
\centering
\begin{picture}(300, 60)
  \put(31, 50){Alice}
  \put(120, 50){public domain}
  \put(240, 50){Bob}
  
  \put(35, 30){\boxed{\phantom{X}}}
  \put(39, 30){$S$}
  \put(5, 5){\boxed{\phantom{X}}}
  \put(9, 5){$R'$}
  \put(5, 30){\boxed{\phantom{X}}}
  \put(9, 30){$\Phi$}
  \put(43,25.5){\line(0, -1){16}}
  \put(4, 0){\line(1,1){17}}
  \put(23, 9){\vector(1, 0){30}}
  \put(55, 5){\boxed{\phantom{X}}}
  \put(59, 5){$V$}
  \put(75, 5){\boxed{\phantom{X}}}
  \put(79, 5){$X$}
  \put(245, 30){\boxed{\phantom{X}}}
  \put(249, 30){$S$}
  \put(215, 5){\boxed{\phantom{X}}}
  \put(219, 5){$U'$}
  \put(253,25.5){\line(0, -1){16}}
  \put(214, 0){\line(1,1){17}}
  \put(233, 9){\vector(1, 0){30}}
  \put(265, 5){\boxed{\phantom{X}}}
  \put(268, 5){$W$}
  \put(285, 5){\boxed{\phantom{X}}}
  \put(289, 5){$Y$}
 \put(215, 30){\boxed{\phantom{X}}}
  \put(219, 30){$\Phi$}
 \put(118, 5){\boxed{\phantom{X~}}}
  \put(120, 5){$F^{\mathrm{si}}$}

  \put(172, 30){\boxed{\phantom{X~}}}
  \put(174, 30){$C^\Omega$}

  \put(146, 5){\boxed{\phantom{X}}}
  \put(150, 5){$S$}
  \put(146, 30){\boxed{\phantom{X~}}}
  \put(148, 30){$C^{\Sigma}$}
  \put(117, 30){\boxed{\phantom{Xx~}}}
  \put(121, 30){$S^{\Phi_{\! B}}$}
  \linethickness{1mm}
  \put(100, 0){ \line(0, 1){60} }
  \put(200, 0){ \line(0, 1){60} }
\end{picture}
\caption{(Reordering.) Using the content of $S^{\Pi}$, Alice reorder their raw keys, $R'$ and $U'$, which become respectively $(X, V)$ and $(Y, W)$.}
\label{fig:reorder}
\end{figure}

\paragraph*{Remaining steps:}

The remaining steps are as in the entanglement-based QKD protocol presented in Part \ref{part:eb}.

\section{Results and Discussion}
\label{sec:results-pm}

The security proof should establish that for any input channel $\cN_{A\to B}$ given to $\qkdpm$, either the protocol outputs secret identical keys, or else it aborts. In the same spirit as the entanglement-based version of Part \ref{part:eb}, we define the security parameter
\begin{align}
  \Delta_{M,\mathrm{si},\mathrm{pe}, \mathrm{ec},\mathrm{pa}} := \sup_{\cN_{A\to BE}} 
   \Big\| \qkdpm(\cN_{A\to BE}) - \texttt{qkd\_ideal}_{M, k,n,\delta,\mathrm{ec},\mathrm{pa}}(\cN_{A\to BE})  \Big\|_{\mathrm{tr}}  , \label{eq:real-sec}
\end{align}
where again $\texttt{qkd\_ideal}_{M, k,n,\delta,\mathrm{si},\mathrm{ec},\mathrm{pa}}$ is defined analogously to the entanglement-based case and simply replaces the output of $\qkdpm(\cN_{A\to BE})$ with a perfect key in case the protocol does not abort. Here, the channels $\cN_{A\to BE}$ have an additional output that goes to an eavesdropper, and it again suffices to consider maps where $E$ is finite-dimensional.
Establishing security thus boils down to showing that this trace distance is small for all such channels.

Our strategy is to show that the realistic protocol is equivalent to applying the idealized QKD protocol on a virtual quantum state $\rho_{AB}$ independent from the uniformly distributed random seed $S^{\Phi}$ for the measurement settings. If this holds, then the security proof of Part \ref{part:eb} for the simple protocol applies, and establishes the security of the prepare-and-measure protocol. For this, we need to make explicit assumptions about (i) the state preparation on Alice's side to make sure that no basis information is leaked and (ii) the measurement device on Bob's side to ensure that the invalid measurement results do not depend on the measurement basis.

Under the assumptions in Section \ref{sc:assumpm}, we show that the prepare-and-measure QKD protocol is secure. 
\begin{theorem}
\label{thm:real}
Let $m, \mathrm{pe}, \mathrm{ec}, \mathrm{pa}$ be such that the protocol \qkdeb{} in Section~\ref{sc:protocol} is $\eps$-secure with device parameter $\bar{c} = \bar{c}'$. Then $\qkdpm$ is also $\eps$-secure. 
\end{theorem}

As will be shown in Section~\ref{sec:proofpm}, the security of the prepare-and-measure protocol is a consequence of that of the simple protocol studied in Part \ref{part:eb}, provided that some additional assumptions are made. 

When assessing the performance of the protocol, however, two modifications appear. First, the device parameter ($\bar{c}'$ instead of $\bar{c}$) needs to be defined differently since it is no longer a function of the measurement device of Alice, but rather of her preparation device. In an ideal implementation, it is still expected to be equal to $\frac{1}{2}$, as was discussed in Section \ref{sc:assumpm}. 
The more important difference is due to the sifting procedure. Indeed, the definition of the \emph{secret key rate} should now be modified to mean the ration between the key length $\ell$ and the number $M$ of individual states prepared and sent by Alice (instead of the number $m$ in the simple entanglement-based protocol). The means that the secret key rate achieved with the prepare-and-measure protocol is given by
\begin{align}
\frac{\ell}{M} = \frac{m}{M} \cdot \frac{\ell}{m}.
\end{align}
The sifting procedure that we have considered here (and described in Section \ref{sub:mod}) is not optimized to maximize the secret key rate (or equivalently the ratio $\frac{m}{M}$), but rather to simplify the analysis as much as possible. Better sifting procedures are discussed in \cite{pfister15} and could involve not fixing the value of $M$ in advance for instance. 

A typical experiment would be characterized by a given overall transmittance $\eta$ of the optical channel, meaning that approximately $\eta M$ photons will be detected by Bob, or in other words, that the expected value of $|\Omega|$ is $\eta M$.  Given that Alice and Bob's measurement bases will coincide on expectation $50\%$ of the time, we therefore expect that 
\begin{align}
\frac{m}{M} \approx \frac{\eta}{2}
\end{align}
holds asymptotically. 

In particular, the secret key rate of the prepare-and-measure protocol is then expected to be equal to $\frac{\eta}{2}$ times the secret key rate of the simple protocol displayed in Figure~\ref{fig:keylength}.


\section{Security reduction}
\label{sec:proofpm}

In Sections~\ref{sec:ass1} and~\ref{sec:ass2}, we discuss some implications of the device assumptions made in Section~\ref{sc:assumpm}. 
The security reduction to the simple protocol will be addressed in Section \ref{sc:pm2}.

\subsection{Preparation: Assumptions on Alice's device}
\label{sec:ass1}

Consider the four states $\{ \rho_{A_i}^{x,\phi} \}_{x,\phi \in \{0,1\} }$ created by Alice in round $i$ of the protocol for some $i \in [M]$. Since these states adhere to the assumptions stated in~\eqref{eq:pm-assum1} and~\eqref{eq:pm-assum2}, the following lemma is applicable:
\begin{lemma}
\label{lm:virtual}
  Let $\{ \rho_A^{\phi,x} \}_{\phi,x} \subset \cS(A)$ where $x$ and $\phi$ are taken from discrete sets. Moreover, let $\{ p_{x}^{\phi} \}_x$ be a probability distribution for each $\phi$ such that \begin{align}
    \sum_{x} p_{x}^\phi \rho_A^{\phi,x} = \sum_{x} p_x^{\phi'} \rho_A^{\phi',x} \quad \textrm{for all $\phi$ and $\phi'$} .
  \end{align}
  Then there exists a state $\tau_{AA'} \in \cS(AA')$ where $A' \equiv A$ and a generalized measurement $\{ M_{A'}^{\phi,x} \}_{x}$ on $A'$ for each $\phi$ such that
  \begin{align}
    p^{\phi}_x \rho_{A}^{\phi,x} = \tr_{A'}\Big[ \tau_{AA'} \, \id_A \otimes \big( M_{A'}^{\phi,x} \big)^{\dag}  M_{A'}^{\phi,x} \Big] \quad \textrm{and} \quad
    c\Big( \big\{ M_{A'}^{\phi,x} \big\}_x, \big\{ M_{A'}^{\phi',x} \big\}_x \Big) = c'\Big( \{ \rho_A^{\phi,x} \}_x, \{ \rho_A^{\phi',x} \}_{x} \Big) .  \label{eq:conditions}
  \end{align}
\end{lemma}

\begin{proof}
  We will explicitly construct the state and measurement as follows. First, let us introduce $\tau_A := \sum_{x} p_x^{\phi} \rho_A^{\phi,x}$ and choose $\tau_{AA'}$ as its purification on $A'$. Now we choose 
  \begin{align}
    M_{A'}^{\phi,x} := \sqrt{p_{x}^{\phi}} \left( \big( \rho_A^{\phi,x} \big)^{\frac{1}{2}} \right)^T \left( \tau_{A}^{-\frac{1}{2}} \right)^T 
  \end{align}
  where the transpose is taken with regards to the Schmidt basis of $\tau_{AA'}$. Let us first verify that this constitutes a generalized measurement. Indeed, for every $\phi$, we find
  \begin{align}
    \sum_x \big( M_{A'}^{\phi,x} \big)^{\dag} M_{A'}^{\phi,x} = \sum_x p_x^{\phi} \left( \tau_{A}^{-\frac{1}{2}} \right)^T  \left( \rho_A^{\phi,x} \right)^{T} \left( \tau_{A}^{-\frac{1}{2}} \right)^T = \id_{A'} \,.
  \end{align}
  Let us now verify the conditions in~\eqref{eq:conditions}. Since $\tau_{AA'} = | \tau_{AA'} \rangle\langle \tau_{AA'} |$ purifies $\tau_A$, we find
  \begin{align}
     \id_A \otimes M_{A'}^{\phi,x} \, \big| \tau_{AA'} \big\rangle = \sqrt{p_{x}^{\phi}}\ \id_A \otimes \left( \big( \rho_A^{\phi,x} \big)^{\frac{1}{2}} \right)^T \left( \tau_{A}^{-\frac{1}{2}} \right)^T  \big| \tau_{AA'} \big\rangle = \sqrt{p_{x}^{\phi}}\  \big| \rho_{AA'}^{\phi,x} \big\rangle \,,
  \end{align}
  where $\rho_{AA'}^{\phi,x} = | \rho_{AA'}^{\phi,x} \rangle\!\langle \rho_{AA'}^{\phi,x} |$ purifies $\rho_{A}^{\phi,x}$. The first equality readily follows. The second equality can be confirmed by consulting the definitions of $c$ and $c'$ in~\eqref{eq:cdef} and~\eqref{eq:cprimedef}, respectively.
\end{proof}

These assumptions on Alice's preparation allow us to replace the state preparation by a measurement performed on a virtual extension of the average prepared state.

\begin{corollary}
\label{cor:virtual}
If the two assumptions in~\eqref{eq:pm-assum1} and~\eqref{eq:pm-assum2} on Alice's preparation device hold, then
for each $i \in [M]$ there exists a state $\rho_{A_iA_i'}$ where $A_i' \equiv A_i$ and generalized measurements on $A_i'$ of the form prescribed in Section~\ref{sec:assumpt} such that $c_i \leq c_i'$, and, in particular, $\bar{c} \leq \bar{c}'$. Combining all these measurements into a map $\mathcal{M}_{A'\to R|S^{\Phi_{\! A}}}$ acting on $A' \equiv A_{[M]}'$, we further have
\begin{align}
\rho_{A R \Phi_A} = \mathcal{M}_{A'\to R|S^{\Phi_{\! A}}}\left(\rho_{AA'} \otimes \rho_{\Phi_{\! A}} \right) .
\end{align}
\end{corollary}

\subsection{Assumption on Bob's measurement:}
\label{sec:ass2}

The objective of the assumption made on Bob's measurement is to show that when the sifting procedure succeeds, the resulting register $S^\Phi$ is independent of the state shared by Alice and Bob, similarly as in the entanglement-based protocol of Part \ref{part:eb}.

\begin{lemma}
\label{lm:sifting2}
If Bob's measurement satisfies \eqref{eq:pm-assum3}, that is, $(M_{B_i}^{0,\fail} )^\dagger (M_{B_i}^{0,\fail} )= (M_{B_i}^{1,\fail} )^\dagger (M_{B_i}^{1,\fail} )$ for all $i \in [M]$, then the measurement map $\cM_{B\to U \Omega |S^{\Phi_{\! B}}}$ can be decomposed as
\begin{align}
\label{eq:pm-assum4} 
\cM_{B\to U\Omega|S^{\Phi_{\! B}}} = \cM_{B\to U |S^{\Phi_{\! B}} \Omega} \circ \cM_{B\to B \Omega} \,.
\end{align}
\end{lemma}

\begin{proof}
Define for each $i \in [M]$ operators $M_{B_i}^{\succ}$ and $M_{B_i}^{\fail}$ satisfying 
\begin{align}
(M_{B_i}^{\succ})^\dagger (M_{B_i}^{\succ})&=\sum_{z\in \{0,1\}} (M_{B_i}^{0,z})^\dagger (M_{B_i}^{0,z})= \sum_{z\in \{0,1\}} (M_{B_i}^{1,z})^\dagger (M_{B_i}^{1,z})\\
(M_{B_i}^{\fail})^\dagger (M_{B_i}^{\fail}) &= (M_{B_i}^{0,\fail} )^\dagger (M_{B_i}^{0,\fail} )= (M_{B_i}^{1,\fail} )^\dagger (M_{B_i}^{1,\fail} ).
\end{align}
Note in particular that $(M_{B_i}^{\succ})^\dagger (M_{B_i}^{\succ}) +  (M_{B_i}^{\fail})^\dagger (M_{B_i}^{\fail}) =\id_{B_i}$.
The generalized measurement $\cM_{B\to B \Omega}$ is then described by
\begin{align}
\cM_{B\to B \Omega}\: : \: \rho_B \mapsto \rho_{B\Omega} = \sum_{c \in \{\succ, \fail\}^M } \proj{c}_\Omega \otimes M_B^c \rho_B (M_B^c)^{\dagger},
\end{align}
with $M_B^c := \bigotimes_{i=1}^M M_{B_i}^{c_i}$.
Up to relabeling, the string $c$ contained in register $\Omega$ describes the subset $\omega$ describing the indices for which the measurement was conclusive. Indeed $\omega(c) = \{ i\in [M] \: : \: c_i \ne \fail\} = \{i \in [M] \: : \: c_i = \succ\}$. It will be convenient in the following to write $M_B^\omega$ instead of $M_B^c$, and put the value $\omega = \omega(c)$ in register $\Omega$.
Let us now define the second measurement, $\cM_{B\to U |S^{\Phi_{\! B}} \Omega}$, characterized by operators $\{M_{B_i}^{u, \phi, c}\}_{u\in \{0,1, \fail\}}$ for each $i \in [M]$. 
In the case where the register $\Omega_i$ is set to $\fail$ (or equivalently that the set $\omega$ does not contain index $i$), we choose
\begin{align}
M_{B_i}^{\fail, \phi, \fail} = \id_{B_i}, \quad M_{B_i}^{0, \phi, \fail} = 0, \quad M_{B_i}^{1, \phi, \fail} = 0
\end{align}  
which means that the measurement is essentially trivial and simply reveals the content of register $\Omega_i$. 
To address the case where the register is set to $\succ$, indicating that a conclusive outcome should be obtained, we choose operators $M_{B_i}^{0, \phi, \succ}$ and $M_{B_i}^{1, \phi, \succ}$ given by  
\begin{align}
M_{B_i}^{0, \phi, \succ} = M_{B_i}^{0, \phi} (M_{B_i}^{ \succ})^{-1/2}, \quad M_{B_i}^{1, \phi, \succ} = M_{B_i}^{1, \phi} (M_{B_i}^{ \succ})^{-1/2}.
\end{align}
It is immediate to check that these define valid generalized measurements and that Eq.~\eqref{eq:pm-assum4} is satisfied. 
\end{proof}

\begin{lemma}
\label{lm:sifting}
If the assumption of Eq.~\eqref{eq:pm-assum3} on Bob's measurement holds, then for any state $\rho_{A'BE}$, define
\begin{align}
\rho_{A''B'  C^\Sigma C^\Omega S^{\Phi} E F^{\mathrm{si}}} = \mathcal{E}_{\mathrm{di}} \circ \mathcal{E}_{\mathrm{si}} \circ \cM_{B\to B\Omega} (\rho_{A'B E} \otimes \rho_{S^{\Phi_{\! A}}} \otimes \rho_{S^{\Phi_{\! B}}}).
\end{align}
Then, the state conditioned on the sifting procedure passing satisfies:
\begin{align}
\rho_{A''B' S^\Phi C^\Sigma C^\Omega E | F^{\mathrm{si}}= \succ} = \rho_{A'B C^\Sigma C^\Omega E | F^{\mathrm{si}}= \succ} \otimes \rho_{S^{\Phi}},
\end{align}
where $\rho_{S^{\Phi}} = \frac{1}{2^m} \sum_{\phi \in \{0,1\}^m} \proj{\phi}_{\Phi }$.
\end{lemma}

\begin{proof}
Since the measurement map $\cM_{B\to B\Omega}$ only acts on register $B$, independently of the value $\Phi_{\! B}$, we have
\begin{align}
\cM_{B\to B \Omega} (\rho_{A'BE} \otimes \rho_{S^{\Phi_{\! A}}} \otimes \rho_{S^{\Phi_{\! B}}}) =
       \rho_{A'B \Omega E} \otimes \rho_{S^{\Phi_{\! A}}} \otimes \rho_{S^{\Phi_{\! B}}},
\end{align}
where the state $\rho_{A' B   E\Omega} = \cM_{B\to B\Omega} (\rho_{A' B E})$ is a classical-quantum state:
\begin{align}
\rho_{A' BE \Omega } &=\sum_{\omega \in 2^{[M]}} \proj{\omega}_\Omega \otimes  (\id_{A'} \otimes M_B^\omega \otimes \id_E)\rho_{A'BE} (\id_{A'} \otimes (M_B^\omega)^{\dagger} \otimes \id_E).
\end{align}
It is straightforward to check that the classical map `sift' has the following property: for all strings $\phi_A, \phi_B, \theta \in \{0,1\}^M$ and any subset $\omega \subseteq [M]$, if the sifting succeeds, then 
\begin{align}
\mathrm{sift}(\phi_A + \theta, \phi_B+\theta, \omega)  =\mathrm{sift}(\phi_A , \phi_B, \omega).  
\end{align}
Indeed, this is true since the map `sift' only examines whether Alice and Bob's measurement bases coincide or not, and not their actual value. In particular, if $\Phi_{\! A}$ and $\Phi_{\! B}$ are uniformly distributed, then $\Phi$, the restriction of $\Phi_{\! A}$ to the subset $\Sigma$ returned by the sifting map when it succeeds, will also be uniformly distributed over the set of strings of length $m$.

Finally, the discarding map $\mathcal{E}_{\mathrm{di}}$ examines register $S^{\Phi_{\! A}}$ and puts its content, restricted to the subset determined by the sifting map, into register $S^{\Phi}$, and traces over all the systems that do not belong to that subset. The above property of the sifting map ensures that the value of $\Phi$ does not depend on $\Omega$.

This establishes that whenever the sifting test passes, the output state takes a tensor product form: $\rho_{A''B' S^\Phi C^\Sigma C^\Omega E | F^{\mathrm{si}}= \succ} = \rho_{A''B' C^\Sigma C^\Omega E| F^{\mathrm{si}}= \succ} \otimes \rho_{S^{\Phi}}$.
\end{proof}

This lemma shows in particular that it is legitimate to consider $S^\Phi$ as a uniform seed, and not as a transcript, hence its notation $S^\Phi$ instead of $C^\Phi$.

\subsection{Security: reduction to the entanglement-based protocol}
\label{sc:pm2}

We are now ready to prove Theorem~\ref{thm:real}.

\begin{proof}
It is sufficient to consider the case where the sifting procedure succeeds, since otherwise the protocol aborts and its output is trivially secret. For this reason, let us define a slight variant $\qkdmod$ of the entanglement-based protocol of Part \ref{part:eb} which differs by taking an additional input register $F^{\mathrm{si}} \in \{ \succ, \fail\}$. The variant starts by examining the content of this registers, and either aborts if the flag is set to $\fail$, or proceeds with the protocol $\qkdeb$ if the flag is set to $\succ$. A second difference between $\qkdeb$ and $\qkdmod$ is that the randomness for the measurement basis choice is explicitly given as an input. 
In particular, for any state $\rho_{A'BE}$, it holds that:
\begin{align}
\qkdmod (\rho_{A'BE} \otimes \rho_{S^{\Phi}} \otimes \proj{\succ}_{F^{\mathrm{si}}}) =    \qkdeb(\rho_{A'BE}) \,.
\end{align}
From it definition, it is immediate that if $\qkdeb$ is $\eps$-secure, then so is $\qkdmod$. Indeed, the only quantitative difference between the two protocol is that the latter one is less robust since it will not output nontrivial keys as soon as the sifting flag is set to $\fail$.

Our goal is therefore to show that in that case, for any input channel $\cN_{A \to BE}$, there exists a state $\rho_{A''B'E'F^{\mathrm{si}}}$ where $A''$ and $B'$ consist of $m$ systems such that 
\begin{align}
\qkdpm(\cN_{A\to BE}) = \qkdmod (\rho_{A''B'E'F^{\mathrm{si}}} \otimes \rho_{S^{\Phi}}).
\end{align}
Let us therefore consider the application of the prepare-and-measure protocol $\qkdpm$ to an arbitrary quantum channel $\cN_{A\to BE}$. According to the description of the protocol, the classical-quantum state shared by Alice, Bob and Eve after the distribution step is given by some $\rho_{RBE S^{\Phi_{\! A}} S^{\Phi_{\! B}}}$.
The assumption made on Alice's preparation shows, as stated in Corollary \ref{cor:virtual}, show that there exist a state $\tau_{AA'}$ and a measurement map  $\mathcal{M}_{A'\to R|S^{\Phi_{\! A}}}$ such that
\begin{align}
\rho_{RBE S^{\Phi_{\! A}} S^{\Phi_{\! B}}} &=\cN_{A \to BE}(\rho_{RAS^{\Phi_{\! A}} S^{\Phi_{\! B}}}) \\
&= \left(\cN_{A \to BE} \circ  \mathcal{M}_{A'\to R|S^{\Phi_{\! A}}}\right) \left(\tau_{AA'} \otimes \rho_{S^{\Phi_{\! A}}} \otimes \rho_{S^{\Phi_{\! B}}} \right) \\
&=\mathcal{M}_{A'\to R|S^{\Phi_{\! A}}} \left(\rho_{A'BE} \otimes \rho_{\Phi_{\! A}} \otimes \rho_{\Phi_{\! B}}  \right),
\end{align}
where we defined $\rho_{A'BE} = \cN_{A\to BE} (\tau_{AA'})$.
The last equality follows from the fact that the maps $\cN_{A\to BE}$ and $\mathcal{M}_{A'\to R|S^{\Phi_{\! A}}}$ trivially commute since they act on distinct systems.
After applying the measurement map $\cM_{B \to B \Omega}$ promised by Lemma \ref{lm:sifting2}, followed by the sifting and \textrm{discard} maps, we obtain
\begin{align}
\rho_{R'B'ES^\Phi S^{\Phi_{\! B}} C^\Sigma  C^\Omega F^{\mathrm{si}}} &= \cE_{\mathrm{di}} \circ \cE_{\mathrm{si}} \circ \cM_{B \to B \Omega} \left(\rho_{RBE S^{\Phi_{\! A}} S^{\Phi_{\! B}}} \right) \\
&=  \cE_{\mathrm{di}} \circ \cE_{\mathrm{si}} \circ  \cM_{B \to B \Omega} \circ \mathcal{M}_{A'\to R|S^{\Phi_{\! A}}} \left(\rho_{A'BE S^{\Phi_{\! A}} S^{\Phi_{\! B}}} \right),
\end{align}
where $R'$ and $B'$ are now restricted to the $m$ indices corresponding to the set $\Sigma$ provided by the sifting map.
Indeed, recall that the \textrm{discard} map $\cE_{\mathrm{di}}$ replaces the $M$-system registers $R$ and $B$ by $m$-system registers $R'$ and $B'$ obtained by tracing over the systems not corresponding to $\Sigma$. 

Since the measurement map $\mathcal{M}_{A'\to R|S^{\Phi_{\! A}}}$ of Alice's system $A'$ commutes with $\cE_{\mathrm{si}} \circ  \cM_{B \to B \Omega}$, we obtain:
\begin{align}
\rho_{R'B'ES^\Phi S^{\Phi_{\! B}}C^\Sigma  C^\Omega F^{\mathrm{si}}} &=  \cE_{\mathrm{di}}\circ \mathcal{M}_{A'\to R|S^{\Phi_{\! A}}} \circ \cE_{\mathrm{si}} \circ  \cM_{B \to B \Omega}  \left(\rho_{A'BE S^{\Phi_{\! A}} S^{\Phi_{\! B}}} \right), \label{eqn:149}
\end{align}

In the mathematical description of the protocol given in Section \ref{sub:mod}, the \textrm{discard} map was applied to registers $R, B$, but since it commutes with measurement maps on either Alice's or Bob's system, the map can be just as well applied to registers $A'$ and $U$, with outputs denoted $A''$ and $U'$, respectively.
We deduce that we can replace the map $\cE_{\mathrm{di}}\circ \mathcal{M}_{A'\to R|S^{\Phi_{\! A}}}$ by $\mathcal{M}_{A''\to R'|S^{\Phi}} \circ \cE_{\mathrm{di}}$ in \eqref{eqn:149}.
This yields:
\begin{align}
\rho_{R'B'ES^\Phi S^{\Phi_{\! B}} C^\Sigma  C^\Omega F^{\mathrm{si}}} &= \mathcal{M}_{A''\to R'|S^{\Phi}} \circ \cE_{\mathrm{di}}\circ  \cE_{\mathrm{si}} \circ  \cM_{B \to B \Omega}  \left(\rho_{A'BE S^{\Phi_{\! A}} S^{\Phi_{\! B}}} \right)\\
&=  \mathcal{M}_{A''\to R'|S^{\Phi}} \left(\rho_{A''B'  C^\Sigma C^\Omega S^{\Phi} S^{\Phi_{\! B}} E F^{\mathrm{si}}}\right).
\end{align}
Lemma \ref{lm:sifting} now shows that
\begin{align}
\rho_{R'B'ES^\Phi S^{\Phi_{\! B}} C^\Sigma C^\Omega | F^{\mathrm{si}}=\succ} &= \mathcal{M}_{A''\to R'|S^{\Phi}} \left(\rho_{A''B'  E C^\Sigma C^\Omega |   F^{\mathrm{si}} = \succ}\right) \otimes \rho_{S^{\Phi}}.
\end{align}

Let us finally collect $E' = ES^{\Phi_{\!B}}C^{\Sigma} C^{\Omega}$. 
If the sifting test passes, then
\begin{align}
\qkdmod\left(\rho_{A''B' S^\Phi E' | F^{\mathrm{si}}= \succ} \otimes \rho_{S^\Phi} \otimes \proj{\succ}_{F^{\mathrm{si}}} \right) 
&= \qkdeb(\rho_{A''B'E'}),
\end{align}
which concludes the proof.
\end{proof}


\section{Conclusion}

We provide a self-contained security proof of QKD detailing all the steps of the protocol and explicitly spelling out all the required assumptions for the security proof to go through. 
For simplicity, we focussed on a variant of the entanglement-based BBM92 protocol as well as the BB84 protocol and showed that practical secret key rates can be achieved, even for moderately large block size. These results, however, come at the price of several assumptions which are sometimes challenging to enforce in practice. This should not come as a surprise since many simplified implementations are known to be vulnerable to quantum hacking, illustrating that there exist necessary trade-offs between ease of implementation and security guarantees.

We believe that there is room for improvement for these trade-offs and that further collaboration between theory and experiment will be essential for achieving this objective. In this context, it is crucial to model the protocols as thoroughly as possible in order to understand what level of security can be obtained, and under which assumptions. Finally, we would like to encourage more research into derving tight finite resource trade-offs for quantum key distribution protocols beyond BB84 and BBM92. This will likely require techniques beyond the entropic uncertainty relation presented here. An example of interest is the 6-state protocol~\cite{bruss98} where entropic uncertainty relations do not currently yield optimal secret key rates (even asymptotically) and approaches based on a more complete tomography of the channel lead to large penalities for finite keys.

\emph{Note added.} After completion of this work a novel and intriguing proof technique (based on R\'enyi entropy accumulation) has been proposed by \citet{dupuis16}. This technique does not yet seem to yield tradeoffs between security and protocol parameters that match those found in~\cite{tomamichellim11} and~\cite{hayashi11}, on which we improve upon here. However, the technique is more versatile and in particular allows to show security of device-independent protocols as demonstrated by~\citet{rotem16}. Device-independent protocols have the advantage that fewer assumptions on the physical devices used in the protocol are required, but are beyond the scope of this work.

\paragraph*{Acknowledgements.} We thank Philippe Grangier, Christopher Portmann and Charles Lim Ci Wen for helpful comments. MT is funded by an ARC Discovery Early Career Researcher Award (DECRA)
fellowship and acknowledges support from the ARC Centre of Excellence for
Engineered Quantum Systems (EQUS).




\appendix

\section{Proof of entropic uncertainty relation in Proposition~\ref{pr:ur}}
\label{app:ur}

Let us restate the desired inequality for the convenience of the reader.

\newtheorem*{prop5}{Proposition \ref{pr:ur} (restated)}

\begin{prop5}
  Let $\tau_{APRS} \in \cSsub(APRS)$ be an arbitrary sub-normalized state with $P$ a classical register, and set $t := \tr\{\tau_{APRS}\}$. Furthermore, let $\eps \in [0, \sqrt{t})$ and let $q$ be a bijective function on $P$ that is a symmetry of $\tau_{APRS}$ in the sense that $\tau_{ARS, P=p} = \tau_{ARS, P=q(p)}$ for all $p \in P$. Then, we have
  \begin{align}
    H_{\min}^{\eps}(X| P R)_{\sigma} + H_{\max}^{\eps}(X| P S)_{\sigma} \geq \log \frac{1}{c_q}
    , \qquad \textrm{where}  \label{eq:ucr}
  \end{align}
  where $c_q = \max_{p \in P} \max_{x,z \in X} \big\| {F_{A}^{q(p),x}} \big( {F_{A}^{p,z}} \big)^{\dag} \big\|_{\infty}^2$.
  Here, $\sigma_{XPRS} = \cM_{A\to X|P}(\tau_{A P RS})$ for the map 
  \begin{align}
     \cM_{A\to X|P}\big[\cdot\big] &=  \tr_A \Bigg\{ \sum_{p \in P} \sum_{x \in X} 
    \ket{x}_{X} 
      \Big( \proj{p}_{P} \otimes F_{A}^{p,x} \Big)\ \cdot\ \Big( \proj{p}_{P} \otimes F_{A}^{p,x} \Big)^{\dagger} \bra{x}_{X} \Bigg\} \,. \label{eq:urth2}
  \end{align}
  and any set (indexed by $p \in P$) of generalized measurements $\{ F_A^{p,x} \}_{x \in X}$.
\end{prop5}

\begin{proof}
  The condition on $\eps$ ensures that all smooth entropies are well-defined.
  We first introduce the Stinespring dilation isometry of the measurement map~$\cM_{A \to X|P}$. This is
  the isometry $V : A \to A X X'|P$ given by
  \begin{align}
    V := \sum_{p \in P} \sum_{ x \in X} 
    \ket{x}_{X} \otimes \ket{x}_{X'} \otimes \proj{p}_{P} \otimes F_{A}^{p,x}. \label{eq:v01}
  \end{align}
  Now note that the measured state $\sigma_{XPRS}$ in~\eqref{eq:urth2} has a natural purification in 
  \begin{align}
    \sigma_{P P'AXX'RSD} &= V \big( \tau_{P P'ARSD } \big) V^{\dag} \,, \quad \textrm{where} \\
    \ket{\tau_{P P'ARSD }} &= \sum_{p \in P} \sqrt{\Pr[P = p]_{\tau}}\, \ket{p}_P \otimes \ket{p}_{P'} \otimes \ket{\tau_{ARSD|P=p}} , 
    \label{eq:v02}
  \end{align}
  where $P'$ is isomorphic to $P$ and $\ket{\tau_{ARSD|P=p}}$ is any purification of $\tau_{ARS|P=p}$ on a sufficiently large auxiliary system $D$.
  (The choice $|D| = |A||R||S|$ ensures that all purifications can be accommodated.) This now allows us to rephrase our target inequality. Using the duality relation for smooth min- and max-entropy in~\eqref{eq:dual} together with the fact that $H_{\min}^{\eps}(X | P R)_{\sigma}=H_{\min}^{\eps}(X | P' R)_{\sigma}$ since $P'$ is a copy of $P$, we find that~\eqref{eq:ucr} is equivalent to 
    \begin{align}
           H_{\min}^{\eps}(X | P R)_{\sigma} \geq H_{\min}^{\eps}(X| PAX'RD)_{\sigma} + \log \frac{1}{c_q} \,.
    \end{align}
  Moreover, the data-processing inequality for the smooth min-entropy in~\eqref{eq:dpi} applied for the map $\tr_D$ yields
  $H_{\min}^{\eps}(X| PAX'RD)_{\sigma} \leq H_{\min}^{\eps}(X| PAX'R)_{\sigma}$
  and thus it in fact suffices to show that\footnote{In its cryptographic application, another intuitive way to justify that we can remove the purifying system $D$ is that, without loss of generality, we may assume that the eavesdropper already holds a purification of the state shared by the honest parties and $D$ is thus trivial. Luckily this physical intuition is corroborated by a mathematical argument\,---\,the data-processing inequality.}
  \begin{align}
    H_{\min}^{\eps}(X | P R)_{\sigma} \geq H_{\min}^{\eps}(X| PAX'R)_{\sigma} + \log \frac{1}{c_q} \,. \label{eq:url1}
  \end{align}
  The remainder of the proof will be concerned with showing the inequality in~\eqref{eq:url1}.

  For this purpose, let us consider the following unitary rotation:
  \begin{align}
    Q_{P} = \sum_{p \in P} \big| q^{-1}(p) \big\rangle\!\big\langle p\big|_{P} \,.
  \end{align}
  that exchanges $p$ with its conjugate, $q(p)$. It clearly acts as a permutation when acting on the classical register $P$ and furthermore we have $Q_{P} (\tau_{APRS}) Q_P^{\dag} = \tau_{APRS}$
   due to the symmetry condition that we imposed on $q$ and $\tau_{PARS}$ in the statement of the proposition.
  Based on this we define the isometry
  \begin{align}
      \bar{V} := Q_{P} V Q_{P}^{\dag} = \sum_{p \in P} \sum_{ x \in X} 
    \ket{x}_{X} \otimes \ket{x}_{X'} \otimes \proj{p}_{P} \otimes F_{A}^{q(p),x}, 
  \end{align}
  that corresponds to a measurement in the basis determined by $q(p)$ instead of $p$. We find
  \begin{align}
     \bar{V} {V}^{\dag} \sigma_{AXX'PRS} V \bar{V}^{\dag} 
     &= Q_{P} V Q_{P}^{\dag} ( \tau_{APRS}  )  Q_{P} V^{\dag} Q_{P}^{\dag} = Q_{P} \sigma_{AXX'PRS} Q_{P}^{\dag} \,, \label{eq:ur10}
  \end{align}
which shows that the trace non-increasing map $\bar{V} {V}^{\dag} (\cdot) {V} \bar{V}^{\dag}$ coherently undoes the measurement in the basis determined by $p$ and then instead measures in the basis determined by $q(p)$.
  
  Now we have the tools at hand to prove the inequality in~\eqref{eq:url1}. By the definition of the smooth min-entropy, $H_{\min}^{\eps}(X|P A X' R)_{\sigma}$, there exists a state $\omega_{PAX'R} \in \cS(PAX'R)$ and a sub-normalized state $\tilde{\sigma}_{PAXX'R} \in \cSsub(PAXX'R)$ that is $\eps$-close to $\sigma_{PAXX'R}$ in the sense that $P({\sigma}_{PAXX'R}, \tilde{\sigma}_{PAXX'R}) \leq \eps$ such that the following inequality holds:
  \begin{align}
    \tilde{\sigma}_{PAXX'R} \leq 2^{-\lambda}\, \id_X \otimes \omega_{PAX'R} \quad \textrm{with} \quad \lambda := H_{\min}^{\eps}(X|P A X' R)_{\sigma} \,. \label{eq:ur15}
  \end{align}
  Next we consider the CP trace non-increasing map
  \begin{align}
    \mathcal{F}_{XX'A \to X| P}[\,\cdot\,] = \sum_{p \in P} \tr_{AX'} \Big( Q_{P}^{\dag} \bar{V} {V}^{\dag} \proj{p}_{P} \,\cdot\, \proj{p}_{P} V \bar{V}^{\dag} Q_{P} \Big).
  \end{align}
  From~\eqref{eq:ur10} we learn that $\mathcal{F}\big[\sigma_{PAXX'R}\big] = \sigma_{PXR}$.
   Thus, using the fact that the purified distance contracts~\eqref{eq:pd-dpi} when we apply $\mathcal{F}$, we find that the state $\hat{\sigma}_{PXR} := \mathcal{F}\big[ \tilde{\sigma}_{PAXX'R} \big]$
  satisfies 
  \begin{align}
  P(\hat{\sigma}_{PXR}, \sigma_{PXR}) \leq P(\tilde{\sigma}_{PAXX'R}, \sigma_{PAXX'R} ) \leq \eps . \label{eq:pdbound}
  \end{align}
  Furthermore, applying $\mathcal{F}$ on both sides of~\eqref{eq:ur15} yields
  \begin{align}
    \hat{\sigma}_{PXR} \leq 2^{-\lambda} \mathcal{F}\big[\id_X \otimes \omega_{PAX'R} \big] = 2^{-\lambda} \tr_{X'A} \Big( Q_{P}^{\dag} \bar{V} {V}^{\dag} \big( \id_X \otimes \hat{\omega}_{PAX'R} \big) V \bar{V}^{\dag} Q_{P} \Big) , \label{eq:op-ineq1}
  \end{align}
  where $\hat{\omega}_{PAX'R} = \sum_{p \in P} \proj{p}_{P} \otimes \hat{\omega}_{AX'R}^{p}$ with $\hat{\omega}_{AX'R}^{p} = \bra{p} \omega_{PAX'R} \ket{p}_{P}$.

  Let us now simplify the right-hand side of this inequality, hoping to capture the incompatibility of the measurements in the basis $p$ versus the basis $q(p)$. First, we note that 
  \begin{align}
    Q_{P}^{\dag} \bar{V} {V}^{\dag} = \sum_{p \in P} \sum_{x, z \in X} |z \rangle\!\langle x|_X \otimes |z \rangle\!\langle x|_{X'} \otimes \big|q(p) \big\rangle\!\big\langle p\big|_{P} \otimes F_{A}^{q(p),z} \big( F_{A}^{p,x} \big)^{\dag} \,.
  \end{align}
  and, hence, we can write
  \begin{align}
    &\tr_{X'A} \Big( Q_{P}^{\dag} \bar{V} {V}^{\dag} \big( \id_X \otimes \hat{\omega}_{PAX'R} \big) V \bar{V}^{\dag} Q_{P} \Big) \nonumber\\
    &\quad = \sum_{p \in P} \sum_{x,z \in X} \proj{x}_{X} \otimes \proj{q(p)}_{P} \otimes \bra{z} \tr_A \Big( F_{A}^{q(p),z} \big( F_{A}^{p,x} \big)^{\dag}  \hat{\omega}_{AX'R}^{p} F_{A}^{p,x} \big( F_{A}^{q(p),z} \big)^{\dag} \Big)   \ket{z}_{X'} \\
    &\quad\leq \sum_{p \in P} \sum_{x,z \in X} \proj{x}_{X} \otimes \proj{q(p)}_{P} \otimes \Big\| F_{A}^{p,x} \big( F_{A}^{q(p),z} \big)^{\dag}  F_{A}^{q(p),z} \big( F_{A}^{p,x} \big)^{\dag} \Big\|_{\infty}\,\bra{z} \tr_A \big(  \hat{\omega}_{AX'R}^{p}   \big)   \ket{z}_{X'} \label{eq:ur22} \\
    &\quad= \sum_{p \in P} \sum_{x,z \in X} \proj{x}_{X} \otimes \proj{q(p)}_{P} \otimes \Big\| F_{A}^{q(p),x} \big( F_{A}^{p,z} \big)^{\dag} \Big\|_{\infty}^2 \,\bra{z} \hat{\omega}_{X'R}^{p} \ket{z}_{X'} \label{eq:ur23} \\
    &\quad\leq \max_{p \in P} \max_{x, z \in X} \Big\| F_{A}^{q(p),x} \big( F_{A}^{p,z} \big)^{\dag} \Big\|_{\infty}^2 \cdot \sum_{p \in P} \sum_{x,z \in X} \proj{x}_{X} \otimes \proj{p}_{P} \otimes \bra{z} \hat{\omega}_{X'R}^{p} \ket{z}_{X'} \\
    &\quad= c_q \cdot \sum_{p \in P} \id_{X} \otimes \proj{p}_{P} \otimes \hat{\omega}_R^{{p}} \,. \label{eq:ur25}
  \end{align}
  To establish~\eqref{eq:ur22} and~\eqref{eq:ur23} we used the fact that $L^{\dag} L \leq \| L^{\dag} L \|_{\infty}\, \id = \| L \|_{\infty}^2\, \id$ for every linear operator $L$ by definition of the operator norm. The final equality~\eqref{eq:ur25} follows from the definition of $c_q$.

  Combining this bound with~\eqref{eq:op-ineq1} yields
  \begin{align}
    \hat{\sigma}_{PXR} \leq 2^{-\lambda} c_q \cdot \id_X \otimes \sum_{p \in P} \proj{p}_{P} \otimes \hat{\omega}_R^{p} \,.
  \end{align}
  Since $\sum_{p \in P} \tr(\hat{\omega}_R^{p}) = 1$ by construction and $P(\hat{\sigma}_{PXR}, \sigma_{PXR}) \leq \eps$ due to~\eqref{eq:pdbound}, the definition of the smooth entropy implies that
  \begin{align}
     H_{\min}^{\eps}(X|PR)_{\sigma} \geq \lambda - \log c_q = H_{\min}^{\eps}(X|PAX'R)_{\sigma} - \log c_q \,,
  \end{align}
  concluding the proof.
\end{proof}


\section{Proof of Leftover Hashing Lemma in Proposition~\ref{pr:leftover}}
\label{app:leftover}

\newtheorem*{prop10}{Proposition \ref{pr:leftover} (restated)}

Our proof of the leftover hashing lemma is based on the following result due to Renner~\cite[Corollary~5.5.2]{renner05}:

\begin{lemma}
\label{lm:leftover}
Let $\sigma_{XD} \in \cSsub(XD)$ be a classical on $X$ and let $\cH$ be a universal$_2$ family of hash functions from $\mathcal{X} = \{0,1\}^n$ to $\mathcal{K} = \{0,1\}^\ell$. Moreover, let $\rho_{S^{H}} = \frac{1}{|\cH|} \sum_{h \in \cH} \proj{h}_{S^H}$ be fully mixed. Then,
\begin{align}
\big\|\omega_{KS^HD} - \chi_K \otimes \omega_{S^HD} \big\|_{\tr} \leq 
\frac12 \sqrt{\tr\{ \sigma_{XD} \}} \, 2^{-\frac{1}{2}\left( H_{\min}(X|D)_{\sigma} - \ell \right)}
\end{align}
where $\chi_K = \frac{1}{2^\ell} \id_{K}$ is the fully mixed state and $\omega_{KS^HD} = \tr_{X}\big( \mathcal{E}_{f}(\sigma_{XD} \otimes \rho_{S^{H}}) \big)$ for the function $f: (x, h) \mapsto h(x)$ that acts on the registers $X$ and $S^{H}$.
\end{lemma}

We provide a short proof for the convenience of the reader (see~\cite[Section~7.3.2]{mybook}).

\begin{proof}
   First, by definition of the min-entropy, there exists a state $\tau_D \in \cS(D)$ such that $\sigma_{XD} \leq 2^{-H_{\min}(X|D)_{\sigma}} \id_X \otimes \tau_D$.
   Next, note that by definition of the trace distance, we have
   \begin{align}
     \big\|\omega_{KS^HD} - \chi_K \otimes \omega_{S^HD} \big\|_{\tr} = \sum_{h \in \cH} \frac{1}{2|\cH|} \big\|\omega_{KD|S^H=h} - \chi_K \otimes \omega_{D|S^H=h} \big\|_1 \, , \label{eq:firstone}
   \end{align}
   where $\| \cdot \|_1$ denotes the Schatten 1-norm. Moreover, by construction of $\omega_{KS^HD}$ it is evident that $\omega_{D|S^H=h} = \omega_D = \sigma_D$ for all $h \in \cH$.
   Due to H\"older's inequality for Schatten norms~\cite[Corollary~IV.2.6]{bhatia97}, we have
   \begin{align}
     \big\|\omega_{KD|S^H=h} - \chi_K \otimes \sigma_D \big\|_1^2 
     &\leq \big\| \id_K \otimes \tau_D^{\frac12} \big\|_2^2 \big\| \id_K \otimes \tau_D^{-\frac12}  \big( \omega_{KD|S^H=h} - \chi_K \otimes \sigma_D \big) \big\|_2^2 \\
     &= 2^{\ell} \tr\{ \tau_D \}  \tr\Big\{ \id_K \otimes \tau_D^{-1}  \big( \omega_{KD|S^H=h} - \chi_K \otimes \sigma_D \big)^2 \Big\} \\
     &= 2^{\ell} \tr\big\{ \big( \id_K \otimes \tau_D^{-1} \big)  \omega_{KD|S^H=h}^2 \big\} - \tr\big\{ \tau_D^{-1} \sigma_D^2 \big\},
   \end{align}
   where $\tau_D$ is inverted on its support. We took advantage of the fact that $\chi_K = \frac{1}{2^\ell} \id_{K}$ is proportional to the identity to simplify the above expression. Combining this with~\eqref{eq:firstone}, Jensen's inequality thus ensures that
   \begin{align}
     \big\|\omega_{KS^HD} - \chi_K \otimes \omega_{S^HD} \big\|_{\tr} \leq \frac12 \sqrt{ \sum_{h \in \cH} \frac{2^{\ell}}{|\cH|} \tr\Big\{ \big( \id_K \otimes \tau_D^{-1} \big) \omega_{KD|S^H=h}^2 \Big\} - \tr\big\{ \tau_D^{-1} \sigma_D^2 \big\} }. \label{eq:secondone}
   \end{align}
   
   Next, observe that $\omega_{KD|S^H=h} = \sum_{k \in \{0,1\}^{\ell}} \sum_{x \in \{0,1\}^n, h(x) = k} |k\rangle\!\langle k| \otimes \sigma_{D \land X=x}$ by construction.
   We then use the universal$_2$ property of $\cH$ which implies that $\frac{1}{|\cH|} \sum_{h \in \cH} 1\{ h(x) = h(y)\} = 2^{-\ell}$ when $x \neq y$. This yields 
   \begin{align}
      \sum_{h \in \cH} \frac{1}{|\cH|} \tr\big\{ \id_K \otimes \tau_D^{-1}  \omega_{KD|S^H=h}^2 \big\} 
      &= \sum_{x,y \in \{0,1\}^n} \frac{1}{|\cH|} \sum_{h \in \cH} 1 \{ h(x) = h(y) \} \tr\big\{ \tau_D^{-1} \sigma_{D \land X = x} \sigma_{D \land X = y} \big\} \\
      &=  \sum_{x \in \{0,1\}^n} \tr\big\{ \tau_D^{-1} \sigma_{D \land X = x}^2 \big\}\, + \!\! \sum_{x,y \in \{0,1\}^n \atop x \neq y} \!\! 2^{-\ell} \tr\big\{ \tau_D^{-1} \sigma_{D \land X = x} \sigma_{D \land X = y} \big\} \\
      &= \big( 1 - 2^{-\ell} \big) \tr \big\{ \big( \id_X \otimes \tau_D^{-1} \big) \sigma_{XD}^2 \big\} + 2^{-\ell} \tr\{ \tau_D^{-1} \sigma_D^2 \} \,.
   \end{align}
   Bounding $1 - 2^{\ell} \leq 1$ and plugging this into~\eqref{eq:secondone}, we find
   \begin{align}
     \big\|\omega_{KS^HD} - \chi_K \otimes \omega_{S^HD} \big\|_{\tr} \leq \frac12 \sqrt{ 2^{\ell} \tr \big\{ \big( \id_X \otimes \tau_D^{-1} \big) \sigma_{XD}^2 \big\} }.  \label{eq:thirdone}
   \end{align}
   Finally, due to the operator anti-monotonicity of the inverse and the definition of $\tau_D$, we have $\id_X \otimes \tau_D^{-1} \leq 2^{-H_{\min}(X|D)} \sigma_{XD}^{-1}$. Combined with~\eqref{eq:thirdone} this yields the desired result.
\end{proof}

Let us restate the desired inequality for the convenience of the reader.

\begin{prop10}
Let $\sigma_{XD} \in \cSsub(XD)$ be a classical on $X$ and let $\cH$ be a universal$_2$ family of hash functions from $\mathcal{X} = \{0,1\}^n$ to $\mathcal{K} = \{0,1\}^\ell$. Moreover, let $\rho_{S^{H}} = \frac{1}{|\cH|} \sum_{h \in \cH} \proj{h}_{S^H}$ be fully mixed. Then,
\begin{align}
\big\|\omega_{KS^HD} - \chi_K \otimes \omega_{S^HD} \big\|_{\tr} \leq \frac12 2^{-\frac{1}{2}\left( H_{\min}^{\eps}(X|D)_{\sigma} - \ell \right)} + 2 \eps
\end{align}
where $\chi_K$ and $\omega_{KS^HD}$ are define as in Lemma~\ref{lm:leftover}.
\end{prop10}%

\begin{proof}
  Let $\tilde{\sigma}_{XD} \in \cSsub(XD)$ be a sub-normalized state such that $H_{\min}^{\eps}(X|D)_{\sigma} = H_{\min}(X|D)_{\tilde{\sigma}}$ and $P(\tilde{\sigma}_{XD}, \sigma_{XD}) \leq \eps$. Without loss of generality we can assume that $\tilde{\sigma}_{XD}$ is classical on $X$.
  Now, the Lemma~\ref{lm:leftover} yields the inequality
  \begin{align}
    \|\tilde{\omega}_{KS^HD} - \chi_K \otimes \tilde{\omega}_{S^HD}\|_{\tr} \leq \frac12 \sqrt{\tr\{\tilde{\sigma}_{XD}\}} \cdot 2^{-\frac{1}{2}\left( H_{\min}(X|D)_{\tilde{\sigma}} - \ell \right)} \leq \frac12 \cdot 2^{-\frac{1}{2}\left( H_{\min}^{\eps}(X|D)_{\sigma} - \ell \right)}\, ,
  \end{align}
  where we constructed $\tilde{\omega}_{KS^HD} = \tr_{X}\big( \mathcal{E}_{f}(\tilde{\sigma}_{XD} \otimes \rho_{S^{H}}) \big)$ and bounded $\tr\{\tilde{\sigma}_{XD}\} \leq 1$ to arrive at the second inequality. Using the monotonicity of the purified distance under CPTP maps we conclude that $P(\tilde{\omega}_{S^HD}, \omega_{S^HD} ) \leq P(\tilde{\omega}_{KS^HD}, \omega_{KS^HD} ) \leq \eps$.
  Finally, exploiting the triangle inequality for the trace norm we find
  \begin{align}
    \|{\omega}_{KS^HD} - \chi_K \otimes {\omega}_{S^HD}\|_{\tr} \leq 2 \eps + \|\tilde{\omega}_{KS^HD} - \chi_K \otimes \tilde{\omega}_{S^HD}\|_{\tr} \,.
  \end{align}
\end{proof}

\end{document}